\documentclass[12pt]{article}
\usepackage[margin = 1.1 in]{geometry}
\usepackage{amssymb}
\usepackage{lscape}
\usepackage{amsfonts}
\usepackage[ruled,vlined]{algorithm2e}
\usepackage{amsmath}
\usepackage{mathtools}
\usepackage{mathrsfs}
\usepackage{lineno}
\usepackage[numbers]{natbib}
\usepackage[colorlinks=true,
            linkcolor=blue,
            urlcolor=blue,
            citecolor=blue]{hyperref}
\usepackage[nameinlink]{cleveref}
\usepackage[doublespacing]{setspace}
\usepackage{bm, bbm}
\usepackage{pgf, tikz}
\usepackage{accents}

\usetikzlibrary{arrows,shapes.arrows,shapes.geometric,shapes.multipart,fit,automata,decorations.pathmorphing,positioning}
\tikzset{
>=stealth',
true/.style={
rectangle,
draw=black, very thick,
text width=6.5em,
minimum height=2em,
text centered,
fill=gray, opacity = 0.5},
punkt/.style={
rectangle,
rounded corners,
draw=black, very thick,
text width=6.5em,
minimum height=2em,
text centered},
est/.style={
circle,
draw=black, very thick,
text centered},
shade/.style={
circle,
draw=black, very thick, fill=gray!50,
text centered},
weight/.style={
circle,
draw=black, very thick,
text width=6.5em,
minimum height=2em,
text centered},
pil/.style={
->,
thick,
shorten <=2pt,
shorten >=2pt,},
double/.style={
<->,
thick,
shorten <=2pt,
shorten >=2pt,},
dash/.style={
dashed,
thick,
shorten <=2pt,
shorten >=2pt,},
dashdouble/.style={
<->,
dashed,
thick,
shorten <=2pt,
shorten >=2pt,}
}
\def\BU{\mathbb{U}}

\def\E{\mathsf{E}}

\def\IF{\mathsf{IF}}
\def\EIF{\mathsf{EIF}}
\def\BP{\mathbb{P}}
\def\cov{\mathsf{cov}}
\def\var{\mathsf{var}}
\def\Pr{\mathsf{Pr}}

\def\BR{\mathbb{R}}

\def\SNMM{\mathsf{SNMM}}
\def\hat{\widehat}
\def\tilde{\widetilde}

\newtheorem{theorem}{Theorem}

\newtheorem{corollary}{Corollary}

\newtheorem{lemma}{Lemma}

\newtheorem{remark}{Remark}

\newenvironment{proof}[1][Proof]{\noindent\textbf{#1.} }{\ \rule{0.5em}{0.5em}}
\begin{document}

\title{\textbf{Efficient estimation of optimal regimes under a no direct effect assumption\thanks{This work was motivated by a talk given by James M. Robins at ENAR 2015 entitled ``How to increase efficiency of estimation when a test used to decide treatment has no direct effect on the outcome'' in a session organized by Michael R. Kosorok.}}}
\author{Lin Liu \hspace{.2cm} \\
Institute of Natural Sciences, School of Mathematical Sciences, MOE-LSC \\
and SJTU-Yale Center for Biostatistics and Data Science, \\
Shanghai Jiao Tong University, \\
Zach Shahn \\
IBM Research, \\
James M. Robins \\
Department of Biostatistics and Epidemiology, Harvard University, \\
and \\
Andrea Rotnitzky \\
Department of Economics, Universidad Torcuato Di Tella and CONICET}
\maketitle

\newpage
\begin{abstract}
We derive new estimators of an optimal joint testing and treatment regime under the no direct effect assumption that a given laboratory, diagnostic, or screening test has no effect on a patient's clinical outcomes except through the effect of the test results on the choice of treatment. We model the optimal joint strategy with an optimal structural nested mean model (opt-SNMM). The proposed estimators are more efficient than previous estimators of the parameters of an opt-SNMM because they efficiently leverage the `no direct effect of testing' assumption. Our methods will be of importance to decision scientists who either perform cost-benefit analyses or are tasked with the estimation of the `value of information' supplied by an expensive diagnostic test (such as an MRI to screen for lung cancer).
\end{abstract}

\noindent\textit{Keywords:} opt-SNMM, optimal dynamic regimes, screening, nuisance tangent space, direct effects

\newpage

\doublespacing

\renewcommand{\baselinestretch}{1}

\if11 \fi

\if01

\begin{center}
\textbf{Efficient estimation of optimal regimes under a no direct effect assumption}
\end{center}

\medskip \fi
\renewcommand{\baselinestretch}{2}

\section{Introduction}
\label{sec:background}

This paper provides new, highly efficient estimators of an optimal joint testing and treatment regime under the no direct effect assumption that a given laboratory, diagnostic, or screening test has no effect on a patient's clinical outcome of interest except through the effect of the test results on the choice of treatment. The proposed estimators attain high efficiency because they leverage this `no direct effect of testing' (abbreviated as NDE) assumption. What is surprising and, indeed, unprecedented in the experience of the authors, is that, in actual substantive studies, estimators that leverage the NDE assumption can provide a 50-fold increase in efficiency (and, thus, a 50 fold reduction in sample size) compared to estimators that fail to leverage the NDE assumption \citep{caniglia2019emulating}.

The following case study of a randomized trial in HIV infected subjects in Africa motivates the issues with which we are concerned. First some background. HIV infected subjects who start on first line antiretroviral therapy (ART) and who later develop evidence of disease progression, quantified by an increase in viral load, a decrease in CD4 immune cell count in blood, or new clinical symptoms must consider switching to second line therapy. Therefore, at each clinic visit, a decision based on the previous laboratory and clinical data must be made as to (1) whether to order viral load and/or CD4 count tests at some cost and burden to the patient and (2) whether to switch treatment. Waiting too long to switch may result in further clinical deterioration and eventually AIDS and/or death. On the other hand, switching too early is unwise as there are only a limited number of alternative regimes and these are associated with increased side effects and expense. Thus, there is a need to determine the optimal testing and treatment regime based on empirical analysis of observational or randomized trial data. Tests of viral load and CD4 count in themselves have no direct biological effect on a patient. Rather the test results are used to help gauge the likely benefit of switching treatment. Hence appropriately timed tests may result in an increase in the expected utility, even if the financial and other costs of the test are considered.

The issue of how to determine when to switch to second line therapy is most pressing in Africa, where there has been a longstanding concern that first line ART is prescribed with limited, often no, laboratory monitoring of disease progression. Therefore, in 2003, HIV infected individuals were recruited into the Development of AntiRetroviral Therapy in Africa (DART) randomized trial in which individuals on first line ART in Uganda and Zimbabwe were randomized to laboratory and clinical monitoring (LCM) versus clinically driven monitoring (CDM). In the LCM arm, CD4 cell counts were measured every 12 weeks and made available to the caregivers. In the CDM arm, CD4 counts were not available to the caregivers. LCM subjects were switched to second line therapy when either the CD4 count fell below 100 or a world health organization (WHO) stage 4 clinical event occurred; CDM arm subjects were switched at the occurrence of a WHO stage 4 event. The trial results were published by the DART Team in 2010 \citep{dart2010routine}. An intention to treat analysis showed an adverse effect of CDM compared to LCM both for overall mortality (hazard ratio 1.35, 95\% CI [1.10,1.65]) and for the primary outcome, time to (the minimum of) a WHO stage 4 event and death (hazard ratio 1.31, 95\% CI [1.14, 1.51]).

In 2012, the DART Team published a cost-effectiveness analysis \citep{lara2012cost}. Owing to both the cost of CD4 count testing every 12th week and the higher price of second line therapy, their analysis estimated an incremental cost of at least 3,000 US dollars for each additional quality-adjusted life year (QALY) gained from following the LCM strategy compared to the CDM strategy. By its definition, a cost effectiveness analysis compares the ratio of incremental dollar costs to the health benefit in QALYs for a number of different health interventions. The philosophy underlying this definition is that those interventions, ranked from the lowest to highest ratio, would be implemented in that order until available funds are exhausted. (This is a simplification, overlooking, for instance, that factors (individual and social) other than QALYs may be used to quantify the benefits of an intervention.) An alternative approach is to combine the costs and benefits into a single utility function whose expected value we wish to maximize. An analysis whose goal is to maximize a single utility function is generally referred to as a cost-benefit, rather than a cost-effectiveness, analysis. A cost-benefit analysis requires that we place a monetary value on an additional year of quality-adjusted life, as discussed next.

The DART Team noted that a WHO publication \citep{world2001report} suggested that, to be cost effective, the incremental cost of a proposed intervention per QALY added should be less than three times the per capita GDP. If one takes the WHO literally (which may not be what WHO meant for one to do), it allows one to place an upper bound on the monetary value of a year of quality-adjusted human life. For example, the per capita GDP of Zimbabwe and Uganda were both less than 480 US dollars in 2008 (presumably the last year available to the authors). The DART Team thus concluded that the LCM strategy was not cost effective compared to the CDM strategy as $3 \times 480 =$ 1,440 is less than incremental cost of 3,000 dollars. That is, they found that the monetized health value of the information (VoI) concerning disease progression obtained from CD4 testing every 12th week was less than the cost of obtaining the information. Thus, in terms of the definitions above, the DART Team had performed a cost-benefit analysis rather than simply a cost-effectiveness analysis. In contrast, the 2008 GDP/per capita of the US was 48,000 dollars. Hence, the WHO criteria, taken literally, would suggest that the monetary value of a QALY for a US citizen was 144,000 dollars, 100 times that of a citizen of Uganda or Zimbabwe. We leave further discussion of the obvious ethical and economic issues raised by monetizing a human life to others. In this paper, we will assume the cost of testing and the health benefit are combined into a single utility function. A natural question that arises is whether the above WHO cost-effectiveness criterion would be met if, instead of testing every 12 weeks as in the trial, testing occurred every 24 weeks or 36 weeks or 48 weeks with switching to second line therapy when the CD4 count first fell below 50,100, 150, or 200. None of these $3 \times 4 = 12$ testing and treatment regimes were studied in the RCT so any attempt to answer such questions requires one analyze the trial data as if it were an observational study using methods related to those proposed herein \citep{ford2015impact}.

Finding an optimal joint testing and treatment regime for HIV infected subjects is but one of many contexts in which our methodology should be applicable. Many laboratory, diagnostic and screening tests have no effect on the clinical outcome of interest, except through the effect of the test results on the choice of treatment. An area in which our new, more efficient estimators should be particularly important is that of cost-benefit analyses wherein the costs of expensive tests (such as MRIs to screen for lung cancer, mammograms to screen for breast cancer, and urinary cytology to screen for bladder cancer) are weighed against the clinical VoI supplied by the test results (e.g. \citet{mushlin1992screening, krahn1994screening}).

There is a large literature in economics and decision science on the VoI. For the purposes of this paper, we define the VoI to be the increase in the expected utility resulting from incorporating costly information without a direct causal effect on the outcomes of interest (such as a screening test) into an optimal regime (e.g. \citet{lavalle1968cashi, lavalle1968cashii, gould1974risk, hilton1981determinants}). That is, the VoI is the difference in expected utilities of two regimes: the optimal testing and treatment regime versus the optimal treatment regime under the constraint that no testing is allowed. Our methodology also allows for more efficient estimation of the VoI.

In fact, the VoI can be directly calculated from the parameters of an optimal regime Structural Nested Mean Models (opt-SNMM). \citet{robins2004optimal}, building on \citet{murphy2003optimal}, introduced the opt-SNMM, a semiparametric model for estimating the optimal testing and treatment regime from data. Under the model, the optimal testing and treatment regime is a deterministic function of the model parameters $\Psi$.

\citet{robins2004optimal} proposed g-estimation, a semiparametric version of dynamic programming (DP), to estimate the parameters $\Psi^{\ast}$ of an opt-SNMM. Under standard assumptions required for the identification of causal effects in longitudinal settings (consistency, positivity, and sequential exchangeability), g-estimation of a correctly specified opt-SNMM yields the regular, asymptotically linear (RAL) estimator $\tilde{\Psi}$ of $\Psi^{\ast}$ and thus of the optimal joint testing and treatment regime and its value (i.e. expected utility), provided the true law generating the data is not an exceptional law as defined in \citet[page 219]{robins2004optimal} and the bias of $\tilde{\Psi}$ is of order $o_{p} (n^{-1/2})$ with $n$ the sample size. In this paper we exclude such exceptional laws for reasons given later in this section. We discuss conditions required for the bias to be $o_{p} (n^{-1/2})$ in Appendix \ref{app:drml}. An estimator $\tilde{\Psi}$ is RAL if (i) $\tilde{\Psi} - \Psi^{\ast}$ is the sum of i.i.d. mean zero random variables (referred to as the influence function of $\tilde{\Psi})$ plus a term of order $o_{p} (n^{-1/2})$ and (ii) is locally asymptotically unbiased.

In this paper we show that under the NDE assumption it is possible to construct RAL estimators that are more efficient than the g-estimators of \citet{robins2004optimal} that do not use the NDE assumption. Our construction is not straightforward because imposing the NDE assumption does not restrict the values of any of the parameters of our opt-SNMM, including those parameters that determine the optimal testing regime. For a more comprehensive review of SNMMs, we refer interested readers to the articles by \citet{robins2000marginal, robins2004optimal} or a more recent piece by \citet{vansteelandt2014structural}.

We now briefly describe our estimator construction. Full details are given in Sections \ref{sec:estimator} to \ref{sec:feasible}. \citet{robins2004optimal} showed that, without imposing the NDE assumption, the g-estimator $\tilde{\Psi}$ was equal to the solution of estimating equations $0 = \hat{\BU} (\Psi)$ where, at the true $\Psi^{\ast}$, $\tilde{\Psi}$ and $\hat{\BU} (\Psi^{\ast})$ were RAL estimators of $\Psi^{\ast}$ and $0$ respectively and, in addition, were doubly robust in the sense of \citet{bang2005doubly}; see Section \ref{sec:estimator} for further discussion. The new estimators $\tilde{\Psi} (b)$ in this paper solve $0 = \hat{\BU} (\Psi, b)$ where $\hat{\BU} (\Psi, b)$ is the residual from the orthogonal projection of the (influence function of) $\hat{\BU} (\Psi)$ into a given linear space of random variables $T_{b}$ indexed by a vector function $b$. In the absence of confounding by unmeasured factors, all components of $T_{b}$ have mean zero under the NDE assumption. To construct $T_{b}$, we first show that the NDE assumption implies the following restriction on the distribution of the observed data. Given past history, the screening variable at any time $t$ is independent of the health outcome of interest after reweighting by the inverse of the conditional probability of a subject's treatment history from $t + 1$ onwards. We then use the independence of screening and the health outcome in this reweighted distribution to construct the vector $T_{b}$. The Pythagorean Theorem then guarantees that the asymptotic relative efficiency (ARE) $\lim_{n \rightarrow \infty} \var (\tilde{\Psi}) / \var (\tilde{\Psi} (b))$ of the old relative to the new estimator is always greater than or equal to 1. Furthermore the new estimator, like the old, is doubly robust.

In the paper, we assume a semiparametric model on the joint distribution of the factual and counterfactual variables defined by the restrictions that (i) the NDE assumption is true, (ii) confounding by unmeasured variables is absent, and (iii) the opt-SNMM holds. The regular estimator $\tilde{\Psi} (b_{opt})$ with minimum asymptotic variance (and hence the notation $b_{opt}$) among all RAL estimators under the model solves $0 = \hat{\BU} (\Psi, b_{opt})$ where $\hat{\BU} (\Psi, b_{opt})$ is the residual from the projection $T_{b_{opt}}$ (indexed by $b_{opt})$ of (the influence function of) $\hat{\BU} (\Psi)$ onto the space of all random variables $T_{b}$ with mean zero under and only under the NDE assumption. We show in Section \ref{sec:optimal_dr} that, when the utility is a continuous random variable, this projection and thus $\tilde{\Psi} (b_{opt})$ depend on the observed data distribution through the solutions to a set of complicated integral equations that do not exist in closed form. In contrast, when the utility is discrete, we obtain a closed form expression for the projection. Furthermore, in the case of a continuous utility (see Remark \ref{rem:sub}), we propose a \textit{computationally tractable} estimator $\tilde{\Psi} (b_{sub})$ solving $0 = \hat{\BU} (\Psi, b_{sub})$ with relative efficiency that can be made arbitrarily close to that of the `optimal' yet \textit{computationally intractable} estimator $\tilde{\Psi} (b_{opt})$. Specifically $\hat{\BU} (\Psi, b_{sub})$ is the residual from the closed-form projection $T_{b_{sub}}$ (indexed by $b_{sub}$) of (the influence function of) $\hat{\BU} (\Psi)$ onto a large, but strict, subspace of the space of random variables with mean zero under the NDE assumption. As the dimension of the chosen subspace increases, the relative efficiency of $\tilde{\Psi} (b_{sub})$ approaches that of $\tilde{\Psi} (b_{opt})$; hence if we allow the dimension to converge to infinity with the sample size at a sufficiently slow rate, $\tilde{\Psi} (b_{sub})$ and $\tilde{\Psi} (b_{opt})$ will be asymptotically equivalent. Because the projection needed to compute $\tilde{\Psi} (b_{sub})$ exists in closed-form, it is much easier to compute than $\tilde{\Psi} (b_{opt})$ and is therefore the estimator we recommend.

We shall study the relative efficiency of the estimators $\tilde{\Psi} (b_{sub})$ and $\tilde{\Psi}$ that, respectively, do and do not use the NDE assumption, as estimators of the parameters $\Psi^{\ast}$ of an opt-SNMM. However, our approach is quite generic in the following sense. Consider a semiparametric model for the joint distribution of the factual and counterfactual variables defined by the restrictions that (i) the NDE assumption is true, (ii) confounding by unmeasured variables is absent, and (iii) a given semiparametric model holds with $\Psi^{\ast}$ the true value of the Euclidean parameter. As one important example other than an opt-SNMM, the given model could be a dynamic marginal structural model (dyn-MSM) for a pre-specified class of testing and treatment regimes \citep{hernan2005discussion, van2007causal, robins2008estimation, orellana2010dynamici, orellana2010dynamicii}. Given a doubly robust RAL estimator $\tilde{\Psi}$ of $\Psi^{\ast}$ that solves $0 = \hat{\BU} (\Psi)$ (with $\hat{\BU} (\Psi)$ an unbiased estimating equation in the semiparametric model in the absence of the NDE assumption), the methods in this paper can be used to construct a RAL doubly robust estimator $\tilde{\Psi} (b_{sub})$ with improved efficiency compared to $\tilde{\Psi}$ by solving $0 = \hat{\BU} (\Psi, b_{sub})$ where again $\hat{\BU} (\Psi, b_{sub})$ is the residual from the closed-form projection (indexed by $b_{sub})$ of (the influence function of) $\hat{\BU} (\Psi)$ onto the exact same subspace as earlier. The only difference is that $\hat{\BU} (\Psi)$ and its influence function differ depending on the chosen semiparametric model (e.g. opt-SNMM vs dyn-MSM) for $\Psi^{\ast}$. Since our procedure is a generic procedure applied to an initial RAL estimator $\tilde{\Psi}$ and we are simply using an opt-SNMM as a particular example, we decided, as mentioned above, to exclude exceptional laws because the opt-SNMM g-estimator $\tilde{\Psi}$ is not a RAL estimator under an exceptional law.

The methodology discussed above constructs highly efficient estimators under the NDE assumption of no direct effect of testing on the clinical outcome $Y^{d}$ of interest. However even more efficient estimators can be constructed if we can impose the stronger NDE assumption of no direct effect of testing not only on $Y^{d}$ but also on measured time dependent covariates, as the stronger NDE assumption implies additional mean zero random variables on which to project. See Section \ref{sec:nde_variety} for more discussion.

The above development offers the reader little guidance or intuition as to when to expect small versus enormous gains in efficiency compared to estimators that do not incorporate the NDE assumption. To provide intuition, in Section \ref{sec:intuition}, we discuss an inverse probability weighting (IPW) estimator $\tilde{\Psi}_{nde\mbox{-}ipw}$ introduced in \citet{robins2008estimation} and used by \citet{caniglia2019emulating} in a substantive paper that uses a dyn-MSM to determine the optimal testing and treatment strategy [within a class of (CD4 cell and HIV RNA) testing and (anti-retroviral) treatment strategies] to prolong the survival of HIV infected individuals included in the Harvard HIV-CAUSAL Collaboration and the Centers for AIDS Research Network of Integrated Clinical Systems. The estimator $\tilde{\Psi}_{nde\mbox{-}ipw}$ has also been studied by \citet{neugebauer2017identification} and \citet{kreif2020exploiting}. $\tilde{\Psi}_{nde\mbox{-}ipw}$ is an asymptotically linear, though generally inefficient, estimator under the stronger NDE assumption; however, it is inconsistent when this assumption is false. In the analysis of \citet{caniglia2019emulating}, $\tilde{\Psi}_{nde\mbox{-}ipw}$ is unprecedentedly 50 times as efficient as the usual IPW estimator $\tilde{\Psi}_{ipw}$ that does not exploit an NDE assumption. In fact, $\tilde{\Psi}_{nde\mbox{-}ipw}$ would be nearly 50 times as efficient as any estimator that remains RAL when the NDE assumption does not hold.

Even so, we show in Section \ref{sec:formal} that $\tilde{\Psi}_{nde\mbox{-}ipw}$ will still be less efficient than $\tilde{\Psi}_{ipw} (b_{opt, ipw})$. Here $\tilde{\Psi}_{ipw} (b_{opt, ipw})$ solves $0 = \hat{\BU}_{IPW} (\Psi, b_{opt, ipw})$ where $\hat{\BU}_{IPW} (\Psi, b_{opt, ipw})$ is the residual from the projection (indexed by $b_{opt, ipw}$) of the estimated influence function of $\tilde{\Psi}_{ipw}$ onto the space of all random variables with mean zero only under the stronger NDE assumption! Indeed, we show that $\tilde{\Psi}_{ipw} (b_{opt, ipw})$ is semiparametric efficient in the model characterized solely by the stronger NDE assumption if the dyn-MSM is a saturated model.

In Section \ref{sec:simulation}, we design a simple data generating process that makes it transparent why $\tilde{\Psi}_{nde\mbox{-}ipw}$ and $\tilde{\Psi}_{ipw} (b_{opt, ipw})$ are so much more efficient than $\tilde{\Psi}_{ipw}$ in the context of \citet{caniglia2019emulating}. 
The above may lead one to suspect that $\tilde{\Psi}_{nde\mbox{-}ipw}$ is always more efficient than the usual IPW estimator $\tilde{\Psi}_{ipw}$, whenever the NDE assumption is imposed. However this is far from the case. In fact, we design a second data generating process (see Appendix \ref{app:worse}) under which, as a particular parameter of the data generating process converges to 0, the ratio of the asymptotic variance of $\tilde{\Psi}_{ipw}$ compared to $\tilde{\Psi}_{nde\mbox{-}ipw}$ tends to 0 and simultaneously the asymptotic variance of $\tilde{\Psi}_{ipw}$ converges to that of the efficient estimator $\tilde{\Psi}_{ipw} (b_{opt, ipw})$.

The organization of the paper is as follows. In Section \ref{sec:notation}, we establish some notation and review the counterfactual framework. In Section \ref{sec:estimator}, we review opt-SNMMs and g-estimation without imposing the NDE assumption. In Section \ref{sec:nde}, we characterize the ortho-complement of the tangent spaces under the NDE assumption as a key step toward constructing more efficient estimators. In Section \ref{sec:dr}, we describe several strategies for obtaining more efficient estimators by leveraging the NDE assumption. In particular we describe the aforementioned computationally tractable procedure used in construction of $\tilde{\Psi} (b_{sub})$ in Section \ref{sec:optimal_dr}. In Section \ref{sec:feasible}, we study the statistical properties of the proposed estimators when the nuisance parameters/functions are estimated from data possibly using machine learning methodology. In Section \ref{sec:nde_variety}, we introduce increasingly strong versions of the NDE assumption and compare their substantive plausibility. In Section \ref{sec:intuition}, we provide the intuition behind our large efficiency gains through a study of the properties of $\tilde{\Psi}_{nde\mbox{-}ipw}$. We also formally compare the statistical properties of all the various estimators treated in the paper. In Section \ref{sec:simulation} and Appendix \ref{sim:further}, we illustrate the method on simulated examples. In Section \ref{sec:discussion}, we conclude with some open problems and future directions.

\section{Notation, Framework, and Background}
\label{sec:notation}

We begin by providing the notation that will be used throughout the paper. Let:

\begin{itemize}
\item $t \in \{0, \ldots, K\}$ index time or visits, assumed discrete, with $K$ being the last occasion;

\item $A_{t} \in \{0, 1\}$ denote whether a screening test is performed at time $t$;

\item $R_{t}$ denote the results of the test (e.g. MRI screening result) if performed at $t - 1$ and $R_{t} = ?$ otherwise;

\item $S_{t}$ be a binary $\{0, 1\}$ variable denoting the treatment (e.g. switching therapy) at time $t$;

\item $L_{t}$ denote covariates at time $t$ that may influence testing and treatment decisions;

\item $Y^{d}$ denote the observed value of the health outcome utility;

\item $Y \equiv Y^{d} - \sum_{t = 0}^{K} c^{\ast} A_{t}$ denote the total utility of interest with $c^{\ast }$ being the known (utility) cost of the test at time $t$; and

\item $\bar{X}_{t}$ denote $(X_{0}, \ldots, X_{t})$ and $\underaccent{\bar}{X}_{t}$ denote $(X_{t}, \ldots, X_{K})$ for an arbitrary vector $X \equiv (X_{0}, \ldots, X_{K})$.
\end{itemize}

We assume that we observe $N$ i.i.d. realizations of the random vector: 
\begin{equation*}
O \equiv (L_{0}, R_{0}, S_{0}, A_{0}, \ldots, L_{K}, R_{K}, S_{K}, A_{K}, Y^{d}).
\end{equation*}

We use capital letters to denote random variables and corresponding lower case letters to denote specific values that random variables might take. We consider the scenario where at each time $t$ the chronological ordering of the variables is $L_{t}$ before $R_{t}$ before $S_{t}$ before $A_{t}$. Our definition of $R_{t}$ implies that the results of tests at time $t$ are not available until time $t + 1$. (By convention, we take $A_{K}$ to always be 0 as testing results at $K$ would not be available until after $Y$ is measured at $K + 1$. As with $A_{K}$, we also take $A_{K, g}$ defined below to be 0. Further we set $R_{0} \equiv 0$ as there is no test before $R_{0}$ in our setting.) If the results of tests at time $t$ were available immediately and hence could influence $S_{t}$, we would simply redefine $R_{t}$ to be the results of testing at $t$ rather than at $t - 1$ and reorder as $(L_{t}, A_{t}, R_{t}, S_{t})$. Furthermore we let $\bar{H}_{m} \equiv (\bar{L}_{m}, \bar{R}_{m}, \bar{S}_{m - 1}, \bar{A}_{m - 1})$ be the past history through time $m$, excluding $S_{m}, A_{m}$. We let $\bar{\mathcal{H}}_{m}$ be the sample space of the random vector $\bar{H}_{m}$.

In Remark \ref{rem:cost} of Section \ref{sec:formal}, we provide conditions under which the cost of testing $c^{\ast}$ at $t$ can be made a function $c (t, \bar{H}_{t}, S_{t})$ of the past information rather than a fixed constant.

A deterministic testing and treatment regime $g \equiv (g_{0}, g_{1}, \ldots, g_{K})$ is a vector of rules or functions $g_{t}: (\bar{L}_{t}, \bar{R}_{t}, \bar{S}_{t - 1}, \bar{A}_{t - 1}) \mapsto (s_{t}, a_{t})$ that determines the values $(s_{t}, a_{t})$ to which $S_{t}$ and $A_{t}$ will be set given the past $(\bar{L}_{t}, \bar{R}_{t}, \bar{S}_{t - 1}, \bar{A}_{t - 1})$. We denote arbitrary regimes by $g$ and we adopt the counterfactual framework of \citet{robins1986new, robins1987addendum} in which $Y_{g}$, $Y_{g}^{d}$, $L_{t, g}$, $R_{t, g}$, $S_{t, g}$ and $A_{t, g}$ are random variables representing the counterfactual data had regime $g$ been followed. Implicit in the notation is the assumption that the treatment regime followed by one patient does not influence the outcome of any other patient. A testing and treatment regime is said to be \textit{static} if $g_{t}$ is a constant function for all $t$. A regime is said to be \textit{dynamic} if it stipulates that testing and/or treatment at time $t$ depends on past covariate and/or treatment values. A random testing and treatment regime replaces the functions $g_t$ by conditional densities taking values in the support of $(S_t, A_t)$. Since the optimal testing and treatment strategy can always be taken to be a deterministic strategy, we do not further consider random regimes in this paper. We make three additional standard assumptions that serve to identify the optimal testing and treatment regime, even when the NDE assumption fails to hold \citep{robins2004optimal}. Throughout, let 
\begin{equation*}
\Pi_{m} \equiv \pi_{m} (\bar{H}_{m}, S_{m}) \equiv \Pr [A_{m} = 1 | \bar{H}_{m}, S_{m}] \text{ and } p (\cdot | \bar{H}_{m}) \equiv \Pr [S_{m} = \cdot | \bar{H}_{m}].
\end{equation*}

\begin{enumerate}
\item \label{id1} \textbf{Positivity: } for all $m = 0, \ldots, K$, $\Pi_{m}$ and $p (s_{m} | \bar{H}_{m})$ for all $s_{m}$ in the sample space of $S_{m}$, are bounded away from 0 and 1 on a set of probability 1;

\item \label{id2} \textbf{Consistency: } $Y^{d} = Y_{g}^{d}$ and $Y = Y_{g}$ if $(\bar{A}_{K, g}, \bar{S}_{K, g}) = (\bar{A}_{K}, \bar{S}_{K})$, $(\bar{L}_{t + 1}, \bar{R}_{t + 1}, \bar{A}_{t + 1}, \bar{S}_{t + 1}) = (\bar{L}_{t + 1, g}, \bar{R}_{t + 1, g}, \bar{A}_{t + 1, g}, \bar{S}_{t + 1, g})$ if $(\bar{A}_{t, g}, \bar{S}_{t, g}) = (\bar{A}_{t}, \bar{S}_{t})$;

\item \label{id3} \textbf{Sequential exchangeability: } 
\begin{equation*}
(Y_{g}^{d}, Y_{g}, \underaccent{\bar}{L}_{t + 1, g}, \underaccent{\bar}{R}_{t + 1, g}, \underaccent{\bar}{A}_{t, g}, \underaccent{\bar}{S}_{t, g}) \amalg (A_{t}, S_{t}) | \bar{L}_{t}, \bar{R}_{t}, (\bar{A}_{t - 1}, \bar{S}_{t - 1}) = \bar{g}_{t - 1} (\bar{H}_{t - 1}) \forall t, g
\end{equation*}
where $\amalg$ stands for statistical independence and $\{ (\bar{A}_{t - 1}, \bar{S}_{t - 1}) = \bar{g}_{t - 1} (\bar{H}_{t - 1})\} \equiv \{ (A_{m}, S_{m}) = g_{m} (\bar{H}_{m}), m = 0, \ldots, t - 1 \}$.
\end{enumerate}

\begin{remark}
\label{rem:missing} 
The sample space of $R_{t + 1}$ includes ``?" and $R_{t + 1}$ is always observed. When $A_{t} = 0$ (no screening at $t$), $R_{t + 1}$ is not missing; rather $R_{t + 1} = ?$. Note $R_{t + 1} = ?$ if and only if $A_{t} = 0$. In Section \ref{sec:nde_variety} we introduce a new underlying variable $R_{t + 1}^{\ast}$ that is equal to $R_{t + 1}$ when $A_{t} = 1$ and is missing when $A_{t} = 0$. Until then the variable $R_{t + 1}^{\ast}$ is not needed as both treatment and screening decisions at time $t + 1$ will depend only on always observed available information such as $R_{t + 1}$.
\end{remark}

Note that we could replace the vector $(Y_{g}^{d}, Y_{g}, \underaccent{\bar}{L}_{t + 1, g}, \underaccent{\bar}{R}_{t + 1, g}, \underaccent{\bar}{A}_{t + 1, g}, \underaccent{\bar}{S}_{t + 1, g})$ in assumption \ref{id3} by $(Y_{g}^{d}, \underaccent{\bar}{L}_{t + 1, g}, \underaccent{\bar}{R}_{t + 1, g})$ because conditional on $(\bar{A}_{t - 1}, \bar{S}_{t - 1}) = \bar{g}_{t - 1} (\bar{H}_{t - 1})$, $(Y_{g}, \underaccent{\bar}{A}_{t, g}, \underaccent{\bar}{S}_{t, g})$ is a deterministic function of $(Y_{g}^{d}, \underaccent{\bar}{L}_{t + 1, g}, \underaccent{\bar}{R}_{t + 1, g})$ and $\bar{H}_{t}$. For future reference, we highlight that assumption \ref{id3} implies 
\begin{equation}
\E \left[ \left. Y_{\bar{S}_{t - 1}, \bar{A}_{t - 1}, s_{t}, a_{t}, \underaccent{\bar}{g}_{t + 1}} \right\vert \bar{H}_{t} \right] = \E \left[ \left. Y_{\bar{S}_{t - 1}, \bar{A}_{t - 1}, s_{t}, a_{t}, \underaccent{\bar}{g}_{t + 1}} \right\vert \bar{H}_{t}, A_{t}, S_{t} \right], \label{radom-eq}
\end{equation}
where for any $g, Y_{\bar{S}_{t - 1}, \bar{A}_{t - 1}, s_{t}, a_{t}, \underaccent{\bar}{g}_{t + 1}}$ denotes the counterfactual total utility $Y$ when a subject takes her observed treatment history $(\bar{S}_{t - 1}, \bar{A}_{t - 1})$ through $t - 1$, and possibly contrary to fact, takes $(s_{t}, a_{t})$ at $t$, and, for $t < K$ the subject follows the dynamic testing and treatment regime $g$ from $t + 1$ onward.

Throughout we call the set of three assumptions \ref{id1} - \ref{id3} the identifying (ID) assumptions. \citet{robins1999testing} noted that assumptions \ref{id2} and \ref{id3} do not impose restrictions on the observed data distribution. \citet{robins1999testing} showed that under the ID assumptions, for any $(s_{t}, a_{t})$, the conditional counterfactual mean $\E [Y_{\bar{S}_{t - 1}, \bar{A}_{t - 1}, \underaccent{\bar}{g}_{t}} \vert \bar{H}_{t} = \bar{h}_{t}]$ is identified and it is equal to the following, so-called g-formula \citep{robins1986new} 
\begin{align*}
& \E \left[ \left. Y_{\bar{S}_{t - 1}, \bar{A}_{t - 1}, \underaccent{\bar}{g}_{t}} \right\vert \bar{H}_{t} = \bar{h}_{t} \right] \\
& = \int y f (y | \bar{h}_{K}, \bar{a}_{K}, \bar{s}_{K}) \prod_{m = t}^{K} I_{g_{m} (\bar{h}_{m})} (a_{m}, s_{m}) \prod_{m = t + 1}^{K} f (l_{m}, r_{m} | \bar{h}_{m - 1}, \bar{a}_{m - 1},\bar{s}_{m - 1}) d \underaccent{\bar}{a}_{m} d \underaccent{\bar}{s}_{m} d \underaccent{\bar}{l}_{m + 1} d \underaccent{\bar}{r}_{m + 1} dy.
\end{align*}
Note that $\E [Y_{\bar{S}_{t - 1}, \bar{A}_{t - 1}, \underaccent{\bar}{g}_{t}} \vert \bar{H}_{t} = \bar{h}_{t}]$ depends on the law of the observed data only through the conditional distributions $f (y | \bar{h}_{K}, a_{K}, s_{K})$ and $f (l_{m}, r_{m} | \bar{h}_{m - 1}, \bar{a}_{m - 1}, \bar{s}_{m - 1})$ for $m = t + 1, \ldots, K$. Let $\mathcal{G}$ be the set of all regimes $g$ satisfying the ID assumptions. Thus, for every $g \in \mathcal{G}$, $\E [Y_{\bar{S}_{t - 1}, \bar{A}_{t - 1}, \underaccent{\bar}{g}_{t}} \vert \bar{H}_{t} = \bar{h}_{t}]$ and $\E [Y_{g}]$ are identified. Our goal is to estimate an optimal dynamic testing and treatment regime $g^{opt} \equiv \arg \max_{g \in \mathcal{G}} \E [Y_{g}]$ and its corresponding value (i.e. expected utility) $\E [Y_{g^{opt}}]$. Since we have excluded exceptional laws, $g^{opt}$ is unique \citep{robins2004optimal}.

Under the ID assumptions, $g^{opt}$ can be computed by dynamic programming \citep{bellman1952theory} as follows. Define $g^{\mathsf{opt}\ast} = (g_{0}^{opt \ast}, \ldots, g_{K}^{opt \ast})$ by the following backward recursion. First, we define $g_{K}^{opt \ast} (\bar{H}_{K}) \equiv \arg \max_{s_{K}, a_{K}} \E [Y_{\bar{S}_{K - 1}, \bar{A}_{K - 1}, s_{K}, a_{K}} \vert \bar{H}_{K}]$ and recursively for $t = K - 1, \ldots, 0$, we define $g_{t}^{opt \ast} (\bar{H}_{t}) \equiv \arg \max_{s_{t}, a_{t}} \E [Y_{\bar{S}_{t - 1}, \bar{A}_{t - 1}, s_{t}, a_{t}, \underaccent{\bar}{g}_{t + 1}^{opt \ast}} \vert \bar{H}_{t}]$. \citet{robins2004optimal} proved that under the ID assumptions, $\E [Y_{\bar{S}_{t - 1}, \bar{A}_{t - 1}, \underaccent{\bar}{g}_{t}^{opt \ast}} \vert \bar{H}_{t}] = \max_{\underaccent{\bar}{g}_{t} \in \underaccent{\bar}{\mathcal{G}}_{t}} \E [Y_{\bar{S}_{t - 1}, \bar{A}_{t - 1}, \underaccent{\bar}{g}_{t}} \vert \bar{H}_{t}]$, thus proving that the DP solution $g^{opt \ast}$ agrees with the optimal treatment regime $g^{opt}$ under the ID assumptions.

\section{Estimator of opt-SNMM Without the NDE Assumption}
\label{sec:estimator}

In this section we review the definition and the estimation of opt-SNMMs \citep{robins2004optimal} for estimating optimal dynamic regimes. For any regime $g$, define the testing and treatment effect contrast 
\begin{align*}
\gamma_{t}^{\underaccent{\bar}{g}_{t + 1}} (\bar{H}_{t},s_{t},a_{t}) \coloneqq \E\left[ \left. Y_{\bar{S}_{t - 1}, \bar{A}_{t - 1}, s_{t}, a_{t}, \underaccent{\bar}{g}_{t + 1}} - Y_{\bar{S}_{t - 1}, \bar{A}_{t - 1}, s_{t} = 0, a_{t} = 0, \underaccent{\bar}{g}_{t + 1}} \right\vert \bar{H}_{t}, \left( S_{t}, A_{t} \right) = \left( s_{t}, a_{t} \right) \right].
\end{align*}
This contrast is the average causal effect, among subjects with the history $\bar{H}_{t}, (S_{t}, A_{t}) = (s_{t}, a_{t})$, of setting possibly contrary to fact $S_{t}$ and $A_{t}$ both to 0 rather than to their observed values when, again possibly contrary to fact, the regime $g$ is followed from time $t + 1$ onward. Because $Y_{\bar{S}_{t - 1}, \bar{A}_{t - 1}, s_{t} = 0, a_{t} = 0, \underaccent{\bar}{g}_{t + 1}}$ does not depend on the free indices $(s_{t}, a_{t})$, 
\begin{equation*}
\arg \max_{(s_{t}, a_{t})} \gamma_{t}^{\underaccent{\bar}{g}_{t + 1}} (\bar{H}_{t}, s_{t}, a_{t}) = \arg \max_{(s_{t}, a_{t})} \E \left[ \left. Y_{\bar{S}_{t - 1}, \bar{A}_{t - 1}, s_{t}, a_{t}, \underaccent{\bar}{g}_{t + 1}} \right\vert \bar{H}_{t}, (S_{t}, A_{t}) = (s_{t},a_{t}) \right]
\end{equation*}
under the ID assumptions. It then follows, that by the arguments given in the preceding section, under the ID assumptions, the optimal testing and treatment regime $g^{opt}$ is given by the following backward recursion: for $t = K, K - 1, \ldots, 0$, $g_{t}^{opt} (\bar{H}_{t}) \coloneqq \arg \max_{(s_{t}, a_{t})} \gamma_{t}^{\underaccent{\bar}{g}_{t + 1}^{opt}} (\bar{H}_{t}, s_{t}, a_{t})$. For $t = 0, \ldots, K$, define $\gamma_{t} (\bar{H}_{t}, s_{t}, a_{t}) \equiv \gamma_{t}^{\underaccent{\bar}{g}_{t + 1}^{opt}} (\bar{H}_{t}, s_{t}, a_{t})$ and $\left( S_{t}^{opt} (\gamma_{t}), A_{t}^{opt} (\gamma_{t}) \right) \equiv g_{t}^{opt} (\bar{H}_{t})$. Finally, define 
\begin{align}
\Delta_{t} (\gamma_{t}; \underaccent{\bar}{\gamma}_{t + 1}) \equiv Y - \gamma_{t} (\bar{H}_{t}, S_{t}, A_{t}) + \sum_{m = t + 1}^{K} \left\{ \gamma_{m} \left( \bar{H}_{m}, S_{m}^{opt} \left( \gamma_{m} \right), A_{m}^{opt} \left( \gamma_{m} \right) \right) - \gamma_{m} \left( \bar{H}_{m}, S_{m}, A_{m} \right) \right\}. \notag
\end{align}
where by convention $\sum_{m = K + 1}^{K} \left( \cdot \right) \equiv 0$. By straightforward algebra it can be shown that for any $t$, 
\begin{equation}
\E [\Delta_{t} (\gamma_{t}; \underaccent{\bar}{\gamma}_{t + 1}) | \bar{H}_{t}, A_{t}, S_{t}] = \E [Y_{\bar{A}_{t - 1}, \bar{S}_{t - 1}, a_{t} = 0, s_{t} = 0, \underaccent{\bar}{g}_{t + 1}^{opt}} | \bar{H}_{t}, A_{t}, S_{t}] \label{new_eq}
\end{equation}
regardless of whether or not the sequential exchangeability holds. Heuristically $\Delta_{t} (\gamma_{t}; \underaccent{\bar}{\gamma}_{t + 1})$ mimics $Y_{\bar{A}_{t - 1}, \bar{S}_{t - 1}, a_{t} = 0, s_{t} = 0, \underaccent{\bar}{g}_{t + 1}^{opt}}$ in the sense that they have the same mean given $(\bar{H}_{t}, A_{t}, S_{t})$. Under the ID assumptions and equation \eqref{new_eq} 
\begin{equation}  \label{eq:preU}
\E \left[ \Delta_{t} (\gamma_{t}; \underaccent{\bar}{\gamma}_{t + 1}) | \bar{H}_{t}, A_{t}, S_{t} \right] = \E \left[ \Delta_{t} (\gamma_{t}; \underaccent{\bar}{\gamma}_{t + 1}) | \bar{H}_{t} \right],
\end{equation}
which implies that, for arbitrary functions $Q_{t} (s_{t}, a_{t}) \equiv q_{t} (\bar{H}_{t}, s_{t}, a_{t})$ selected by the analyst, the random variable 
\begin{equation}  \label{eq:U}
U_{t} (q_{t}, \underaccent{\bar}{\gamma}_{t}) \equiv \left\{ \Delta_{t} (\gamma_{t}; \underaccent{\bar}{\gamma}_{t + 1}) - \E [\Delta_{t} (\gamma _{t}; \underaccent{\bar}{\gamma}_{t + 1}) | \bar{H}_{t}] \right\} \times \left\{ Q_{t} (S_{t}, A_{t}) - \E [Q_{t} (S_{t}, A_{t}) | \bar{H}_{t}] \right\}
\end{equation}
has mean zero by equation \eqref{eq:preU}. In fact, \citet{robins2004optimal} showed that $\gamma_{t} (\bar{H}_{t}, S_{t}, A_{t})$ is the unique function of $(\bar{H}_{t}, S_{t}, A_{t})$ satisfying $\gamma_{t} (\bar{H}_{t}, 0, 0) = 0$ that also satisfies the condition $\E [U_{t} (q_{t}, \underaccent{\bar}{\gamma}_{t})] = 0$ for all $q_{t}$ such that $U_{t} (q_{t}, \underaccent{\bar}{\gamma}_{t})^{2}$ has finite mean. Therefore $\gamma_{t}$ is identified through the system of equations \eqref{eq:U} for all $q_{t}$. One could choose $Q_{t} (S_{t}, A_{t})$ either to simplify the form or minimize the variance of the opt-SNMM estimator defined below.

An opt-SNMM assumes a parametric model for the treatment effect contrasts $\gamma_{t} (\bar{H}_{t}, s_{t}, a_{t})$, i.e. it assumes that 
\begin{equation}
\gamma_{t} (\bar{H}_{t}, s_{t}, a_{t}) = \gamma_{t} (\bar{H}_{t}, s_{t}, a_{t}; \Psi_{t}^{\ast})  \label{model}
\end{equation}
where $\Psi_{t}^{\ast}$ is an unknown parameter vector and $\gamma_{t} (\bar{H}_{t}, s_{t}, a_{t}; \Psi_{t})$ is a known function equal to 0 whenever $s_{t} = a_{t} = 0$ or $\Psi_{t} = 0$.

Under the model the optimal testing and treatment choice at $t$ is $(S_{t}^{opt} (\Psi_{t}), A_{t}^{opt} (\Psi_{t})) \equiv \arg \max_{s_{t}, a_{t}} \gamma_{t} (\bar{H}_{t}, s_{t}, a_{t}; \Psi_{t})$ for $\Psi = \Psi^{\ast}$. In an abuse of notation, we define $\Delta_{t} (\Psi_{t}; \underaccent{\bar}{\Psi}_{t + 1})$ just like $\Delta_{t} (\gamma_{t}; \underaccent{\bar}{\gamma}_{t + 1})$ but with the functions $\gamma_{m} (\bar{H}_{m}, S_{m}, A_{m}; \Psi _{m})$ replacing the true functions $\gamma_{m} (\bar{H}_{m}, S_{m}, A_{m})$.

The ID assumptions and the opt-SNMM model \eqref{model} determine a semiparametric model $\mathbb{M}_{1}$ for the observed data distribution defined by the restriction that there exists a unique $\Psi^{\ast} \equiv (\Psi_{0}^{\ast}, \ldots, \Psi_{K}^{\ast})$ such that for all $t = 0, 1, \ldots, K$, $\E [\Delta_{t} (\Psi_{t}^{\ast}; \underaccent{\bar}{\Psi}_{t + 1}^{\ast}) | \bar{H}_{t}, A_{t}, S_{t}] = \E [\Delta_{t} (\Psi_{t}^{\ast}; \underaccent{\bar}{\Psi}_{t + 1}^{\ast}) | \bar{H}_{t}]$. Before proceeding further we need to review some well-known results from semiparametric theory.

\subsection{Review of semiparametric theory}
\label{sec:review}

For a law $P$ in a semiparametric model $\mathbb{M}$, the tangent space $\Lambda \equiv \Lambda (P)$ at $P$ is defined as the closed linear span of the scores $\phi (P)$ at $P$ for all one dimensional parametric submodels $\alpha \rightarrow P_{\alpha}$ such that $P_{\alpha = 0} = P$. For a given functional of interest $\beta \equiv \beta (P)$, a random variable $\IF = \IF (P)$ is referred to as an influence function for $\beta$ at $P$ if $\E_{P} [\IF (P)] = 0$ and for each submodel $P_{\alpha}$, $\partial \beta (P_{\alpha}) / \partial \alpha |_{\alpha = 0} = \E_{P} [\IF (P) \phi (P)]$ for the score $\phi (P)$ of model $P_{\alpha}$ at $\alpha = 0$. The efficient influence function $\EIF = \EIF (P)$ for a functional $\beta$ is the unique influence function lying in $\Lambda$. A random variable $\IF$ is an influence function if and only if $\IF - \EIF \in \Lambda^{\perp}$ where $\Lambda^{\perp}$ is the ortho-complement to $\Lambda$ in $L_{2, 0} (P)$ where $L_{2, 0} (P)$ is the subspace of $L_{2} (P)$ comprised of mean zero elements of $L_{2} (P)$. Hence the set $\mathcal{IF}$ of all influence functions (at $P)$ is $\left\{ \EIF + H; H \in \Lambda^{\perp} \right\}$. It follows that for any influence function $\IF$, $\EIF = \IF - \Pi [\IF | \Lambda^{\perp}]$ with $\Pi = \Pi_{P}$ the projection operator in $L_{2, 0} (P)$. The nuisance tangent space $\Lambda_{\mathsf{nuis}} \equiv \Lambda_{\mathsf{nuis}} (P)$ for $\beta$ at $P$ is the subspace of $\Lambda (P)$ generated by the scores of one dimensional parametric submodels in $\mathbb{M}$ for which $\Psi (P_{\alpha})$ is constant. The set $\mathcal{IF}$ of influence functions is contained in $\Lambda_{\mathsf{nuis}}^{\perp}$.

Estimators $\hat{\beta}$ of $\beta = \beta (P)$ with the property that there exists a random variable $D_{i} \equiv d (O_{i})$ with mean zero and finite variance under $P$ such that $n^{1/2} (\hat{\beta} - \beta) = n^{-1/2} \sum_{i = 1}^{n} D_{i} + o_{p} (1)$ are called asymptotically linear at $P$. By the Central Limit Theorem and Slutsky's Theorem, any asymptotically linear estimator is consistent and asymptotically normal with asymptotic variance equal to $\var (D)$. Furthermore, any two asymptotically linear estimators, say $\hat{\beta}_{1}$ and $\hat{\beta}_{2}$ with the same $D$ are asymptotically equivalent in the sense that $n^{1/2} (\hat{\beta}_{1} - \hat{\beta}_{2}) = o_{p} (1)$. An estimator of $\hat{\beta}$ is regular at $P$ in a semiparametric model $\mathbb{M}$ if its convergence to $\beta$ is locally uniform \citep{van1998asymptotic}. Any regular, asymptotically linear (RAL) estimator has $D$ equal to some influence function $\mathsf{IF}$. It follows that the minimal possible variance of any RAL estimator is $\E [\EIF^{2}]$, which is referred to as the semiparametric variance bound for $\beta$ at $P$ in model $\mathbb{M}$. If $\Lambda = \Lambda (P)$ is all of $L_{2, 0} (P)$, then all RAL estimators have the same influence function at $P$.

\subsection{Semiparametric inference in model $\mathbb{M}_{1}$}
\label{sec:review}

The set $\Lambda_{\mathsf{ancillary}}$ of conditional scores for $p [A_{t}, S_{t} | \bar{H}_{t}]$, for $t = 0, \cdots, K$ consists of all functions of $(A_{t}, S_{t}, \bar{H}_{t})$ in $L_{2, 0} (P)$ that have mean zero given $\bar{H}_{t}$. Let $\Lambda_{1, \mathsf{nuis}}$ denote the nuisance tangent space of model $\mathbb{M}_{1}$. Because $\Psi^{\ast} \equiv \Psi^{\ast} (P)$ does not depend on $P$ through $p [A_{t}, S_{t} | \bar{H}_{t}]$ for $t = 0, \cdots, K$, we conclude that $\Lambda_{\mathsf{ancillary}} \subset \Lambda_{1, nuis}$.

\citet[Theorem 3.3, eq (3.10)]{robins2004optimal} proved that under model $\mathbb{M}_{1}$, 
\begin{equation*}
\Lambda_{1, \mathsf{nuis}}^{\perp} = \left\{ \BU (\bm{q}, \Psi^{\ast}) = \sum_{t = 0}^{K} U_{t} (q_{t}, \underaccent{\bar}{\Psi}_{t}^{\ast}); q_{t}(\bar{H}_{t}, a_{t}, s_{t}) \right\}
\end{equation*}
where the $q_{t}$ are vector functions of $\dim (\Psi _{t}^{\ast})$ that are unrestricted except for the requirement $\E [U_{t} (q_{t})^{2}] < + \infty$ where $Q_{t} (s_{t}, a_{t}) \equiv q_{t} (\bar{H}_{t}, s_{t}, a_{t})$ and 
\begin{equation*}
U_{t} (q_{t}, \underaccent{\bar}{\Psi}_{t}) \equiv \left\{ \Delta_{t} (\Psi_{t}; \underaccent{\bar}{\Psi}_{t + 1}) - \E [\Delta_{t} (\Psi_{t}; \underaccent{\bar}{\Psi}_{t + 1}) | \bar{H}_{t}] \right\} \times \left\{ Q_{t} (S_{t}, A_{t}) - \E [Q_{t} (S_{t}, A_{t}) | \bar{H}_{t}] \right\}.
\end{equation*}
Furthermore, $U_{t} (q_{t}, \underaccent{\bar}{\Psi}_{t})$ is a doubly robust estimating function in the sense that it still has mean zero at $\underaccent{\bar}{\Psi}_{t}^{\ast}$ if either (but not both) $\E \left[ \Delta_{t} (\Psi_{t}^{\ast}; \underaccent{\bar}{\Psi}_{t + 1}^{\ast}) | \bar{H}_{t} \right]$ or $\E [Q_{t} (S_{t}, A_{t}) | \bar{H}_{t}]$ is replaced by an arbitrary function of $\bar{H}_{t}$.

The identity $\E [\BU (\bm{q}, \Psi^{\ast})] = 0$ suggests that one solve $\BP_{n} [\BU (\bm{q}, \Psi)] = 0$ to estimate $\Psi^{\ast}$. Assuming that $\E [\BU (\bm{q}, \Psi)] = 0$ has a unique solution and that $\left. \frac{\partial}{\partial \Psi} \E [\BU (\bm{q}, \Psi)] \right\vert_{\Psi = \Psi^{\ast}}$ is invertible, $\hat{\Psi} (\bm{q})$ will be a RAL estimator of $\Psi^{\ast}$ under standard regularity conditions. For a specific choice $\bm{q}_{opt} = \bm{q}_{opt} (P)$ of $\bm{q}$, the estimator $\hat{\Psi} (\bm{q}_{opt})$ attains semiparametric efficiency bound for regular estimators of $\Psi^{\ast}$ under model $\mathbb{M}_{1}$. Because of its dependence on $P$,  $\bm{q}_{opt}$ would have to be estimated from the data. In this paper, we decided not to consider the issue of finding $\bm{q}_{opt}$ and instead focus only on adjusting for the NDE assumption. Finding $\bm{q}_{opt}$ often leads to solving complicated integral equations. In practice, a heuristic choice of $\bm{q}$ that generally does not depend on the data can be used \citep{vansteelandt2014structural}.

Notice that $\hat{\Psi} (\bm{q})$ is infeasible because $\BU (\bm{q}, \Psi^{\ast})$ depends on the unknowns $\E [\Delta_{t} (\Psi_{t}; \underaccent{\bar}{\Psi}_{t + 1}) | \bar{H}_{t}]$ and $\E [Q_{t} (S_{t}, A_{t}) | \bar{H}_{t}]$. A feasible estimator $\tilde{\Psi} (\bm{q})$ solves $\BP_{n} [\hat{\BU} (\bm{q}, \Psi)] = 0$ where $\hat{\BU} (\bm{q}, \Psi)$ is defined like $\BU (\bm{q}, \Psi)$ but with estimates $\hat{\E} [\Delta_{t} (\Psi_{t}; \underaccent{\bar}{\Psi}_{t + 1}) | \bar{H}_{t}]$ and $\hat{\E} [Q_{t} (S_{t}, A_{t}) | \bar{H}_{t}]$ replacing the unknown true expectations. We discuss feasible estimators in Section \ref{sec:feasible}. Until then, to focus on important issues, we restrict our discussion to infeasible, also called oracle, estimators.

Given oracle estimators $\hat{\Psi} \equiv \hat{\Psi} (\bm{q})$, we can estimate the optimal value function $\theta_{g^{opt}} \equiv \E [Y_{g^{opt}}]$ by noting $\E [\Delta_{0} (\Psi_{0}^{\ast}; \underaccent{\bar}{\Psi}_{1}^{\ast})] = \E [Y_{a_{0} = 0, \underaccent{\bar}{g}_{1}^{opt}}]$ and hence $\E [Y_{g^{opt}}] = \E [\Delta_{0} (\Psi_{0}^{\ast}; \underaccent{\bar}{\Psi}_{1}^{\ast}) + \gamma_{0} (L_{0}, R_{0}, S_{0}^{opt} (\underaccent{\bar}{\Psi}_{0}^{\ast}), A_{0}^{opt} (\underaccent{\bar}{\Psi}_{0}^{\ast }); \Psi_{0}^{\ast})]$. Our oracle estimate of $\theta_{g^{opt}} = \E [Y_{g^{opt}}]$ is then $\hat{\theta}_{\SNMM, g^{opt}}$ that solves 
\begin{equation}
\BP_{n} [U_{\SNMM, g^{opt}} (\theta, \hat{\Psi})] = 0, U_{\SNMM, g^{opt}} (\theta, \hat{\Psi}) = \Delta_{0} (\hat{\Psi}_{0}; \underaccent{\bar}{\hat{\Psi}}_{1}) + \gamma_{0} (\bar{L}_{0}, \bar{R}_{0}, S_{0}^{opt} (\underaccent{\bar}{\hat{\Psi}}_{0}), A_{0}^{opt} (\hat{\underaccent{\bar}{\Psi}}_{0}); \hat{\Psi}_{0}) - \theta. \label{optsnmm}
\end{equation}

\section{The No Direct Effect of testing assumption}
\label{sec:nde}

The restriction 
\begin{equation}
Y_{\bar{a}_{K}, \bar{s}_{K}}^{d} = Y_{\bar{a}_{K}^{\prime }, \bar{s}_{K}}^{d} \ \forall \bar{a}_{K}, \bar{a}_{K}^{\prime}, \bar{s}_{K}
\label{nde_test}
\end{equation}
encodes the assumption that testing history $\bar{A}_{K}$ has no direct effect on observed health outcome $Y^{d}$ not through treatment $\bar{S}_{K}$. We refer to equation \eqref{nde_test} as the No Direct Effect (NDE) on $Y^{d}$ [NDE$(Y^{d})$] assumption. The assumption means that if we intervene and set the treatment history $\bar{S}_{K}$ to any $\bar{s}_{K}$, then further intervening on testing $A$ at any time has no effect on the health outcome $Y^{d}$. [Note however that $A_{t}$ has a direct effect on the total utility $Y = Y^{d} - c^{\ast} \sum_{t} A_{t}$ for $c^{\ast} \neq 0$ even under the NDE assumption for the health outcome $Y^{d}$.] We will write NDE as shorthand for NDE$(Y^{d})$.

Consider the model defined by the ID assumptions \ref{id1}-\ref{id3} and the assumption \eqref{nde_test}. \citet{robins1999testing} showed that such a model determines a model $\mathbb{M}_{\mathsf{NDE}}$ for the observed data characterized, possibly up to inequality constraints, by the following set of restrictions, with $W_{t} \equiv \prod_{m = t}^{K} p (S_{m} | \bar{H}_{m})$ 
\begin{equation}
\E \left[ \left. \frac{b_{t} (\bar{H}_{t}, \underaccent{\bar}{S}_{t}, Y^{d})}{W_{t + 1}} \right\vert \bar{H}_{t}, S_{t}, A_{t} \right] = \E \left[ \left. \frac{b_{t} (\bar{H}_{t}, \underaccent{\bar}{S}_{t}, Y^{d})}{W_{t + 1}} \right\vert \bar{H}_{t}, S_{t} \right] \ t = 0, \ldots, K \label{rest2}
\end{equation}
for any function $b_{t} (\bar{H}_{t},\underaccent{\bar}{S}_{t},Y^{d})$ of all the history $\bar{H}_{t}$ up to time $t$, all the treatments $\underaccent{\bar}{S}_{t}$ at $t$ and onward, and the disease outcome $Y^{d}$.

The intuition behind these restrictions is that, under the three ID assumptions, weighting by $W_{t + 1}^{-1}$, which is heuristically, the inverse of the probability of $\underaccent{\bar}{S}_{t + 1}$, mimics a randomized trial in which each of the treatments in $\underaccent{\bar}{S}_{t + 1}$ are assigned independently with probability $1 / 2$ and the testing indicator $A_{t}$ is randomly assigned given the past $(\bar{H}_{t}, S_{t})$. Therefore, in this weighted distribution, $A_{t}$ is independent of $(\underaccent{\bar}{S}_{t + 1}, Y^{d})$ [and thus of any function $b_{t} (\bar{H}_{t}, \underaccent{\bar}{S}_{t}, Y^{d})$] conditional on $(\bar{H}_{t}, S_{t})$, by $\underaccent{\bar}{S}_{t + 1}$ being completely randomized and the NDE assumption.

In what follows, for any $b_{t}, t = 0, \ldots, K$, we let $\mathcal{B}_{t} \equiv \left\{ b_{t} (\bar{H}_{t}, \underaccent{\bar}{S}_{t}, Y^{d}): \E [T_{b, t}^{2}] < + \infty \right\}$ and $\mathcal{T}_{t} \equiv \left\{ T_{b, t}; b_{t} \in \mathcal{B}_{t} \right\}$ where 
\begin{eqnarray}
T_{b, t} &\equiv& D_{b, t} - \left\{ \E [D_{b, t} | \bar{H}_{t}, S_{t}, A_{t}] - \E [D_{b, t} | \bar{H}_{t}, S_{t}] \right\} - \sum_{m = t + 1}^{K} \left\{ \E [D_{b, t} | \bar{H}_{m}, S_{m}] - \E [D_{b, t} | \bar{H}_{m}] \right\}, \label{tbt-def} \\
\text{ with } D_{b, t} &\equiv& \frac{b_{t} (\bar{H}_{t}, \underaccent{\bar}{S}_{t}, Y^{d})}{W_{t + 1}} (A_{t} - \E [A_{t} | \bar{H}_{t}, S_{t}]).
\end{eqnarray}
Here $T_{b, t}$ is the residual from the projection of $D_{b, t}$ on the space 
\begin{equation*}
\Lambda_{\mathsf{ancillary}} = \left\{ V = \sum_{t = 0}^{K} \left\{ V_{t} - \E [V_{t} | \bar{H}_{t}] \right\}; V_{t} = v_{t} (\bar{H}_{t}, S_{t}, A_{t}): \E [\var (V_{t} | \bar{H}_{t})] < + \infty \right\} 
\end{equation*}
of the conditional scores for the testing and treatment processes.

\begin{remark}
Under the ID assumptions, the NDE assumption is empirically falsifiable; since the NDE$(Y^{d})$ assumption implies $\E [D_{b, t}] = 0$ for all $b_{t} \in \mathcal{B}_{t}$. Owing to limited space, we do not consider falsification tests further.
\end{remark}

\begin{remark}
The projection of any random variable $C = c(O)$ on conditional scores for testing and treatment is $\Pi [C | \Lambda_{\mathsf{ancillary}}] = \sum\limits_{m = 0}^{K} (\E [C | \bar{H}_{m}, S_{m}, A_{m}] - \E [C | \bar{H}_{m}])$. In Appendix \ref{app:nuisance_nde}, however, we show that the NDE assumption implies that if $C = D_{b, t}$, this projection does not depend on conditional scores for $\Pr (A_{m} = 1 | \bar{H}_{m}, S_{m})$ for $m > t$ and $\E [D_{b, t} | \bar{H}_{t}, S_{t}] = 0$ proving equation \eqref{tbt-def} (for $m < t$, both $\E [C | \bar{H}_{m}, S_{m}, A_{m}]$ and $\E [C | \bar{H}_{m}]$ are zero by the NDE assumption).
\end{remark}

We then have the following theorems, the proofs of which are deferred to Appendix \ref{app:nuisance_nde} and \ref{app:alternative_nde}.

\begin{theorem}
\label{lem:nuisance_nde} 
The space $\Lambda_{\mathsf{NDE}}^{\perp}$ of mean zero random variables orthogonal to the tangent space $\Lambda_{\mathsf{NDE}}$ of model $\mathbb{M}_{\mathsf{NDE}}$ is given by 
\begin{equation*}
\Lambda_{\mathsf{NDE}}^{\perp} = \mathcal{T}_{0} + \cdots + \mathcal{T}_{K} = \left\{ T_{b} \equiv \sum_{t = 0}^{K} T_{b, t}; b \equiv (b_{0}, \ldots, b_{K}) \text{ with } b_{t} \in \mathcal{B}_{t}, \ \ t = 0, \ldots, K \right\}.
\end{equation*}
\end{theorem}

\begin{theorem}
\label{thm:tdr} 
Each $T_{b} \in \Lambda_{\mathsf{NDE}}^{\perp}$ is doubly robust in the sense that it has mean zero if either (i) the true densities $p (L_{k + 1} | \bar{H}_{k}, S_{k})$ are replaced by arbitrary densities $p^{\dag} (L_{k + 1} | \bar{H}_{k}, S_{k})$ for $k = 0, \ldots, K$ except with $H_{K + 1} \equiv Y^{d}$ replacing $L_{K + 1}$ or (ii) $p (S_{k} | \bar{H}_{k})$ and $\Pi_{k}$ are replaced by arbitrary conditional probability functions $p^{\dag} (S_{k} | \bar{H}_{k})$ and $\Pi_{k}^{\dag}$ for $k = 0, \ldots, K$, but not both.
\end{theorem}

We prove Theorem \ref{thm:tdr} by actually proving a stronger version of double robustness in Appendix \ref{app:alternative_nde}, based on an alternative representation of any $T_{b}$.

\section{Oracle estimators of $\Psi^{\ast}$ with improved efficiency}
\label{sec:dr}

\subsection{Sub-optimal estimators with improved efficiency}
\label{sec:simple_dr}

Let the semiparametric model $\mathbb{M}_{int} = \mathbb{M}_{1} \cap \mathbb{M}_{\mathsf{NDE}}$ consist of the set of observed data laws that satisfy the restrictions of models $\mathbb{M}_{1}$ and $\mathbb{M}_{\mathsf{NDE}}$. Since under model $\mathbb{M}_{\mathsf{NDE}}$, $\E [T_{b}] = 0$ for $b \equiv (b_{0}, \ldots, b_{K})^{\top}$ where $T_{b} \equiv \sum_{t = 0}^{K} T_{b, t}$ and $T_{b, t}$ is defined as in equation \eqref{tbt-def}, it follows that under model $\mathbb{M}_{int}$ we can consider the class of estimating functions $\left\{ \BU (\bm{q}, b, \Psi) = \BU (\bm{q}, \Psi) - c_{OLS} (\bm{q}, b, \Psi) T_{b}; b \right\}$ where $c_{OLS} (\bm{q}, b, \Psi) = \E [\BU (\bm{q}, \Psi) T_{b}^{\top}] \{\E [T_{b} T_{b}^{\top}]\}^{-1}$. Note that $\hat{\Psi} (\bm{q})$ defined in Section \ref{sec:estimator} is the same as $\hat{\Psi} (\bm{q}, b)$ for $b$ identically zero. Define $J = \frac{\partial}{\partial \Psi^{\top}} \left. \E [\BU(\bm{q}, b)] \right\vert_{\Psi = \Psi^{\ast}} = \frac{\partial}{\partial \Psi^{\top}} \left. \E [\BU (\bm{q})] \right\vert_{\Psi = \Psi^{\ast}}$. Since $J$ does not depend on $b$, we have the following result:

\begin{theorem}
Suppose $\hat{\Psi} (\bm{q}, b)$ is a RAL estimator of $\Psi^{\ast}$. Then it has influence function $J^{-1} \{\BU (\bm{q}, \Psi^{\ast}) - c_{OLS} (\bm{q}, b, \Psi) T_{b}\}$. Consequently, $\sqrt{n} \{\hat{\Psi} (\bm{q}, b) - \Psi^{\ast}\}$ is asymptotically normal with mean zero and variance equal to 
\begin{equation*}
V_{oracle} (\bm{q}, b) \equiv J^{-1} (\var [\BU (\bm{q}, \Psi)] - c_{OLS} (\bm{q}, b, \Psi) \E [T_{b} T_{b}^{\top}] c_{OLS} (\bm{q}, b, \Psi)^{\top}) J^{-\top}.
\end{equation*}
Furthermore, $\hat{\Psi} (\bm{q})$ has influence function $J^{-1} \BU (\bm{q}, \Psi^{\ast})$ and asymptotic variance $V_{oracle} (\bm{q}, b = 0)$ greater than $V_{oracle} (\bm{q}, b)$ (in the positive definitive sense) whenever $\E [\BU (\bm{q}, \Psi^{\ast}) T_{b}] \neq 0$.
\end{theorem}

\subsection{Near optimal oracle estimators with improved efficiency}
\label{sec:optimal_dr}

Our goal is to improve efficiency by finding the function $b_{opt} (\bm{q})$ that minimizes $V_{oracle} (\bm{q}, b)$ for a given $\mathbf{q}$. To do so, we first define the orthogonal projection of a mean zero random variable $U$ onto a closed linear space $\digamma$ of mean zero random variables to be the unique element $F^{\ast} \in \digamma$ defined as $F^{\ast} = \arg \min_{F \in \digamma} \var [U - F]$ in the positive definite sense. Now for any two nested subspaces $\Phi_{l} \subset \Phi_{l^{\prime}} \subseteq \Lambda_{\mathsf{NDE}}^{\perp}$, let $T_{b_{l}}$ and $T_{b_{l^{\prime}}}$ be equal to $\Pi [\BU (\bm{q}, \Psi^{\ast}) | \Phi_{l}]$ and $\Pi [\BU (\bm{q}, \Psi^{\ast}) | \Phi_{l^{\prime}}]$. From the definition of a projection, $\var (T_{b_{l}}) \leq \var (T_{b_{l^{\prime}}})$. It follows that $V_{oracle} (\bm{q}, b_{l}) \geq V_{oracle} (\bm{q}, b_{l^{\prime}})$, demonstrating that the larger the subspace $\Phi$ of $\Lambda_{\mathsf{NDE}}^{\perp}$ one projects on, the more efficient the estimator. It also follows that $b_{opt} (\bm{q})$ is the unique function $b = (b_{0}, \ldots, b_{K})^{\top}$ that satisfies $T_{b} = \Pi [\BU (\bm{q}, \Psi^{\ast}) | \Lambda_{\mathsf{NDE}}^{\perp}]$. Here each $b_{t} \in \mathcal{B}_{t}$ is a column vector function of the same dimension as $\Psi_{t}^{\ast}, t = 0, \ldots, K$.

Our next task is to characterize $\Pi [U | \Lambda_{\mathsf{NDE}}^{\perp}]$ for any random variable $U$. In the Appendix \ref{app:cont} we show that when $U$ is a continuous random variable, $\Pi [U | \Lambda_{\mathsf{NDE}}^{\perp}]$ is a solution to a set of complicated integral equations that do not exist in closed form. However, for $U$ a discrete random variable with finite sample space we will give below a closed form expression for $\Pi [U | \Lambda_{\mathsf{NDE}}^{\perp}]$. Furthermore, in the case of a continuous $U$, we will derive a closed form expression for $\Pi [U | \Omega]$ for a particular subspace $\Omega \subset \Lambda _{\mathsf{NDE}}^{\perp}$. We will argue that the associated estimator $\hat{\Psi} (\bm{q}, b_{sub}) \equiv \hat{\Psi} (\bm{q}, b_{sub} (\bm{q}))$ where $T_{b_{sub}} = \Pi [\BU (\bm{q}, \Psi^{\ast}) | \Omega]$ should have efficiency nearly equal to that of the computationally intractable optimal estimator $\hat{\Psi} (\bm{q}, b_{opt}) \equiv \hat{\Psi} (\bm{q}, b_{opt} (\bm{q}))$, with $T_{b_{opt}} (\bm{q}) = \Pi [\BU (\bm{q}, \Psi^{\ast}) | \Lambda_{\mathsf{NDE}}^{\perp}]$. Our results are corollaries of the next theorem, proved in Appendix \ref{app:sub}.

In what follows for $t = 0, \ldots, K$, let $T_{t}$ be a fixed $\delta_{t} \times 1$ random vector and let $\Gamma_{t} \equiv \left\{ d_{t} (\bar{H}_{t}) T_{t}: d_{t} (\bar{H}_{t}) \text{ any } 1 \times \delta_{t} \text{ vector with } \E \left[ d_{t} (\bar{H}_{t})^{2} \right] < + \infty \right\}$. For $j = 0, \ldots, K$, define $T_{j}^{(K)} \equiv T_{j}$ and define also recursively for $t = K - 1, \cdots, j$, 
\begin{equation*}
T_{j}^{(t)} \equiv T_{j}^{(t + 1)} - \E \left[ \left. T_{j}^{(t + 1)} T_{t + 1}^{(t + 1) \top} \right\vert \bar{H}_{t + 1} \right] \E \left[ \left. T_{t + 1}^{(t + 1)} T_{t + 1}^{(t + 1) \top} \right\vert \bar{H}_{t + 1} \right]^{-1} T_{t+1}^{(t + 1)}.
\end{equation*}

\begin{theorem}
\label{thm:sub} 
For any random variable $U$, $\Pi [U | \Gamma_{0} + \ldots + \Gamma_{K}] = \sum_{t = 0}^{K} d_{t}^{\ast} (\bar{H}_{t}) T_{t}$ where 
\begin{equation}
d_{0}^{\ast} \left( \bar{H}_{0} \right) \equiv \E \left[ \left. U T_{0}^{(0) \top} \right\vert \bar{H}_{0} \right] \E \left[ \left. T_{0}^{(0)} T_{0}^{(0) \top} \right\vert \bar{H}_{0} \right]^{-1} \label{main_0}
\end{equation}
and for $t = 1, \ldots, K$, 
\begin{equation}
d_{t}^{\ast} \left( \bar{H}_{t} \right) \equiv \E \left[ \left. \left\{ U - \sum_{j = 0}^{t - 1} d_{j}^{\ast} \left( \bar{H}_{j} \right) T_{j}^{(t)} \right\} T_{t}^{(t) \top} \right\vert \bar{H}_{t} \right] \E \left[ \left. T_{t}^{(t)} T_{t}^{(t) \top} \right\vert \bar{H}_{t} \right]^{-1}. \label{main_d1}
\end{equation}
\end{theorem}

The preceding theorem has the following important consequences. Let $\mathcal{S}_{t}$ denote the sample space of the treatment variable $S_{t}$ and let $I_{t}$ be the $card(\mathcal{S}_{t} \times \cdots \times \mathcal{S}_{K}) \times 1$ vector whose elements are the indicators that $\underaccent{\bar}{S}_{t}$ take a specific value $\underaccent{\bar}{s}_{t} \in \mathcal{S}_{t} \times \cdots \times \mathcal{S}_{K}$, i.e. $I_{t} \equiv (I_{\underaccent{\bar}{s}_{t}}(\underaccent{\bar}{S}_{t}))_{\underaccent{\bar}{s}_{t} \in \mathcal{S}_{t} \times \cdots \times \mathcal{S}_{K}}$. Next, for any given $\xi \times 1$ vector $\bm{\varphi} (Y^{d}) \equiv (\varphi_{1} (Y^{d}), \ldots, \varphi_{\xi} (Y^{d}))^{\top}$ of linearly independent functions of $Y^{d}$, define for each $t = 0, \ldots, K$, the $\delta_{t} \times 1$ vector function 
\begin{equation}
b_{t}^{\ast} (\bar{H}_{t}, \underaccent{\bar}{S}_{t}, Y^{d}) = (\varphi_{1} (Y^{d}) I_{t}^{\top}, \varphi_{2} (Y^{d}) I_{t}^{\top}, \cdots, \varphi_{\xi} (Y^{d}) I_{t}^{\top})^{\top} \label{eq:sublinear1}
\end{equation}
where $\delta_{t} \equiv card(\mathcal{S}_{t} \times \cdots \times \mathcal{S}_{K}) \xi$ and the dependence on $\xi$ has been suppressed in the notation. Note $b_{t}^{\ast} (\bar{H}_{t}, \underaccent{\bar}{S}_{t}, Y^{d})$ does not actually depend on $\bar{H}_{t}$. Define the set $\Omega_{t}$ to be the set $\Gamma_{t}$ with $T_{b^{\ast}, t}$ (see equation \eqref{tbt-def} for its form) substituted for $T_{t}$. Clearly $\Omega_{t}$ is a subspace of $\mathcal{T}_{t}$, defined in Theorem \ref{lem:nuisance_nde}. In particular, we have the following important corollary.

\begin{corollary}
\label{cor:sub} 
Consider the space $\Omega = \sum_{t = 0}^{K} \Omega_{t}$. Then 
\begin{equation}
\Pi [\BU (\bm{q}, \Psi^{\ast}) | \Omega] = \sum_{t = 0}^{K} d_{t}^{\ast} (\bar{H}_{t}) T_{b^{\ast}, t} =: \sum_{t = 0}^{K} T_{b_{sub}, t} \label{eq:sublinear2}
\end{equation}
where $d_{t}^{\ast} (\bar{H}_{t})$ is defined as in equations \eqref{main_0} and \eqref{main_d1} but with $T_{t}$ replaced by $T_{b^{\ast}, t}$, and $b_{sub} = (b_{sub, 0}, \ldots, b_{sub, K})^{\top}$, where the dependence of $b_{sub}$ on $\mathbf{q}$ has been suppressed in the notation. In particular, when $Y$ is discrete and $\xi$ is the cardinality of the sample space of $Y$ then $\Lambda_{\mathsf{NDE}}^{\perp} = \Omega$ and consequently $\Pi [\BU (\bm{q}, \Psi^{\ast}) | \Lambda_{\mathsf{NDE}}^{\perp}] = \Pi [\BU (\bm{q}, \Psi^{\ast}) | \Omega] = \sum\limits_{t = 0}^{K} T_{b_{sub}, t}$.
\end{corollary}

\begin{remark}
\label{rem:sub} 
Consider the case where $Y$ is continuous and $\bm{\varphi} (Y)$ is the vector of the first $\xi$ elements of a complete basis for $L_{2} (\mu)$ with $\mu$ the Lebesgue measure. Then as $\xi \rightarrow \infty$, $\Pi [\BU (\bm{q}, \Psi^{\ast}) | \Omega] = \sum\limits_{t = 0}^{K} d_{t}^{\ast} (\bar{H}_{t})T_{b^{\ast}, t}$ should converge to $\Pi [\BU (\bm{q}, \Psi^{\ast}) | \Lambda_{\mathsf{NDE}}^{\perp}]$. As a consequence, by choosing $\xi \rightarrow \infty$ slowly with the sample size $n$, the asymptotic variance $V_{oracle} (\bm{q}, b_{sub})$ of $\hat{\Psi} (\bm{q}, b_{sub})$ should converge to the asymptotic variance  $V_{oracle} (\bm{q}, b_{opt})$ of $\hat{\Psi} (\bm{q}, b_{opt})$ and thus the oracle estimators $\hat{\Psi} (\bm{q}, b_{sub})$ and $\hat{\Psi} (\bm{q}, b_{opt})$ will be asymptotically equivalent.
\end{remark}

\section{Doubly Robust Nearly Efficient Estimation}
\label{sec:feasible}

In this section we consider the construction of feasible estimators of $\Psi^{\ast}$. To do so, we shall need to estimate the unknown conditional means and densities (hereafter nuisance functions) that are present in $\hat{\Psi} (\bm{q}, b)$, $\hat{\Psi} (\bm{q}, b_{sub})$, $\hat{\Psi} (\bm{q}, b_{opt})$ where the latter two estimators contain additional nuisance functions due to the dependence of $b_{opt}$ and $b_{sub}$ on on the distribution $P$. In the following to cover all cases we write the generic oracle estimator as $\hat{\Psi} (\bm{q}, b(P))$. We will consider the state-of-the-art cross-fit doubly robust machine learning (DR-ML) estimators \citep{chernozhukov2018double, smucler2019unifying} $\tilde{\Psi}_{cf} (\bm{q}, \hat{b})$ in which the nuisance functions are estimated by arbitrary machine learning algorithms chosen by the analyst. [If $b (P) = b$ does not depend on $P$, then $\hat{b} = b$.]

We need to use sample splitting because the functions estimated by modern machine learning algorithms (e.g. deep neural nets) have unknown statistical properties and, in particular, may not lie in a so-called Donsker class (see e.g. \citet[Chapter 2]{van1996weak}) - a condition generally needed for asymptotic linearity when we do not split the sample. Cross-fitting can recover the information lost due to sample splitting, provided that $\tilde{\Psi}_{cf} (\bm{q}, \hat{b})$ is asymptotically linear.

The following algorithm computes $\tilde{\Psi}_{cf} (\bm{q}, \hat{b})$:

\begin{itemize}
\item[(i)] The $n$ study subjects are randomly split into 2 parts: an estimation sample of size $n_{1}$ and a nuisance sample of size $n_{2} = n - n_{1}$ with $n_{1} / n \approx 1 / 2$.

\item[(ii)] Estimate all the unknown conditional expectation and density functions occurring in $\BU (\bm{q}, b(P), \Psi) \equiv \BU (\bm{q}, \Psi) - c_{OLS} (\bm{q}, b(P), \Psi)^{\top} T_{b(P)}$ from the nuisance sample data by machine learning. However, unconditional expectations can be estimated from the estimation sample.

\item[(iii)] Define $\hat{\BU} (\bm{q}, \hat{b}, \Psi)$ to be $\BU (\bm{q}, b(P), \Psi)$ except with the estimates under (ii) substituted for their estimands. Find the (assumed unique) solution $\tilde{\Psi}^{(1)} (\bm{q}, \hat{b})$ to $\BP_{n}^{(1)} [\hat{\BU} (\bm{q}, \hat{b}, \Psi)] = 0$ where $\BP_{n}^{(1)}$ is the sample average operator in the estimation sample.

\item[(iv)] Let $\tilde{\Psi}_{cf} (\bm{q}, \hat{b}) = \left\{ \tilde{\Psi}^{(1)} (\bm{q}, \hat{b}) + \tilde{\Psi}^{(2)} (\bm{q}, \hat{b}) \right\} / 2$ where $\tilde{\Psi}^{(2)} (\bm{q}, \hat{b})$ is $\tilde{\Psi}^{(1)} (\bm{q}, \hat{b})$ but with the \textit{training} and \textit{estimation} samples reversed.
\end{itemize}

In Appendix \ref{app:drml} we provide regularity conditions under which $\tilde{\Psi}_{cf} (\bm{q}, \hat{b})$ is RAL with the same influence function as $\hat{\Psi} (\bm{q}, b(P))$. Further we describe modified versions of $\tilde{\Psi}_{cf} (\bm{q}, \hat{b})$ that should exhibit better finite sample behavior. Finally we show that under the linear model $\gamma_{t} (\bar{H}_{t}, S_{t}, A_{t}; \Psi_{t}) = \Psi_{t}^{\top} V_{t}$, for a given vector $V_{t} = v_{t} (\bar{H}_{t}, S_{t}, A_{t})$, for each $t$, $\tilde{\Psi}_{cf} (\bm{q},b)$ exists in closed form and thus is simple to compute. A more detailed discussion of this strategy is given in Remark \ref{rem:close} in the Appendix and also see \citet{kallus2019localized} for an extended discussion and theoretical details.

We let $\tilde{\Psi}_{ss} (\bm{q}, \hat{b})$ denote the single sample estimator solving $\BP_{n} [\hat{\BU} (\bm{q}, \hat{b}, \Psi)] = 0$ where $\BP_{n}$ is the sample average operator in the entire sample and the unknown nuisance functions in $\BU (\bm{q}, b(P), \Psi)$ are estimated from all $n$ subjects. In the sequel we will use $\tilde{\Psi}$ to represent both estimators, as both are in common use, although the regularity conditions under which $\tilde{\Psi}_{ss} (\bm{q}, \hat{b})$ is RAL are more restrictive, generally requiring nuisance functions to lie in a Donsker class. For convenience, we will assume the necessary regularity conditions hold for any given $\tilde{\Psi}$ to be RAL, as regularity conditions are not the focus of the paper.

\section{Varieties of NDE assumptions -- substantive plausibility}
\label{sec:nde_variety}

In this section, we first define additional NDE assumptions that are stronger than the NDE$(Y^{d})$ assumption and discuss their differing statistical implications and substantive meaning. Let $C_{t} \subseteq L_{t}, t = 0, \ldots, K$, be a (possibly improper or null) subset of the covariates $L_{t}$. A stronger assumption than \eqref{nde_test} is that $\bar{A}_{K}$ has no direct effect on observed responses $\bar{C}_{K}$ and $Y^{d}$ not through treatment $\bar{S}_{K}$. Formally, this is the assumption that 
\begin{equation}
C_{t, \bar{a}_{t - 1}, \bar{s}_{t - 1}} = C_{t, \bar{a}_{t - 1}^{\prime}, \bar{s}_{t - 1}}, Y_{\bar{a}_{K}, \bar{s}_{K}}^{d} = Y_{\bar{a}_{K}^{\prime}, \bar{s}_{K}}^{d} \ \forall \bar{a}_{K}, \bar{a}_{K}^{\prime}, \bar{s}_{K}, t \label{Stronger_NDE}
\end{equation}
which we refer to as the NDE$(\bar{C}_{K}, Y^{d})$ assumption. Note, for any $C_{t}^{\prime} \subset C_{t}$, the NDE$(\bar{C}_{K}, Y^{d})$ assumption implies the NDE$(\bar{C}_{K}^{\prime})$, NDE$(\bar{C}_{K})$ and NDE$(Y^{d})$ assumptions, where the NDE$(\bar{C}_{K})$ assumption is given by equation \eqref{Stronger_NDE} without reference to $Y^{d}$. It follows from \citet{robins1999testing} that all the above results obtained under NDE$(Y^{d})$ hold under NDE$(\bar{C}_{K}, Y^{d})$ when we replace $D_{b, t}$, $T_{b, t}$, $\mathcal{B}_{t}$, and $\Lambda_{\mathsf{NDE}}^{\perp}$ with $D_{b, t}^{c}$, $T_{b, t}^{c}$, $\mathcal{B}_{t}^{c}$, and $\Lambda_{\mathsf{NDE}}^{c \perp}$ where the latter are defined like the former except with the set of functions $\{b_{t}^{c} (\bar{H}_{t}, \underaccent{\bar}{S}_{t}, \underaccent{\bar}{C}_{t + 1}, Y^{d})\}$ replacing the set $\{b_{t} (\bar{H}_{t}, \underaccent{\bar}{S}_{t}, Y^{d})\}$. Since the latter set is contained in the former, it follows that the asymptotic variance $V_{oracle} (\bm{q}, b_{opt}^{c})$ of $\hat{\Psi} (\bm{q}, b_{opt}^{c})$ is never greater and usually less than $V_{oracle} (\bm{q}, b_{opt})$ of $\hat{\Psi} (\bm{q}, b_{opt})$, where $b_{opt}^{c}$ is the minimizer of $V_{oracle} (\bm{q}, b^{c})$ over functions in $\mathcal{B}_{t}^{c}$.

We next briefly discuss, via examples, how one might assess whether a particular NDE assumption is substantively plausible. Without thinking carefully, most tend to assume that a test will not have an effect on an outcome not through the treatment $S_{t}$ under study. Therefore, it will be pedagogically most useful to give examples where this is not the case. In the HIV setting, a test may have an adverse effect on mortality $Y^{d}$ not through anti-retroviral therapy $S_{t}$, hence falsifying the NDE$(Y^{d})$ assumption, if poor test results (high HIV RNA and low CD4 count) motivate physicians to provide life-extending ancillary care such as prophylactic therapy against opportunistic infections or a list of avoidable behaviors likely to increase viral replication. This is true whether the ancillary care is generally (prophylaxis therapy) or rarely (a list of avoidable behaviors) recorded in the medical record.

Consider next an example in which a test has no effect on $Y^{d}$ except through $S_{t}$ but does have an effect on a covariate $C_{t}$. A cystoscopic exam $A_{t - 1}$ to evaluate recurrence of a bladder cancer can cause intense pain $C_{t} = 1$ in certain individuals. Pain itself has no effect on bladder cancer mortality $Y^{d}$, except by (possibly delaying) chemotherapy $S_{t}$. Hence NDE$(Y^{d})$ is true but NDE$(\bar{C}_{K}, Y^{d})$ is false. In this setting, an analyst who uses the estimator $\hat{\Psi} (\bm{q}, b_{opt})$ will succeed in gaining efficiency without incurring bias while an analyst who uses $\hat{\Psi} (\bm{q}, b_{opt}^{c})$ may introduce bias as $T_{b}^{c}$ may not have mean zero.

Finally, we consider an example wherein the assumption that `test at time $t$ has no direct effect on the results of later tests' is false. To formalize, we follow \citet{robins2008estimation} and assume the existence of an underlying variable $R_{t}^{\ast}$ that in conjunction with $A_{t - 1}$ determines $R_{t}$ by $R_{t} = A_{t - 1} R_{t}^{\ast} + (1 - A_{t - 1}) ?$. That is, if a test is performed at $t - 1$, the underlying variable $R_{t}^{\ast}$ is revealed; otherwise $R_{t}^{\ast}$ is hidden and recorded as $?$. Let $R_{t, \bar{a}_{t - 1}, \bar{s}_{t - 1}}^{\ast}$ and $R_{t, g}^{\ast}$ be the corresponding counterfactuals. We assume $R_{t}^{\ast}$ just precedes $R_{t}$ in our temporal ordering, as would be the case if a test (e.g. a cardiac stress test) ordered at visit $t - 1$ was both conducted and results reported at visit $t$. (With minor notational changes, we could have defined $R_{t} = A_{t - 1} R_{t - 1}^{\ast} + (1 - A_{t - 1})?$, with $R_{t - 1}^{\ast}$ just preceding $S_{t - 1}$, if the results of a test ordered and conducted at $t - 1$ are not available till $t$.) We say that the NDE$(\bar{R}_{K}^{\ast}, Y^{d})$ assumption of no effect of testing on $\bar{R}_{K}^{\ast}, Y^{d}$ holds if 
\begin{equation*}
R_{t, \bar{a}_{t - 1}, \bar{s}_{t - 1}}^{\ast} = R_{t, \bar{a}_{t - 1} \bar{s}_{t - 1}}^{\ast}, Y_{\bar{a}_{K}, \bar{s}_{K}}^{d} = Y_{\bar{a}_{K} \bar{s}_{K}}^{d}, \ \forall \bar{a}_{t}, \bar{a}_{t}^{\prime }, \bar{s}_{t - 1}, t.
\end{equation*}
Consider a study of an effect of a drug $S_{t}$ on brain amyloid content $Y^{d}$ in early Alzheimer's disease in which a test $A_{t}$ of (a secondary outcome) mental ability $R_{t + 1}^{\ast}$ is repeatedly administered to some but not other participants. Then, even when the NDE$(Y^{d})$ holds, the NDE$(\bar{R}_{K}^{\ast}, Y^{d})$ assumption will not hold if repetition has a direct effect on later test scores $R_{t}^{\ast}$ due to `practice effects'. That is, $R_{t, \bar{a}_{t - 1} = \bar{0}_{t - 1}, a_{t} = 1, \bar{s}_{t - 1}}^{\ast}$ will generally be less than $R_{t, \bar{a}_{t - 1} = \bar{1}_{t - 1}, a_{t} = 1, \bar{s}_{t - 1}}^{\ast}$.

Furthermore as in \citet{robins1999testing} we modify the identifying assumptions to incorporate $\bar{R}_{K}^{\ast}$, by (a) modifying the sequential exchangeability assumption both by adding the counterfactual $\underaccent{\bar}{R}_{t + 1, g}$ and conditioning on $\bar{R}_{t}^{\ast}$, (b) having $\bar{R}_{K}^{\ast}$ satisfy consistency, and (c) assuming $\bar{R}_{t}^{\ast} \amalg (S_{t}, A_{t}) | \bar{H}_{t}$ to insure that the modified ID assumptions both imply the unmodified ID assumptions and leave $\Pi_{m} \equiv \mathsf{Pr} [A_{m} = 1 | \bar{H}_{m}, S_{m}]$ and $p (s_{m} | \bar{H}_{m})$ unchanged, with no dependence on $\bar{R}_{m}^{\ast}$ except through $\bar{R}_{m}$.

Suppose the NDE$(\bar{R}_{K}^{\ast}, Y^{d})$ and modified ID assumptions hold. Then we would like to use the estimator $\hat{\Psi} (\bm{q}, b_{opt}^{r^{\ast}})$ where $T_{b_{opt}^{r^{\ast}}}^{r} \equiv \Pi [\BU (\bm{q}, \Psi^{\ast}) | \Lambda_{\mathsf{NDE}}^{r^{\ast} \perp}]$. However it is not possible to project onto $\Lambda_{\mathsf{NDE}}^{r^{\ast} \perp}$, because though $\underaccent{\bar}{R}_{t + 1}$ is observed, $\underaccent{\bar}{R}_{t + 1}^{\ast}$ is unobserved except for subjects with $\underaccent{\bar}{A}_{t + 1} = \underaccent{\bar}{1}_{t + 1}$. Instead, we project onto the ortho-complement of the tangent space of the induced observed data model. Specifically and, more generally, let $\mathbb{M}_{\mathsf{NDE}} (\bar{C}_{K}, \bar{R}_{K}^{\ast}, Y^{d})$ be the model defined by the NDE$(\bar{R}_{K}^{\ast}, \bar{C}_{K}, Y^{d})$ and modified ID assumptions. Let $\mathbb{M}_{\mathsf{NDE}, obs} (\bar{C}_{K}, \bar{R}_{K}^{\ast}, Y^{d})$ be the induced model for the observed data $O$. Let $\Lambda_{\mathsf{NDE}, obs}^{c, r^{\ast} \perp}$ be the observed data ortho-complement to the tangent space of $\mathbb{M}_{\mathsf{NDE}, obs} (\bar{C}_{K}, \bar{R}_{K}^{\ast}, Y^{d})$. In Appendix \ref{app:close_obs} we show that $\Lambda_{\mathsf{NDE}, obs}^{c, r^{\ast} \perp}$ is a strict subset of $\Lambda_{\mathsf{NDE}}^{c, r^{\ast} \perp}$, but a superset of $\Lambda_{\mathsf{NDE}}^{c \perp}$. In Theorem \ref{thm:obs_nde} and Corollary \ref{lem:obs_proj} of Appendix \ref{app:close_obs}, we provide closed-form expressions both for $\Lambda_{\mathsf{NDE}, obs}^{c, r^{\ast} \perp}$ and for the projection of any random variable onto a large subspace $\Omega_{obs}^{c, r^{\ast}}$ of $\Lambda_{\mathsf{NDE}, obs}^{c, r^{\ast} \perp}$ using Theorem \ref{thm:sub} and Corollary \ref{cor:sub}. Hence to gain maximum efficiency we use the algorithm of Section \ref{sec:feasible} to construct feasible versions of the oracle estimators $\hat{\Psi} (\bm{q}, b_{opt, obs}^{c, r^{\ast}})$ or $\hat{\Psi} (\bm{q}, b_{sub, obs}^{c, r^{\ast}})$ where $T_{b_{opt}^{c, r^{\ast}}, obs} = \Pi [\BU (\bm{q}, \Psi^{\ast}) | \Lambda_{\mathsf{NDE}, obs}^{c, r^{\ast} \perp}]$ and $T_{b_{sub}^{c, r^{\ast}}, obs} = \Pi [\BU (\bm{q}, \Psi^{\ast}) | \Omega_{obs}^{c, r^{\ast}}]$. 

\section{Efficiency and Intuition}
\label{sec:intuition}

Above we described how to construct novel highly efficient estimators under various NDE assumptions. In this section we provide guidance and intuition as to when novel estimators will achieve small versus large gains in efficiency compared to estimators that do not rely on any NDE assumption. To do so, we begin by discussing the paper of \citet{caniglia2019emulating}. The authors of that paper analyzed data from an observational study that followed 41,724 HIV infected individuals on first line anti-retroviral therapy (ART) whose HIV RNA had been suppressed to less than 200 copies/ml at baseline. They estimated the optimal testing and treatment strategy among four candidate regimes. Testing costs were ignored. The utility $Y = Y^{d}$ was the indicator of survival at $60$ months of follow-up; hence the expected counterfactual utility $\theta_{g} \equiv \E [Y_{g}]$ is the probability of surviving to $60$ months under regime $g$. See \citet{robins2008estimation} and later Remark \ref{rem:cost} for the case of costly tests. Data were recorded at monthly intervals and the ID assumptions of Section \ref{sec:notation} were assumed to hold. Their analysis compared the standard IPW estimator $\tilde{\theta}_{ipw, g}$ with a novel IPW estimator $\tilde{\theta}_{nde\mbox{-}ipw, g}$ introduced in \citet{robins2008estimation}. Below we show that the latter is a consistent estimator of $\theta_{g}$ under the NDE$(\bar{R}_{K}^{\ast}, Y^{d})$ assumption, but not under the weaker NDE$(Y^{d})$ assumption; the former is consistent regardless of whether any NDE assumption holds. They found $\tilde{\theta}_{nde\mbox{-}ipw, g}$ to be 50 times as efficient as the usual IPW estimator $\tilde{\theta}_{ipw, g}$. In this subsection we describe results in \citet{caniglia2019emulating} and provide an intuitive understanding of the reason for such a remarkable efficiency gain.

\citet{caniglia2019emulating} compared $\theta_{g}$ for four candidate regimes $g$ formed by crossing two testing with two treatment regimes. For reasons of both pedagogy and space, in the main text, we describe a somewhat simplified version of these regimes. Appendix \ref{app:actual} compares these approximate regimes with the actual regimes used for data analysis. The two testing regimes $g_{x}$ require simultaneous tests for both HIV RNA and CD4 count at baseline and then every 11 months while RNA $<$ 200 copies/ml and CD4 $>x$ counts/ml for $x \in \{350, 500\}$; otherwise the tests are performed every 5 months. The two treatment regimes $g_{z}$ require one to switch to second line ART therapy 2 months after the first time the measured RNA level exceeds $z$ copies for $z \in \{200, 1000\}$. Let $g_{x, z} = (g_{x}, g_{z}), x \in \{350, 500\}, z \in \{200, 1000\}$ denote the four regimes.

At baseline all 41,724 subjects can follow each of the 4 regimes. To avoid unnecessary clutter, we henceforth restrict attention to a comparison of the regimes $g_{x = 350, z = 200}$ and $g_{x = 500, z = 200}$. Consider first an analysis that does not impose the NDE assumption. A subject is censored under (i.e. stops following) regime $g_{x, z}$ at the first month $t$ that either the subject's observed treatment indicator $S_{t}$ or testing indicator $A_{t}$ is incompatible with $g_{x, z}$, meaning that either $S_{t, g} \neq S_{t}$ or $A_{t, g} \neq A_{t}$. For regime $g_{x = 500, z = 200}$, of the 41,724 initially following the regime, the number still following the regime (i.e. still uncensored) at month 12, 24, 36, 48, 60 were 8926, 2634, 949, 348, 152. The corresponding numbers 4806, 1006, 204, 94, 45 for $g_{x = 350, z = 200}$ were even smaller. The explanation for the rapid decrease in the number of uncensored subjects is that most individuals were tested more frequently in the observed data than is allowed by either of the two testing regimes under consideration. In fact, even under regime $g_{x = 500, z = 200}$, well over 95\% of the censoring events were due to early retest (i.e. $A_{t, g} = 0, A_{t} = 1$).

\citet{caniglia2019emulating} constructed a so-called NDE dataset in which subjects are no longer censored from regime $g_{x, z}$ for reasons of overly frequent testing, thereby increasing the number of uncensored subjects by roughly one to two orders of magnitude. The estimator $\tilde{\theta}_{nde\mbox{-}ipw, g_{x, z}}$ is computed from the uncensored subjects in the NDE dataset and thus would be expected to have a much smaller variance than $\tilde{\theta}_{ipw, g_{x, z}}$. In fact, \citet{caniglia2019emulating} reported the nonparametric bootstrap variance estimators of the two different IPW estimators of the difference $\theta_{g_{x = 350, z = 200}} - \theta_{g_{x = 500, z = 200}}$. The nonparametric bootstrap estimator of the variance of $\tilde{\theta}_{ipw, g_{x = 350, z = 200}} - \tilde{\theta}_{ipw, g_{x = 500, z = 200}}$ was $1.69 \times 10^{-4}$ compared to $3.22 \times 10^{-6}$ for the variance of $\tilde{\theta}_{nde\mbox{-}ipw, g_{x = 350, z = 200}} - \tilde{\theta}_{nde\mbox{-}ipw, g_{x = 500, z = 200}}$, indicating a relative efficiency in favor of the latter of 53 fold. However, as noted above, if the relevant NDE assumptions fail to hold, $\tilde{\theta}_{nde\mbox{-}ipw, g_{x, z}}$ is inconsistent for $\theta_{g_{x, z}}$.

\subsection{Definitions and Statistical Properties of $\tilde{\theta}_{ipw, g}$ and $\tilde{\theta}_{nde\mbox{-}ipw, g}$}
\label{sec:formal}

Prior to the current paper the estimator $\tilde{\theta}_{nde\mbox{-}ipw, g}$ was the only estimator in the literature that leveraged the NDE assumptions to increase efficiency. Thus it is of interest to understand how $\tilde{\theta}_{nde\mbox{-}ipw, g}$ is related to the estimators proposed in the current paper that leverage NDE assumptions by subtracting from an inefficient estimator the projection of its influence function onto the ortho-complement to the tangent space of an NDE model. In this section we study this relationship after we formally define $\tilde{\theta}_{nde\mbox{-}ipw, g}$. Recall $W_{t} \equiv \prod\limits_{m = t}^{K} p (S_{m} | \bar{H}_{m})$ and $\Pi_{m} \equiv \Pr (A_{m} = 1 | \bar{H}_{m}, S_{m})$ and let $\hat{W}_{t}$ and $\hat{\Pi}_{m}$ be corresponding estimates. We ignore testing costs so $Y = Y^{d}$. Under the ID assumptions and regardless of whether the NDE assumption holds, $\theta_{g} = \E [Y_{g}]$ is identified by the IPW representation $\E [V_{ipw, g}]$ of the g-formula where 
\begin{equation*}
V_{ipw, g} \equiv \frac{\prod\limits_{t = 0}^{K} \mathbbm{1}\{A_{t} = A_{t, g}, S_{t} = S_{t, g}\} Y}{\prod\limits_{t = 0}^{K} \Pr (A_{t} = A_{t, g} | \bar{H}_{t}, S_{t}) W_{0}} = \prod\limits_{t = 0}^{K} \left\{ \frac{\mathbbm{1} \{A_{t} = 0\}}{1 - \Pi_{t}} \right\}^{1 - A_{t, g}} \prod\limits_{t = 0}^{K} \left\{ \frac{\mathbbm{1} \{A_{t} = 1\}}{\Pi_{t}} \right\}^{A_{t, g}} \frac{\prod\limits_{t = 0}^{K} \mathbbm{1} \{S_{t} = S_{t, g}\}}{W_{0}} Y.
\end{equation*}
The quantities $\hat{\theta}_{ipw, g} = \mathbb{P}_{n} [V_{ipw, g}]$ and $\tilde{\theta}_{ipw, g} = \mathbb{P}_{n} [\hat{V}_{ipw, g}]$ are, respectively, the oracle IPW estimator and the IPW estimator of $\theta_{g}$, where, for any `oracle' random vector $Q$ depending on unknown conditional expectations and densities, $\hat{Q}$ replaces these unknown functions by estimates. By consistency and $g = (g_{a}, g_{s})$ having both the testing and treatment regimes $g_{a}$ and $g_{s}$ deterministic, one can easily show that $V_{ipw, g}$ is a function of the observed data assuming the oracle knowledge that all the nuisance functions are known. [Hence $\hat{V}_{ipw, g}$ is a statistic without such oracle knowledge.] Furthermore, a necessary condition for $V_{ipw, g}$ to be non-zero is that $\bar{H}_{t, g} = \bar{H}_{t}$ for $t = 0, \ldots, K$, and $S_{K} = g_{s, K} (\bar{H}_{K})$. A subject is said to have been censored under regime $g$ if and only if for some $t \leq K$ at least one of the following 4 events happens: $[S_{t} = 0, S_{t, g} = 1]$, $[S_{t} = 1, S_{t, g} = 0]$, $[A_{t} = 0, A_{t, g} = 1]$, and $[A_{t} = 1, A_{t, g} = 0]$. Hence all censored subjects have $V_{ipw, g} = 0$.

\citet{robins2008estimation} introduced the NDE-IPW estimator $\tilde{\theta}_{nde\mbox{-}ipw, g} = \BP_{n} [\hat{V}_{nde\mbox{-}ipw, g}]$ with oracle version $\hat{\theta}_{nde\mbox{-}ipw, g} = \BP_{n} [V_{nde\mbox{-}ipw, g}]$, where 
\begin{equation*}
V_{nde\mbox{-}ipw, g} = \prod\limits_{t = 0}^{K} \left\{ \frac{\mathbbm{1} \{A_{t} = 1\}}{\Pi_{t}} \right\}^{A_{t, g}} \frac{\prod\limits_{t = 0}^{K} \mathbbm{1} \{S_{t} = S_{t, g}\}}{W_{0}} Y.
\end{equation*}
Note $V_{nde\mbox{-}ipw, g}$ differs from $V_{ipw, g}$ only in that we have removed the first factor. A subject who experiences the event $[A_{t} = 1, A_{t, g} = 0]$ need not have $V_{nde\mbox{-}ipw, g} = 0$. Therefore we will say a subject is NDE-censored under $g$ (and thus not included in the NDE dataset under regime $g$) if and only if, for some $t \leq K$, at least one of the three NDE censoring events $[S_{t} = 0, S_{t, g} = 1]$, $[S_{t} = 1, S_{t, g} = 0]$, $[A_{t} = 0, A_{t, g} = 1]$ happens.

\begin{remark}
\label{rem:cost} 
If we wish to incorporate testing costs so $Y$ need not equal $Y^{d}$, it is necessary to replace $Y$ in the definition of $V_{nde\mbox{-}ipw, g}$ by $Y_{g} = Y + \sum\limits_{t = 0}^{K - 1} c^{\ast} \mathbbm{1} [A_{t} = 1, A_{t, g} = 0]$ where $Y_{g}$ adds to $Y$ the number of times $t$ that a subject experiences the event $[A_{t} = 1, A_{t, g} = 0]$ multiplied by the cost $c^{\ast}$ of the test. To see why, note that, for non-NDE-censored subjects under regime $g$, the observed total utility $Y$ is less than the utility $Y_{g}$ under regime $g$ by $c^{\ast}$ times the difference in the number of tests $\sum\limits_{t = 0}^{K - 1} \mathbbm{1} [A_{t} = 1, A_{t, g} = 0]$. All the results given below continue to hold when testing costs are included if we substitute $Y_{g}$ for $Y$ in the definitions of $V_{nde\mbox{-}ipw, g}$ and $V_{ipw, g}$.
\end{remark}

There is a subtle, but critical point we have glossed over: assuming that the nuisance functions are known, it is not immediately clear that $V_{nde\mbox{-}ipw, g}$ itself is a statistic (i.e. a function of the observed data $O$ with probability 1) because, unlike for $V_{ipw, g}$, an appeal to consistency does not suffice (since if the event $[A_{t - 1} = 1, A_{t - 1, g} = 0]$ occurs, then $\bar{H}_{t^{\prime}, g} \neq \bar{H}_{t^{\prime}}$ for $t^{\prime} \geq t)$ and $\bar{H}_{t^{\prime}, g}$ may then not be a function of the observed data. Given a testing and treatment regime $g$ of interest, the following lemma, proved in Appendix \ref{app:lem_err}, provides sufficient conditions for $V_{nde\mbox{-}ipw, g}$ to be a statistic. As in Section \ref{sec:nde_variety}, we assume there exists an underlying variable $R_{t + 1}^{\ast}$ that is revealed only if a test is performed, i.e. $R_{t + 1} = A_{t} R_{t + 1}^{\ast} + (1 - A_{t}) ?$.

\begin{lemma}
\label{lem:nde-ipw-1} 
Suppose for regime $g$ (i) $g_{t}(\bar{h}_{t})$ depends on $\bar{h}_{t}$ through and only through the testing results $\bar{r}_{t}$, covariates $\bar{c}_{t}$, and treatments $\bar{s}_{t - 1}$; equivalently $g_{t} (\bar{h}_{t}) = g_{t}^{\dag} (\bar{h}_{t}^{\dag})$ for known $g_{t}^{\dag}$ where $\bar{h}_{t}^{\dag} = (\bar{s}_{t - 1}, \bar{c}_{t}, \bar{r}_{t})$. Then, if (ii) the NDE$(\bar{C}_{K}, \bar{R}_{K}^{\ast})$ assumption holds, $V_{nde\mbox{-}ipw, g}$ is a statistic.
\end{lemma}

\begin{remark}\leavevmode
\label{rem:6}
\begin{enumerate}
\item We could have included $\bar{a}_{t - 1}$ in $\bar{h}_{t - 1}^{\dag}$ but it would have been redundant since $\bar{a}_{t - 1}$ is deterministic function of $\bar{r}_{t}$, as $r_{m} = ?$ implies $a_{m - 1} = 0$ and $r_{m} \neq ?$ implies $a_{m - 1} = 1$ for $1 \leq m \leq K$. Had we included $\bar{a}_{t - 1}$ in $\bar{h}_{t - 1}^{\dag}$, then $\bar{h}_{t} = (\bar{h}_{t}^{\dag}, \bar{l}_{t-1} \setminus \bar{c}_{t})$ would differ from $\bar{h}_{t}^{\dag}$ only through the components $\bar{l}_{t} \setminus \bar{c}_{t}$ of $\bar{l}_{t}$ not included in $\bar{c}_{t}$. Note further the Lemma implies that, if the NDE$(\bar{C}_{K}, \bar{R}_{K}^{\ast})$ assumption holds, then, for non-NDE-censored subjects, $Y_{g} = Y + \sum\limits_{t = 0}^{K - 1} \mathbbm{1} [A_{t} = 1, A_{t, g} = 0] c^{\ast} (t, \bar{C}_{t}, \bar{R}_{t}, \bar{S}_{t})$ for cost function $c^{\ast} (t, \bar{c}_{t}, \bar{r}_{t}, \bar{s}_{t})$ is a statistic and hence generalizes Remark \ref{rem:cost} by allowing $c^{\ast} (t, \bar{C}_{t}, \bar{R}_{t}, \bar{S}_{t})$ to replace $c^{\ast}$. See Remark \ref{rem:app_cost} in Appendix \ref{app:lem_err} for additional discussion.

\item To understand the need for the NDE$(\bar{R}_{K}^{\ast})$ assumption in Lemma \ref{lem:nde-ipw-1}, suppose the assumption fails to hold and consider the following counterexample. Consider a regime $g$ defined by $S_{0, g} = 0$, $A_{0, g} = 0$, $S_{1, g} = 1$, $A_{1, g} = 1$, $S_{2, g} = g_{s, 2} (R_{2, g})$. Then, for a subject with observed data $S_{0} = 0$, $A_{0} = 1$, $S_{1} = 1$, $A_{1} = 1$, $R_{2} \equiv R_{2, a_{0} = 1, a_{1} = 1, s_{0} = 0, s_{1} = 1}^{\ast} = 1$, $S_{2} = 1$, we cannot compute $R_{2, g} \equiv R_{2, a_{0} = 0, a_{1} = 1, s_{0} = 0, s_{1} = 1}^{\ast}$ and thus $S_{2, g}$ from the observed data because $A_{0} \neq A_{0, g}$. Hence $\hat{V}_{nde\mbox{-}ipw, g}$ is not a statistic. However if the NDE$(\bar{R}_{K}^{\ast})$ assumption holds, $R_{2, g} = R_{2} = R_{2, s_{0} = 0, s_{1} = 1}^{\ast}$ and thus $\hat{V}_{nde\mbox{-}ipw, g}$ is a statistic. An analogous counterexample can be constructed when the NDE$(\bar{C}_{K})$ assumption fails to hold.
\end{enumerate}
\end{remark}

The next lemma gives conditions under which $V_{nde\mbox{-}ipw, g}$ is unbiased for $\theta_{g} = \E [Y_{g}]$.

\begin{lemma}
\label{lem:nde-ipw-2} 
Suppose the modified ID assumptions of Section \ref{sec:nde_variety} hold and $g_{t}$ depends on and only on $\bar{r}_{t}$, $\bar{c}_{t}$, and $\bar{s}_{t - 1}$. Then the NDE$(\bar{C}_{K}, \bar{R}_{K}^{\ast}, Y^{d})$ assumption implies that $\E [V_{nde\mbox{-}ipw, g}] = \theta_{g}$ and, thus, also that $\tilde{\theta}_{nde\mbox{-}ipw, g} = \mathbb{P}_{n} [\hat{V}_{nde\mbox{-}ipw, g}]$ is a RAL estimator of $\theta_{g}$ under regularity conditions.
\end{lemma}

\begin{remark}
It is of interest that the NDE$(\bar{C}_{K}, \bar{R}_{K}^{\ast})$ assumption was already required for $\hat{V}_{nde\mbox{-}ipw, g}$ to even be a statistic. \citet{robins2008estimation} proved Lemma \ref{lem:nde-ipw-2} under the stronger NDE$(\bar{L}_{K}, \bar{R}_{K}^{\ast}, Y^{d})$ assumption; inspection of their proof shows they actually proved Lemma \ref{lem:nde-ipw-2}; in fact, they showed that Lemma \ref{lem:nde-ipw-2} remains true without requiring $\Pi_{m} \equiv \Pr (A_{m} = 1 | \bar{H}_{m}, S_{m})$ be bounded away from 1.
\end{remark}

Here we offer an alternate proof of this lemma that reveals the relationship of $V_{nde\mbox{-}ipw, g}$ to our earlier discussions of semiparametric models and their tangent spaces. Corollary \ref{cor:nde-ipw-3} below shows that Lemma \ref{lem:nde-ipw-2} follows directly from Lemma \ref{lem:nde-ipw-3} below. In our proof of Lemma \ref{lem:nde-ipw-3} in Appendix \ref{app:lem_err}, we show that the estimator $\hat{\theta}_{nde\mbox{-}ipw, g}$ (and its feasible version $\tilde{\theta}_{nde\mbox{-}ipw, g}$) is a member of a much larger class of estimators to which the results we obtain for $\hat{\theta}_{nde\mbox{-}ipw, g}$ (and its feasible version $\tilde{\theta}_{nde\mbox{-}ipw, g}$) also apply.

The model $\mathbb{M}_{np, obs}$ for the observed data defined by the modified ID assumptions has tangent space $\Lambda = L_{2, 0} (P)$, as the model places no restrictions on the distribution of $O$. It follows that the functional $\theta_{g} = \E [V_{ipw, g}]$ has a unique influence function $\IF_{ipw, g} \equiv V_{ipw, g} - \Pi [V_{ipw, g} | \Lambda_{\mathsf{ancillary}}] - \theta_{g}$ in model $\mathbb{M}_{np, obs}$, which is thus also an influence function in the smaller model $\mathbb{M}_{\mathsf{NDE}, obs} (\bar{C}_{K}, \bar{R}_{K}^{\ast}, Y^{d})$.

\begin{lemma}
\label{lem:nde-ipw-3} 
Under the assumptions of Lemma \ref{lem:nde-ipw-2}, $\IF_{ipw, g} - \theta_{g} - \left\{ V_{nde\mbox{-}ipw, g} - \Pi [V_{nde\mbox{-}ipw, g} | \Lambda_{\mathsf{ancillary}}] \right\} \in \Lambda_{\mathsf{NDE}, obs}^{c, r^{\ast} \perp}$.
\end{lemma}

The corollary below then follows from the fact that in any semiparametric model the sum of any element of the ortho-complement to the tangent space and an influence function is itself an influence function.

\begin{corollary}
\label{cor:nde-ipw-3} 
Under the assumptions of Lemma \ref{lem:nde-ipw-2}, $\IF_{nde\mbox{-}ipw, g}^{c, r^{\ast}} \equiv V_{nde\mbox{-}ipw, g} - \Pi [V_{nde\mbox{-}ipw, g} | \Lambda_{\mathsf{ancillary}}] - \theta_{g}$ is an influence function in model $\mathbb{M}_{\mathsf{NDE}, obs} (\bar{C}_{K}, \bar{R}_{K}^{\ast}, Y^{d})$ and thus $V_{nde\mbox{-}ipw, g}$ has mean $\theta_{g}$.
\end{corollary}

This following remark summarizes the implications of the above results. 

\begin{remark}
\leavevmode

\begin{enumerate}
\item Lemma \ref{lem:nde-ipw-3} fails to hold if we replace $\Lambda_{\mathsf{NDE}, obs}^{c, r^{\ast} \perp}$ with the ortho-complement of any model for which the NDE assumption only holds for a strict subset of $\{\bar{C}_{K}, \bar{R}_{K}^{\ast}, Y^{d}\}$, since $V_{nde\mbox{-}ipw, g} - \Pi [V_{nde\mbox{-}ipw, g} | \Lambda_{\mathsf{ancillary}}]$ explicitly depends on each of $Y^{d}, \bar{C}_{K}$, and $\bar{R}_{K}^{\ast}$. Indeed, unless the NDE$(\bar{C}_{K}, \bar{R}_{K}^{\ast})$ assumption holds, $V_{nde\mbox{-}ipw, g} - \Pi [V_{nde\mbox{-}ipw, g} | \Lambda_{\mathsf{ancillary}}]$ is not a function of the observed data. (So neither is $\hat{V}_{nde\mbox{-}ipw, g} - \hat{\Pi} [\hat{V}_{nde\mbox{-}ipw, g} | \Lambda_{\mathsf{ancillary}}]$, where $\hat{\Pi}$ denotes the projection onto the space $\Lambda_{\mathsf{ancillary}}$ with all nuisance functions replaced by their estimates.)

\item Next suppose, we are only willing to assume model $\mathbb{M}_{\mathsf{NDE}}(Y^{d})$ that just imposes the NDE$(Y^{d})$ and modified ID assumptions. Our goal is again to estimate $\theta_{g}$ for a regime $g_{t}$ that depends on $\bar{r}_{t}$, $\bar{c}_{t}$, and $\bar{s}_{t - 1}$ so $\hat{\theta}_{nde\mbox{-}ipw, g}$ (and hence $\tilde{\theta}_{nde\mbox{-}ipw, g}$) is not a statistic. Nevertheless we can still improve upon the efficiency of $\tilde{\theta}_{ipw, g} = \BP_{n} [\hat{V}_{ipw, g}]$. We first give the oracle version of the procedure. Let $U_{ipw, g} (\theta) = V_{ipw, g} - \Pi [V_{ipw, g} | \Lambda_{\mathsf{ancillary}}] - \theta_{g}$ so $U_{ipw, g} (\theta_{g}) = \IF_{ipw, g}$ is an influence function for $\theta_{g}$ in $\mathbb{M}_{\mathsf{NDE}} (Y^{d})$. It follows that $\hat{\theta}_{ipw, g} (b_{opt, ipw, g})$ solving $\BP_{n} [U_{ipw} (\theta, b_{opt, ipw, g})] \equiv \BP_{n} [U_{ipw, g} (\theta) - \Pi [U_{ipw, g} (\theta) | \Lambda_{\mathsf{NDE}}^{\perp}]] = 0$ is semiparametric efficient in model in $\mathbb{M}_{\mathsf{NDE}}(Y^{d})$, as its influence function is $\EIF = \IF_{ipw, g} - \Pi [\IF_{ipw, g} | \Lambda_{\mathsf{NDE}}^{\perp}]$ under the model. Note that 
\begin{equation*}
U_{ipw, g} (\theta) - \Pi [U_{ipw, g} (\theta) | \Lambda_{\mathsf{NDE}}^{\perp}] = V_{ipw, g} - \theta - \Pi [V_{ipw, g} | \Lambda_{\mathsf{ancillary}}] - \Pi [V_{ipw, g} | \Lambda_{\mathsf{NDE}}^{\perp}],
\end{equation*}
as neither $\Pi [V_{ipw, g} | \Lambda_{\mathsf{ancillary}}]$ nor $\theta$ has a projection on $\Lambda_{\mathsf{NDE}}^{\perp}$. Hence $\hat{\theta}_{ipw, g} (b_{opt, ipw, g}) = \BP_{n} \{V_{ipw, g} - \Pi [V_{ipw, g} | \Lambda_{\mathsf{ancillary}}] - \Pi [V_{ipw, g} | \Lambda_{\mathsf{NDE}}^{\perp}]\}$. We can obtain a feasible version $\tilde{\theta}_{ipw, g} (\hat{b}_{opt, ipw, g})$ of $\hat{\theta}_{ipw, g} (b_{opt, ipw, g})$ using the analogue of the algorithm proposed in Section \ref{sec:feasible}.

\item Suppose we again assume model $\mathbb{M}_{\mathsf{NDE}, obs} (\bar{C}_{K}, \bar{R}_{K}^{\ast}, Y^{d})$ is true so both $\tilde{\theta}_{ipw, g} = \BP_{n} [\hat{V}_{ipw, g}]$ and $\tilde{\theta}_{nde\mbox{-}ipw, g} = \BP_{n} [\hat{V}_{nde\mbox{-}ipw, g}]$ are consistent for $\theta_{g}$. Furthermore  $\tilde{\theta}_{ipw, g} - \hat{\Pi} [V_{ipw, g} | \Lambda_{\mathsf{ancillary}}]$ and $\tilde{\theta}_{nde\mbox{-}ipw, g} - \hat{\Pi} [V_{nde\mbox{-}ipw, g} | \Lambda_{\mathsf{ancillary}}]$ are RAL with the influence functions $\IF_{ipw, g}$ and $\IF_{nde\mbox{-}ipw, g}^{c, r^{\ast}}$ for $\theta_{g}$. Recall that in \citet{caniglia2019emulating} $\tilde{\theta}_{nde\mbox{-}ipw, g}$ was 50 times as efficient as $\tilde{\theta}_{ipw, g}$. Now the asymptotic variances of the single sample estimators $\tilde{\theta}_{ipw, g}$ and $\tilde{\theta}_{nde\mbox{-}ipw, g}$ used by \citet{caniglia2019emulating} approximate $\var (\IF_{ipw, g})$ and $\var (\IF_{nde\mbox{-}ipw, g}^{c, r^{\ast}})$, because flexible, high dimensional models were used in estimating $\Pi_{t}$ and $W_{0}$. But (i) neither $\var (\IF_{ipw, g})$ nor $\var (\IF_{nde\mbox{-}ipw, g}^{c, r^{\ast}})$ dominates the other at all laws in the model and both exceed the semiparametric variance bound of model $\mathbb{M}_{\mathsf{NDE}, obs} (\bar{C}_{K},\bar{R}_{K}^{\ast },Y^{d})$. Rather the oracle estimators $\hat{\theta}_{ipw, g} (b_{opt, ipw, g}^{c, r^{\ast}}) = \BP_{n} \{V_{ipw, g} - \Pi [V_{ipw, g} | \Lambda_{\mathsf{ancillary}}] - \Pi [V_{ipw, g} | \Lambda_{\mathsf{NDE}, obs}^{c, r^{\ast} \perp}]$ solving $\BP_{n} [U_{ipw, g} (\theta) - \Pi [U_{ipw, g} (\theta) | \Lambda_{\mathsf{NDE}, obs}^{c, r^{\ast} \perp}]] = 0$ and $\hat{\theta}_{nde\mbox{-}ipw, g} (b_{opt, nde\mbox{-}ipw, g}^{c, r^{\ast}}) = \BP_{n} \{V_{nde\mbox{-}ipw, g} - \Pi [V_{nde\mbox{-}ipw, g} | \Lambda_{\mathsf{ancillary}}] - \Pi [V_{nde\mbox{-}ipw, g} | \Lambda_{\mathsf{NDE}, obs}^{c, r^{\ast} \perp}]\}$ solving $\BP_{n} [U_{nde\mbox{-}ipw, g} (\theta) - \Pi [U_{nde\mbox{-}ipw, g} (\theta) | \Lambda_{\mathsf{NDE}, obs}^{c, r^{\ast} \perp}]] = 0$ are semiparametric efficient in $\mathbb{M}_{\mathsf{NDE}, obs} (\bar{C}_{K}, \bar{R}_{K}^{\ast}, Y^{d})$, with common influence function $\EIF_{obs}^{c, r^{\ast}} = \IF_{ipw, g} - \Pi [\IF_{ipw, g} | \Lambda_{\mathsf{NDE}, obs}^{c, r^{\ast} \perp}] = \IF_{nde\mbox{-}ipw, g}^{c, r^{\ast}} - \Pi [\IF_{nde\mbox{-}ipw, g}^{c, r^{\ast}} | \Lambda_{\mathsf{NDE}, obs}^{c, r^{\ast} \perp}]$. Again the algorithm in Section \ref{sec:feasible} can be used to obtain feasible efficient estimators $\tilde{\theta}_{ipw, g} (\hat{b}_{opt, ipw, g}^{c, r^{\ast}})$ and $\tilde{\theta}_{nde\mbox{-}ipw, g} (\hat{b}_{opt, nde\mbox{-}ipw, g}^{c, r^{\ast}})$ with influence function $\EIF_{obs}^{c, r^{\ast}}$.
\end{enumerate}
\end{remark}

\section{Simulation studies}
\label{sec:simulation}

\subsection{A simple data generating process related to \citet{caniglia2019emulating}}

This simulation study explores in greater depth data generating processes (DGPs) under which $\tilde{\theta}_{nde\mbox{-}ipw, g}$ is much more efficient than $\tilde{\theta}_{ipw, g}$. The study also relates these estimators to estimators of the optimal value function $\E [Y_{g^{opt}}]$ based on opt-SNMMs considered in Section \ref{sec:estimator}.

We consider a much simplified version of the HIV study of \citet{caniglia2019emulating} in which $K = 1$ with time-ordered random variables: $A_{0}, R_{1}^{\ast}, R_{1}, S_{1}, Y^{d}, Y$. Here $A_{0}$ is the indicator of having an HIV RNA test; $R_{1}$ is the test result: $R_{1}$ is $?$ if $A_{0} = 0$ and is $R_{1}^{\ast}$ if $A_{0} = 1$; $R_{1}^{\ast}$ is the CD4 count dichotomized as $R_{1}^{\ast} = 0$ if low and $R_{1}^{\ast} = 1$ otherwise; $S_{1}$ is the indicator of treatment (i.e. of switching to second line anti-retroviral therapy); $Y^{d}$ is a health utility that is a decreasing function of HIV RNA levels at time $K + 1$ with $K + 1 = 2$; $Y \coloneqq Y^{d} - c^{\ast} A_{0}$ with $c^{\ast} = 3$ is the overall observed utility incorporating the cost of testing. Since we are including the cost of testing we replace $Y$ by $Y_{g} = Y + \mathbbm{1} [A_{0} = 1, A_{0, g} = 0]$ in defining $\tilde{\theta}_{nde\mbox{-}ipw, g}$ and $\tilde{\theta}_{ipw, g}$ as discussed in Remark \ref{rem:cost}. $Y_{g}$ is just the health utility $Y^{d}$ for any regime with $A_{0, g} \equiv 0$.

Our DGP was designed such that (i) the NDE$(R_{1}^{\ast },Y^{d})$ and modified ID assumptions hold and (ii) $g^{opt}=(a_{0}^{opt}=0,s_{1}^{opt}=1)$ and hence $\theta _{g^{opt}}=\E[Y_{a_{0}=0,s_{1}=1}]\equiv \E [Y_{a_{0}=0,s_{1}=1}^{d}]$. Simulations with more complex DGPs with $g^{opt}$ depending on $R$ can be found in Appendix \ref{sim:further}.

The data generating process (DGP) is:

\begin{itemize}
\item $A_{0} \sim \text{Bernoulli} (\rho)$, with $\rho \in \{0.1, 0.2, \ldots, 0.8, 0.9, 0.95, 0.99\}$;

\item $R_{1}^{\ast} \sim \text{Bernoulli} (0.5)$;

\item $R_{1} = R_{1}^{\ast}$ if $A_{0} = 1$ and ? otherwise;

\item $S_{1} \sim \text{Bernoulli} (0.15) \mathbbm{1} \{A_{0} = 1, R_{1} = 0\} + \text{Bernoulli} (0.85) \mathbbm{1} \{A_{0} = 1, R_{1} = 1\} + \text{Bernoulli} (0.15) \mathbbm{1} \{A_{0} = 0, R_{1} = ?\}$;

\item $Y^{d} \sim N(-5 R_{1}^{\ast} + 2 S_{1} R_{1}^{\ast} - 0.1S_{1} (1 - R_{1}^{\ast}), 1)$;

\item $Y \coloneqq Y^{d} - c^{\ast} A_{0}$ with $c^{\ast} = 3$.
\end{itemize}

Since $A_{0}$ is neither a direct cause of $Y$ nor of $R_{1}^{\ast}$, the NDE$(R_{1}^{\ast}, Y^{d})$ assumption holds. The sample size was $n = 2.5 \times 10^{4}$ or $5 \times 10^{4}$. We computed the ``truth'' based on a single simulated dataset with sample size $10^{7}$. \textit{Monte Carlo} summary statistics are based on 100 replications. The R codes for the simulation studies are provided in the online supplementary materials.

We note several other properties of our DGP. Whenever $K = 1$, the explicit representation of $\Lambda _{\mathsf{NDE}, obs}^{r^{\ast} \perp}$ in Appendix \ref{app:close_obs} shows that the ortho-complements to the observed data tangent spaces in models $\mathbb{M}_{\mathsf{NDE}} (Y^{d})$ and $\mathbb{M}_{\mathsf{NDE}, obs} (\bar{R}_{K}^{\ast}, Y^{d})$ are equal, i.e. $\Lambda_{\mathsf{NDE}}^{\perp} = \Lambda_{\mathsf{NDE}, obs}^{r^{\ast} \perp}$ [The intuition is that $\Lambda_{\mathsf{NDE}, obs}^{r^{\ast} \perp}$ cannot be a function of either $R_{1}^{\ast}$ or $R_{1}$ as $(A_{0} - \rho) R_{1}^{\ast}$ has mean zero but is not a statistic and the statistic $(A_{0} - \rho) R_{1}^{\ast} A_{0} / \rho$ does not have mean zero.] Thus, as models for the observed data, the two models are locally equivalent and therefore we can restrict attention to $\mathbb{M}_{\mathsf{NDE}} (Y^{d})$. We further note that as $\rho \rightarrow 1$, the probability of the event $\{A_{0} = 1, A_{0, g^{opt}} = 0\}$ that testing occurred too early with regard to $g^{opt}$ increases to 1, mimicking what was observed in \citet{caniglia2019emulating} discussed in Section \ref{sec:intuition}. We are thus interested in precisely how $\rho$ affects the efficiency to be gained by imposing the NDE assumption.

We now turn to data analysis. Since $A_{0},R_{1},S_{1}$ are discrete, we estimated their joint distribution by their empirical distribution, which is also the nonparametric MLE. Because $Y$ is continuous, we approximated $\Lambda_{\mathsf{NDE}}^{\perp}$ by the large, but strict, subspace $\Omega$ of Section \ref{sec:optimal_dr} using natural spline transformations $\bm{\varphi}$ with $\xi = 6$. In our simulation, $\Omega = \{d_{0} T_{b^{\ast}, 0}: d_{0} \in \BR^{2 \xi}\}$, with $b_{0}^{\ast} (\bar{H}_{0}, \underaccent{\bar}{S}_{0}, Y^{d}) \equiv b_{0}^{\ast} (S_{1}, Y^{d}) \equiv (S_{1} \bm{\varphi} (Y^{d})^{\top}, (1 - S_{1})\bm{\varphi} (Y^{d})^{\top})^{\top}$. Hence, for any $U$, $\Pi [U | \Omega] = d_{0}^{\ast} T_{b^{\ast}, 0} \equiv T_{b_{sub}, 0}$ (following the notation in Corollary \ref{cor:sub}), with $d_{0}^{\ast} = \E [U T_{b^{\ast}, 0}^{\top}] \{\E [T_{b^{\ast}, 0} T_{b^{\ast}, 0}^{\top}]\}^{-1}$ and $T_{b^{\ast}, 0}$ as in equation \eqref{tbt-def}, with $b_{t}$ replaced by $b_{0}^{\ast} (S_{1}, Y^{d})$. Because we use substitution estimators based on the empirical distribution of discrete variables with few levels, sample splitting was unnecessary \citep{robins1995semiparametric}. The feasible version of $\Pi [U | \Omega]$ in our simulation is denoted as $\hat{T}_{b_{sub}, 0} \coloneqq \hat{\Pi} [U | \Omega] = \tilde{d}_{0}^{\ast} \hat{T}_{b^{\ast}, 0}$, where $\tilde{d}_{0}^{\ast} = \BP_{n} [U T_{b^{\ast}, 0}^{\top}] \{\BP_{n} [\hat{T}_{b^{\ast}, 0} \hat{T}_{b^{\ast}, 0}^{\top}]\}^{-1}$ and $\hat{T}_{b^{\ast}, 0}$ is $T_{b^{\ast}, 0}$ with all unknown nuisance functions replaced by their empirical version as described above. We used the entire sample to estimate $\theta_{g^{opt}}$ with each of the following estimators: 
\begin{align*}
\tilde{\theta}_{ipw, g^{opt}} & \coloneqq \mathbb{P}_{n} [\hat{V}_{ipw, g^{opt}}], \hat{V}_{ipw, g^{opt}} = \frac{(1 - A_{0}) S_{1} Y}{\hat{\Pr} (A_{0} = 0) \hat{\Pr} (S_{1} = 1 | A_{0} = 0, R_{1})}, \\
\tilde{\theta}_{nde\mbox{-}ipw, g^{opt}} & \coloneqq \mathbb{P}_{n} [\hat{V}_{nde\mbox{-}ipw, g^{opt}}], \hat{V}_{nde\mbox{-}ipw, g^{opt}} = \frac{S_{1} Y_{g^{opt}}}{\hat{\Pr} (S_{1} = 1 | A_{0}, R_{1})}, \\
\tilde{\theta}_{ipw, g^{opt}} (\hat{b}_{sub, ipw, g^{opt}}) & \coloneqq \tilde{\theta}_{ipw, g^{opt}} - \mathbb{P}_{n} \{\hat{\Pi} [\hat{V}_{ipw, g^{opt}} - \tilde{\theta}_{ipw, g^{opt}} | \Omega]\}, \\
\tilde{\theta}_{nde\mbox{-}ipw, g^{opt}} (\hat{b}_{sub, nde\mbox{-}ipw, g^{opt}}) & \coloneqq \tilde{\theta}_{nde\mbox{-}ipw, g^{opt}} - \mathbb{P}_{n} \{\hat{\Pi} [\hat{V}_{nde\mbox{-}ipw, g^{opt}} - \tilde{\theta}_{nde\mbox{-}ipw, g^{opt}} | \Omega]\}, \\
\tilde{\theta}_{\mathsf{SNMM}, g^{opt}} (\tilde{\Psi}) & \equiv \text{ the solution to } \mathbb{P}_{n} [\widehat{U}_{\mathsf{SNMM}, g^{opt}} (\theta (\widetilde{\Psi}), \widetilde{\Psi})] = 0, \\
\text{with } & \text{$\widehat{U}_{\mathsf{SNMM}, g^{opt}} (\theta (\widetilde{\Psi}), \widetilde{\Psi})$ defined in equation \eqref{optsnmm} and} \\
& \text{ $\widetilde{\Psi}$ the solution to $\mathbb{P}_{n} [\hat{\mathbb{U}} (\bm{q}, \Psi)] = 0$, defined in Section \ref{sec:review}, } \\
\tilde{\theta}_{\mathsf{SNMM}, g^{opt}} (\tilde{\Psi}; \hat{b}_{sub, g^{opt}}) & \equiv \text{ the solution to } \mathbb{P}_{n} [\widehat{U}_{\mathsf{SNMM}, g^{opt}} (\theta (\widetilde{\Psi} (\hat{b}_{sub, g^{opt}})), \widetilde{\Psi} (\hat{b}_{sub, g^{opt}}))] = 0, \\
\text{with } & \text{$\widetilde{\Psi} (\hat{b}_{sub, g^{opt}})$ the solution to $\mathbb{P}_{n} [\hat{\mathbb{U}} (\bm{q}, \Psi) - \hat{\Pi} [\hat{\mathbb{U}} (\bm{q}, \Psi) | \Omega]] = 0$,}
\end{align*}
where, for the last two estimators, we fit a saturated opt-SNMM. In a saturated opt-SNMM, $\Psi^{\ast}$ is 7 dimensional with $6 = 2 \times 3$ parameters for the effect of $S_{1}$ on $Y_{g}$ within levels of $\bar{H}_{1} = (A_{0}, R_{1})$ and 1 parameter for the effect of $A_{0}$.

\subsection{Simulation results}

In the simulation, at both sample sizes $n = 2.5 \times 10^{4}$ or $5 \times 10^{4}$, the true optimal regime $g^{opt} = (a_{0} = 0, s_{1} = 1)$ is correctly selected in each of the 100 replications by both $\tilde{\Psi}$ and $\tilde{\Psi} (\hat{b}_{sub, g^{opt}})$. Because $\tilde{\theta}_{\mathsf{SNMM}, g^{opt}} (\tilde{\Psi})$ and $\tilde{\theta}_{ipw, g^{opt}}$ are both nonparametric MLEs of $\theta_{g^{opt}}$ under an observed data model that place no restrictions on the distribution of the observed data $O$, we expected and empirically verified that $\tilde{\theta}_{\mathsf{SNMM}, g^{opt}} (\widetilde{\Psi})$ and $\widetilde{\theta}_{ipw, g^{opt}}$ are algebraically identical in each of the 100 replications. Furthermore since $\widetilde{\theta}_{\mathsf{SNMM}, g^{opt}} (\widetilde{\Psi}; \widehat{b}_{sub, g^{opt}})$ and $\widetilde{\theta}_{ipw, g^{opt}} (\widehat{b}_{sub, ipw, g^{opt}})$ are computed by subtracting from $\widetilde{\theta}_{\mathsf{SNMM}, g^{opt}} (\widetilde{\Psi})$ and $\widetilde{\theta}_{ipw, g^{opt}}$ the empirical projections of their algebraically identical empirical influence functions on the same space $\Omega$, we expected and empirically verified that $\widetilde{\theta}_{\mathsf{SNMM}, g^{opt}} (\widetilde{\Psi}; \widehat{b}_{sub, g^{opt}})$ and $\widetilde{\theta}_{ipw, g^{opt}} (\widehat{b}_{sub, ipw, g^{opt}})$ are also algebraically identical. As a result it suffices to report results for $\widetilde{\theta}_{ipw, g^{opt}}$ and $\widetilde{\theta}_{ipw, g^{opt}} (\widehat{b}_{sub, ipw, g^{opt}})$.

The results of the simulation study are in Figure \ref{fig:dgp1value}. The results are consistent with the intuition we gained in Section \ref{sec:intuition}. In Figure \ref{fig:dgp1value} we plot the following as a function of $\rho $: (1) in the left panel, the Monte Carlo variances of $\tilde{\theta}_{ipw, g^{opt}}$ (black), $\tilde{\theta}_{nde\mbox{-}ipw, g^{opt}}$ (blue), $\tilde{\theta}_{ipw, g^{opt}} (\hat{b}_{sub, ipw, g^{opt}})$ (red) and $\tilde{\theta}_{nde\mbox{-}ipw, g^{opt}} (\hat{b}_{sub, nde\mbox{-}ipw, g^{opt}})$ (green); and (2) in the right panel, the inverse relative efficiencies (REs) of $\tilde{\theta}_{nde\mbox{-}ipw, g^{opt}}$ (blue), $\tilde{\theta}_{ipw, g^{opt}} (\hat{b}_{sub, ipw, g^{opt}})$ (red) and $\tilde{\theta}_{nde\mbox{-}ipw, g^{opt}} (\hat{b}_{sub, nde\mbox{-}ipw, g^{opt}})$ (green) against $\tilde{\theta}_{ipw, g^{opt}}$.

We now summarize the results of Figure \ref{fig:dgp1value}:

\begin{itemize}
\item In the left panel of Figure \ref{fig:dgp1value}, the Monte Carlo variance of $\tilde{\theta}_{ipw, g^{opt}}$ increases dramatically as $\rho$ (and thus the probability of the event $[A_{0} = 1, A_{0, g^{opt}} = 0]$) approaches 1. The variance increases because all subjects $i$ with $A_{0, i} = 1$ are censored so $\hat{V}_{ipw, g^{opt}, i} = 0$. In contrast, the Monte Carlo variances of the three estimators that leverage the NDE assumption are stable as $\rho$ varies and cannot be easily distinguished from one another at this scale.

\item In the right panel of Figure \ref{fig:dgp1value}, we compare the inverse RE of the three estimators leveraging the NDE assumption relative to the estimator $\tilde{\theta}_{ipw, g^{opt}}$. As $\rho$ approaches 1, the inverse relative efficiency of the three estimators decreases to $0$. When $\rho$ is far from 1, the estimators $\tilde{\theta}_{ipw, g^{opt}}(\hat{b}_{sub, ipw, g^{opt}})$ and $\tilde{\theta}_{nde\mbox{-}ipw, g^{opt}} (\hat{b}_{sub, nde\mbox{-}ipw, g^{opt}})$ (red and green curves, which are nearly on top of each other) are more efficient than $\tilde{\theta}_{nde\mbox{-}ipw, g^{opt}}$ because the former two are nearly semiparametric efficient in the model $\mathbb{M}_{\mathsf{NDE}} (Y^{d})$, as they are residuals from the projection of an influence function onto a large subspace $\Omega$ of the ortho-complement $\Lambda_{\mathsf{NDE}}^{\perp}$ to the tangent space of $\mathbb{M}_{\mathsf{NDE}} (Y^{d})$. However, as $\rho$ approaches 1, the variance of $\tilde{\theta}_{nde\mbox{-}ipw, g^{opt}}$ converges from above to that of $\tilde{\theta}_{ipw, g^{opt}} (\hat{b}_{sub, ipw, g^{opt}})$ and $\tilde{\theta}_{nde\mbox{-}ipw, g^{opt}} (\hat{b}_{sub, nde\mbox{-}ipw, g^{opt}})$, as the ortho-complement to the tangent space of model $\mathbb{M}_{\mathsf{NDE}} (Y^{d})$ becomes degenerate at $\rho = 1$. When $\rho = 1$ (not shown in the plot), we can show the following: $\tilde{\theta}_{ipw, g^{opt}}$ and, thus $\tilde{\theta}_{ipw, g^{opt}} (\hat{b}_{sub, ipw, g^{opt}})$ are undefined and thus inconsistent due to the positivity violation: $A_{0} = 0$ with probability 1. In contrast, $\tilde{\theta}_{nde\mbox{-}ipw, g^{opt}}$ and $\tilde{\theta}_{nde\mbox{-}ipw, g^{opt}} (\hat{b}_{sub, nde\mbox{-}ipw, g^{opt}})$ are equal with probability 1 and are RAL. In fact, they are semiparametric efficient for $\theta_{g^{opt}}$ in model $\mathbb{M}_{\mathsf{NDE}} (Y^{d})$ since, at $\rho = 1$, $\theta_{g^{opt}}$ is just identified and so has a unique influence function. Further details of this simulation are presented in Appendix \ref{sim:further}.
\end{itemize}
In Appendix \ref{sim:cost}, we provide a second simulation which shows that adjusting for the NDE assumption in an opt-SNMM can greatly improve the efficiency of estimating the VoI in the context of a cost benefit analysis. As mentioned in Section \ref{sec:background}, $\tilde{\theta}_{nde\mbox{-}ipw, g^{opt}}$ can be less efficient than the estimator $\tilde{\theta}_{ipw, g^{opt}}$ that ignores the NDE assumption under certain DGPs. In Appendix \ref{app:worse}, we design a data generating process to illustrate this point. We perform exact calculations of the variance of the oracle estimators $\hat{\theta}_{ipw, g^{opt}}$ and $\hat{\theta}_{nde\mbox{-}ipw, g^{opt}}$ and show $\var [\hat{\theta}_{nde\mbox{-}ipw, g^{opt}}] / \var [\hat{\theta}_{ipw, g^{opt}}] > 1$. We actually show this ratio can be made arbitrarily large by varying a parameter of the DGP.

\begin{figure}[tbp]
\centering
\includegraphics[scale=.35]{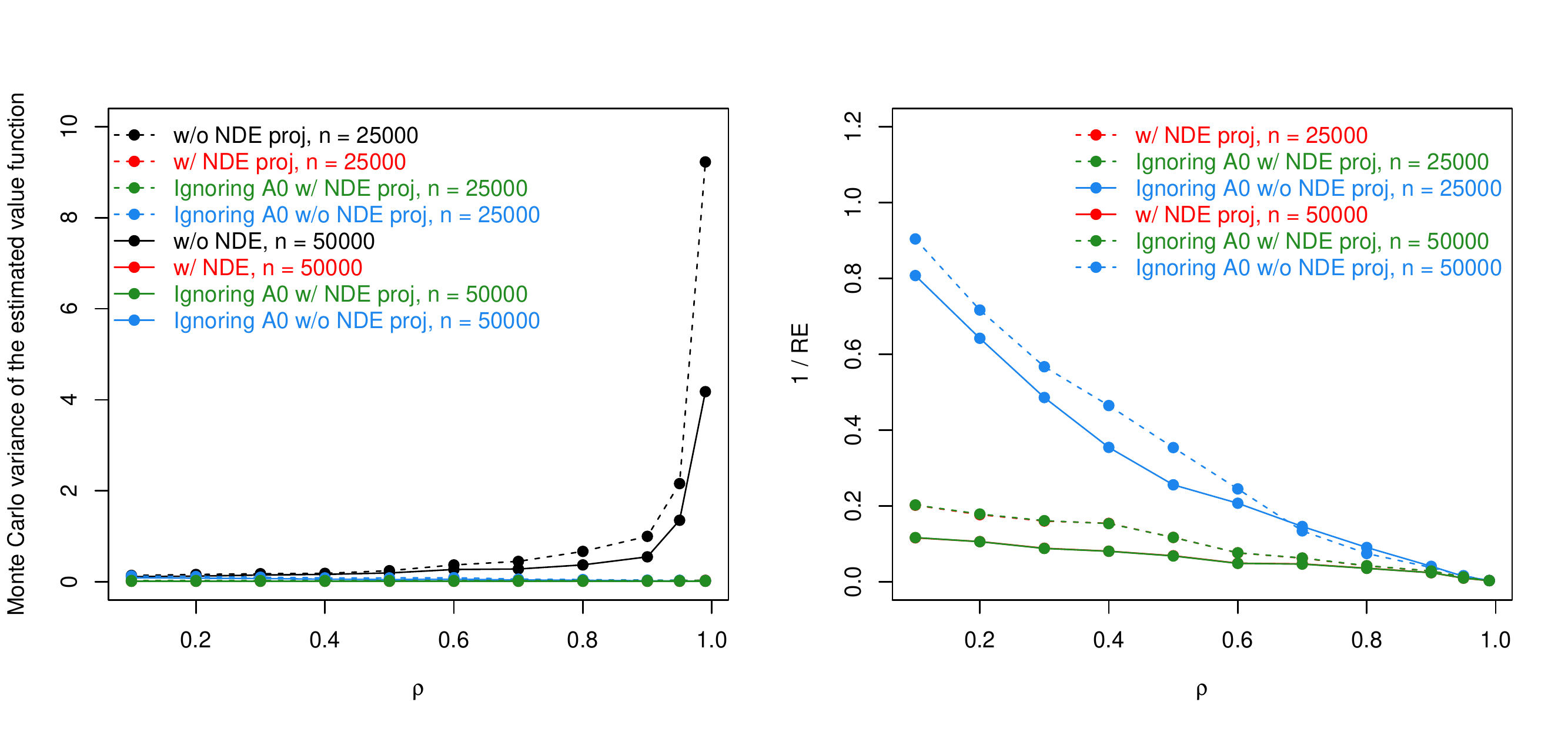}
\caption{The Monte Carlo variances (left panel) of $\tilde{\theta}_{ipw, g^{opt}}$ (black), $\hat{\theta}_{nde\mbox{-}ipw, g^{opt}}$ (blue), $\tilde{\theta}_{ipw, g^{opt}} (\hat{b}_{sub, ipw, g^{opt}})$ (red) and $\tilde{\theta}_{nde\mbox{-}ipw, g^{opt}} (\hat{b}_{sub, nde\mbox{-}ipw, g^{opt}})$ (green), and inverse REs (right panel) of $\tilde{\theta}_{nde\mbox{-}ipw, g^{opt}}$ (blue), $\tilde{\theta}_{ipw, g^{opt}} (\hat{b}_{sub, ipw, g^{opt}})$ (red) and $\tilde{\theta}_{nde\mbox{-}ipw, g^{opt}} (\hat{b}_{sub, nde\mbox{-}ipw, g^{opt}})$ (green) against $\tilde{\theta}_{ipw, g^{opt}}$ over different $\rho$'s. The red and green lines lie on top of one another; due to scaling, the blue and green lines are indistinguishable in the left panel.} \label{fig:dgp1value}
\end{figure}

\section{Discussion}
\label{sec:discussion} 
To conclude, we discuss five open problems. The optimal choice of $\bm{q}^{opt}$ for the user-supplied function $\bm{q}$ is the unique vector function $\bm{q}^{opt}$ satisfying for all $\bm{q} = \{q_{t} (\bar{H}_{t}, s_{t}, a_{t}); t = 0, \ldots, K\}$, 
\begin{equation*}
\E \left[ \frac{\partial \E \left[ \BU (\bm{q}, \Psi^{\ast}) \right]}{\partial \Psi^{\top}} \right] = \E [\BU (\bm{q}, \Psi^{\ast}) \BU (\bm{q}^{opt}, b_{opt} (\bm{q}^{opt}), \Psi^{\ast})^{\top}].
\end{equation*}
Furthermore $\left\{ \E \left[ \BU (\bm{q}^{opt}, b_{opt} (\bm{q}^{opt}), \Psi^{\ast}) \BU (\bm{q}^{opt}, b_{opt} (\bm{q}^{opt}), \Psi^{\ast})^{\top} \right] \right\}^{-1}$ is the semiparametric variance bound for $\Psi^{\ast}$ in model $\mathbb{M}$. We did not explore estimation of $\bm{q}^{opt}$ because, as discussed earlier, in practice, users of opt-SNMM have chosen to employ heuristic choices of $\bm{q}$ in their analyses. Nonetheless, it would be interesting to study the estimation with $\bm{q}^{opt}$ in future research. A second open problem is to develop multiply robust estimators of the parameters of an opt-SNMM under the NDE assumption to provide even more robustness than that obtained by the doubly robust estimators proposed in the current paper \citet{luedtke2017sequential, rotnitzky2017multiply}.

A third open problem is to extend our methodology to the case where either (i) the dimension of the parameter $\Psi^{\ast}$ of the opt-SNMM is allowed to grow with the sample size or (ii) the optimal blip functions $\gamma_{t}^{g^{opt}} (\bar{H}_{t}, s_{t}, a_{t})$ are modeled non-parametrically under (e.g. smoothness) restrictions on their complexity. This problem is important because our estimator of the optimal regime is not robust to misspecification of the opt-SNMM model. A fourth important open problem is the extension of our methods to data generated under exceptional laws as defined in \citet{robins2004optimal}.

We conclude by discussing a fifth open problem. Consider the setting discussed earlier under which positivity fails because $\Pr (A_{t} = 1) = 1$ for all $t$ but other aspects of positivity, consistency and sequential exchangeability continue to hold. In this setting, the parameter vector $\Psi^{\ast}$ of a saturated opt-SNMM is identified only if the NDE assumption holds. In fact, under the NDE assumption, all $g$ for which $\theta_{g}$ was identified under positivity remain identified; similarly, since $\Psi^{\ast}$ is identified, $g^{opt}$ and $\theta_{g^{opt}}$ remain identified under the NDE assumption. However, by lack of identification in the absence of the NDE assumption, the methodology of this paper cannot be used to estimate $\Psi^{\ast}$. It remains open how to efficiently estimate $\Psi^{\ast}$ and thus $g^{opt}$ and $\theta_{g^{opt}}$ in this setting under the NDE assumption. Preliminary investigations indicate that estimation of $g^{opt}$ by dynamic programming is not possible. If true, then any algorithms that compute or estimate $g^{opt}$ may well be computationally intractable.

\section*{Online supplementary materials}

Appendix and R codes are included in the online supplementary materials.

\section*{Acknowledgement}

We thank Michael R. Kosorok and two anonymous referees for helpful comments. Lin Liu and James M. Robins were supported by the U.S. Office of Naval Research Grant N000141912446, and National Institutes of Health (NIH) awards R01 AG057869 and R01 AI127271. Lin Liu was also partially sponsored by Shanghai Pujiang Program Research Grant 20PJ1408900 and Shanghai Jiao Tong University Start-up Grant WF220441912. The computations were run on the FASRC Cannon cluster supported by the FAS Division of Science Research Computing Group at Harvard University.

\bibliographystyle{chicago}
\bibliography{./snmm}

\newpage

\appendix
\renewcommand{\thesection}{\Alph{section}} \setcounter{section}{0}

\section{Appendix}

The online supplementary materials include technical proofs and more details on the simulation studies omitted in the main text.

\subsection{Proof of Theorem \ref{lem:nuisance_nde}}
\label{app:nuisance_nde}

In this section, we prove Theorem \ref{lem:nuisance_nde}.

\begin{proof}
Since $p (S_{m} | \bar{H}_{m})$ and $p (A_{m} | \bar{H}_{m}, S_{m})$ are unconstrained under the model, it follows that $\Lambda_{\mathsf{NDE}}^{\perp}$ must be comprised of linear combinations of the residuals from the orthogonal projections in $L_{2, 0} (P)$ of $D_{b, t}$ onto scores for arbitrary parametric submodels for $p (S_{m} | \bar{H}_{m})$ and $p (A_{m} | \bar{H}_{m}, S_{m})$. Thus, any element of $\Lambda_{\mathsf{NDE}}^{\perp}$ must be linear combinations of random variables of the form $\sum_{t = 0}^{K} R_{b, t}$ where 
\begin{equation*}
R_{b, t} = D_{b, t} - \sum_{m = t}^{K} \{\E [D_{b, m} | \bar{H}_{m}, S_{m}, A_{m}] - \E [D_{b, m} | \bar{H}_{m}, S_{m}]\} - \sum_{m = t + 1}^{K} \{\E [D_{b, t} | \bar{H}_{t}, S_{t}] - \E [D_{b, t} | \bar{H}_{t}]\}
\end{equation*}
However, the terms $\E [D_{b, t} | \bar{H}_{m}, S_{m}, A_{m}] - \E [D_{b, t} | \bar{H}_{m}, S_{m}]$ occurring in $R_{b, t}$ are $0$ for $m > t$, as then $\E [D_{b, t} | \bar{H}_{m}, S_{m}, A_{m}]$ does not depend on $A_{m}$ since 
\begin{eqnarray*}
\E [D_{b, t} | \bar{H}_{m}, A_{m}, S_{m}] &=& \E \left[ \left. \frac{b_{t} (\bar{H}_{t}, \underaccent{\bar}{S}_{t}, Y^{d})}{W_{m + 1}} \right\vert \bar{H}_{m}, A_{m}, S_{m} \right] \frac{(A_{t} - \E [A_{t} | \bar{H}_{t}, S_{t}])}{\prod_{j = t + 1}^{m} p (S_{j} | \bar{H}_{j})} \\
&=& \E \left[ \left. \frac{b_{t} (\bar{H}_{t}, \underaccent{\bar}{S}_{t}, Y^{d})}{W_{m+1}} \right\vert \bar{H}_{m}, S_{m} \right] \frac{(A_{t} - \E [A_{t} | \bar{H}_{t}, S_{t}])}{\prod_{j = t + 1}^{m} p (S_{j} | \bar{H}_{j})}
\end{eqnarray*}
where the last equality follows from equation \eqref{rest2}. Thus $R_{b, t} = T_{b, t}$ as defined in equation \eqref{tbt-def}.
\end{proof}

\subsection{Proof of Theorem \ref{thm:tdr}}
\label{app:alternative_nde}

In this section, we prove Theorem \ref{thm:tdr}. To facilitate our proof, we first provide an algebraically equivalent expression for $T_{b, t}$.

For $0 \leq t \leq j \leq K$ define $W_{t}^{j} \equiv \prod_{m = t}^{j} p (S_{m} | \bar{H}_{m})$ and for $j < t$ define $W_{t}^{j} \equiv 1$. Given arbitrary functions $\eta_{k}^{\dag} (\bar{H}_{k}, S_{k})$, for $k = 0, \ldots, K$, define $y_{k + 1, \eta_{k + 1}^{\dag}} (\bar{H}_{k + 1}) \equiv \sum_{s_{k + 1}} \eta_{k + 1}^{\dag} (\bar{H}_{k + 1}, s_{k + 1})$. Notice that $y_{k + 1, \eta_{k + 1}^{\dag}} (\bar{H}_{k + 1})$ coincides with $\E \left[ \left. \frac{\eta_{k + 1}^{\dag} (\bar{H}_{k + 1}, S_{k + 1})}{p (S_{k + 1} | \bar{H}_{k + 1})} \right\vert \bar{H}_{k + 1} \right]$.

Given a fixed time $t$ and a function $b_{t} (\bar{H}_{t}, \underaccent{\bar}{S}_{t}, Y^{d})$ define $\eta_{K} (\bar{H}_{K}, S_{K}) \equiv \E [b_{t} (\bar{H}_{t}, \underaccent{\bar}{S}_{t}, Y^{d}) | \bar{H}_{K}, S_{K}]$ and for $k = K - 1, \ldots, t$ define recursively $\eta_{k} (\bar{H}_{k}, S_{k}) \equiv \E \left[ y_{k + 1, \eta_{k + 1}} (\bar{H}_{k + 1}) | \bar{H}_{k}, S_{k} \right]$. Finally, also define $y_{K + 1, \eta_{K + 1}} (\bar{H}_{K + 1}) \equiv b_{t} (\bar{H}_{t}, \underaccent{\bar}{S}_{t}, Y^{d})$. In the preceding definitions we have suppressed the dependence on the time $t$ and the function $b_{t}$ to simplify the notation. We thus have the following result.

\begin{lemma}[Alternative expression for $T_{b, t}$]
\label{lem:alternative_nde} 
For any $b_{t} (\bar{H}_{t}, \underaccent{\bar}{S}_{t}, Y^{d})$ it holds that 
\begin{eqnarray}
T_{b, t} &=& \left[ \frac{b_{t} (\bar{H}_{t}, \underaccent{\bar}{S}_{t}, Y^{d})}{W_{t + 1}^{K}} - \eta_{t} (\bar{H}_{t}, S_{t}) - \sum_{k = t + 1}^{K} \frac{1}{W_{t + 1}^{k - 1}} \left\{ \frac{\eta_{k} (\bar{H}_{k}, S_{k})}{p (S_{k} | \bar{H}_{k})} - y_{k, \eta_{k}} (\bar{H}_{k}) \right\} \right] \left( A_{t} - \Pi_{t} \right) \notag \\
&=& \left[ \sum_{k = t + 1}^{K + 1} \frac{1}{W_{t + 1}^{k - 1}} \left\{ y_{k, \eta_{k}}(\bar{H}_{k}) - \eta_{k - 1} (\bar{H}_{k - 1}, S_{k - 1}) \right\} \right] \left( A_{t} - \Pi_{t} \right). \label{tbt}
\end{eqnarray}
Furthermore, $T_{b, t}$ is doubly robust in the sense that it has mean zero if (i) $\eta_{k} (\cdot, \cdot)$ are replaced by arbitrary functions for all $k = 0, \ldots, K$ or, (ii) $p (S_{k} | \bar{H}_{k})$ and $\Pi_{t}$ are replaced by arbitrary conditional probability functions.
\end{lemma}

Note that $T_{b, t}$ has mean zero under (i) follows immediately from the first expression for $T_{b, t}$ in equation \eqref{tbt} and under (ii) from the second expression. Now we prove Lemma \ref{lem:alternative_nde} and hence finish the proof of Theorem \ref{thm:tdr}.

\begin{proof}
Suppose that $\eta_{k}$ is replaced by an arbitrary function $\eta_{k}^{\dag}$ for all $k = t, \ldots, K$. In the first expression for $T_{b, t}$, (a) $\E [\frac{b_{t} (\bar{H}_{t}, \underaccent{\bar}{S}_{t}, Y^{d})}{W_{t + 1}^{K}} (A_{t} - \Pi_{t}) | \bar{H}_{t}, S_{t}] = 0$, by the restriction of model $\mathbb{M}_{\mathsf{NDE}}$, (b) $\E [\eta_{t}^{\dag} (\bar{H}_{t}, S_{t}) (A_{t} - \Pi_{t}) | \bar{H}_{t}, S_{t}] = 0$ because $\Pi_{t} \equiv \E [A_{t} | \bar{H}_{t}, S_{t}]$ and (c) for $k > t$ 
\begin{align*}
& \E \left[ \left. \frac{1}{W_{t + 1}^{k - 1}} \left\{ \frac{\eta_{k}^{\dag} (\bar{H}_{k}, S_{k})}{p (S_{k} | \bar{H}_{k})} - y_{k, \eta_{k}^{\dag}} (\bar{H}_{k})\right\} \right\vert \bar{H}_{t}, S_{t} \right] \\
= & \E \left[ \left. \frac{1}{W_{t + 1}^{k - 1}} \E \left[ \left. \left\{ \frac{\eta_{k}^{\dag} (\bar{H}_{k}, S_{k})}{p (S_{k} | \bar{H}_{k})} - y_{k, \eta_{k}^{\dag}} (\bar{H}_{k}) \right\} \right\vert \bar{H}_{k} \right] \right\vert \bar{H}_{t}, S_{t} \right] = 0
\end{align*}
because $y_{k, \eta_{k}^{\dag}} (\bar{H}_{k})$ coincides $\E \left[ \left. \frac{\eta_{k}^{\dag} (\bar{H}_{k}, S_{k})}{p (S_{k} | \bar{H}_{k})} \right\vert \bar{H}_{k} \right]$. On the other hand, suppose that $p (S_{k} | \bar{H}_{k})$ and $\Pi_{t}$ are replaced by arbitrary conditional probability functions $p^{\dag} (S_{k} | \bar{H}_{k})$ and $\Pi_{t}^{\dag}$. In the second expression for $T_{b, t}$, 
\begin{align*}
& \E \left[ \left. \frac{1}{W_{t + 1}^{\dag k - 1}} \left\{ y_{k, \eta_{k}} (\bar{H}_{k}) - \eta_{k - 1} (\bar{H}_{k - 1}, S_{k - 1}) \right\} (A_{t} - \Pi_{t}^{\dag}) \right\vert \bar{H}_{k - 1}, S_{k - 1} \right] \\
= & \frac{1}{W_{t + 1}^{\dag k - 1}} (A_{t} - \Pi_{t}^{\dag}) \E \left[ \left. \left\{ y_{k, \eta_{k}} (\bar{H}_{k}) - \eta_{k - 1} (\bar{H}_{k - 1}, S_{k - 1}) \right\} \right\vert \bar{H}_{k - 1}, S_{k - 1} \right] = 0
\end{align*}
because by definition, $y_{k, \eta_{k}}$ and $\eta_{k}$ do not depend on the conditional probabilities of $p (S_{t} | \bar{H}_{t})$ and $\pi_{t} (\bar{H}_{t}, S_{t})$ for any $t$ and $\eta_{k - 1} (\bar{H}_{k - 1}, S_{k - 1}) \equiv \E \left[ \left. y_{k, \eta_{k}} (\bar{H}_{k}) \right\vert \bar{H}_{k - 1}, S_{k - 1} \right]$.
\end{proof}

In this section, we instead showed a stronger form of double robustness than that of Theorem \ref{thm:tdr}: we do not have to correctly specify the full law of every covariates $L_{t}$ given the past; rather, we only need to ensure $\eta_{t} (\bar{H}_{t}, S_{t}) = \E^{\dag} [y_{t + 1, \eta_{t + 1}} (\bar{H}_{t}) | \bar{H}_{t}, S_{t}]$ for the law $p^{\dag} (L_{t} | \bar{H}_{t - 1}, S_{t}, A_{t})$.

\subsection{The projection $\Pi [U | \Lambda_{\mathsf{NDE}}^{\perp}]$ when $U$ is a continuous random variable}
\label{app:cont}

We will now derive the equations that define the projection of any random variable into $\Lambda_{\mathsf{NDE}}^{\perp}$. Let $\Lambda_{t}^{trx, test} = \{q_{t} (\bar{H}_{t}, S_{t}, A_{t}): \E [q_{t} (\bar{H}_{t}, S_{t}, A_{t}) | \bar{H}_{t}] = 0\}$ and let $\underaccent{\bar}{\Lambda}_{t}^{trx, test} \equiv \oplus _{m = t}^{K} \Lambda_{m}^{trx, test}$.

We next show that for any $0 \leq r, t \leq K$ and any $b_{t} (\bar{H}_{t}, \underaccent{\bar}{S}_{t}, Y^{d})$ and $b_{r}^{\ast} (\bar{H}_{r}, \underaccent{\bar}{S}_{r}, Y^{d})$ in $L_{2} (P)$, 
\begin{equation}
\E \left[ T_{b^{\ast}, r} T_{b, t} \right] = \E \left[ T_{b^{\ast}, r} D_{b, t} \right] \label{a_big_id}
\end{equation}
where $T_{b, t} = D_{b, t} - \Pi [D_{b, t} | \underaccent{\bar}{\Lambda}_{t}^{trx, test}]$ and $T_{b^{\ast}, r}$ is defined likewise.

Suppose first that $r \leq t$. Then $T_{b^{\ast}, r}$ is orthogonal with $\underaccent{\bar}{\Lambda}_{t}^{trx, test}$ because $\underaccent{\bar}{\Lambda}_{t}^{trx, test} \subset \underaccent{\bar}{\Lambda}_{r}^{trx, test}$ if $r \leq t$, and $T_{b^{\ast}, r}$ is a residual from a projection into $\underaccent{\bar}{\Lambda}_{r}^{trx, test}$. Therefore, $\E [T_{b^{\ast}, r} T_{b, t}] = \E [T_{b^{\ast}, r} \{D_{b, t} - \Pi [D_{b, t} | \Lambda_{t}^{trx, test}]\}] = \E [T_{b^{\ast}, r} D_{b, t}]$.

Suppose next that $r > t$. Then $\Pi [D_{b^{\ast}, r} | \Lambda_{t}^{trx, test}] = 0$. Consequently, $T_{b^{\ast}, r} = D_{b^{\ast}, r} - \Pi [D_{b^{\ast}, r} | \underaccent{\bar}{\Lambda}_{r}^{trx, test}] = D_{b^{\ast}, r} - \Pi [D_{b^{\ast}, r} | \underaccent{\bar}{\Lambda}_{t}^{trx, test}]$ is orthogonal to $\underaccent{\bar}{\Lambda}_{t}^{trx, test}$ which again yields, $\E [T_{b^{\ast}, r} T_{b, t}] = \mathsf{E} [T_{b^{\ast}, r} D_{b, t}]$.

Now, recalling that $\Lambda_{\mathsf{NDE}}^{\perp} = \sum_{j = 0}^{K} \mathcal{T}_{j}$, we conclude that 
\begin{equation*}
\Pi [U | \Lambda_{\mathsf{NDE}}^{\perp}] = \sum_{t = 0}^{K} T_{b^{\ast}, t}
\end{equation*}
for $b_{t}^{\ast} \in \mathcal{B}_{t}, t = 0, \ldots, K$ if and only if for all $b_{t} (\bar{H}_{t}, \underaccent{\bar}{S}_{t}, Y^{d})$ such that $\E [b_{t} (\bar{H}_{t}, \underaccent{\bar}{S}_{t}, Y^{d})^{2}] < \infty$ it holds that 
\begin{eqnarray*}
0 &=& \E \left[ \left\{ U - \sum_{t = 0}^{K} T_{b^{\ast}, t} \right\} \sum_{t = 0}^{K} T_{b, t} \right] \\
&=& \sum_{t = 0}^{K} \E \left[ U T_{b, t} \right] - \sum_{t = 0}^{K} \sum_{r = 0}^{K} \E \left[ T_{b^{\ast}, r} T_{b, t} \right] \\
&=& \sum_{t = 0}^{K} \E \left[ \left\{ U - \Pi \left[ U | \underaccent{\bar}{\Lambda}_{t}^{trx, test} \right] \right\} D_{b, t} \right] - \sum_{t = 0}^{K} \sum_{r = 0}^{K} \E \left[ T_{b^{\ast}, r} D_{b, t} \right] \\
&=& \sum_{t = 0}^{K} \E \left[ \left\{ U - \Pi \left[ U | \underaccent{\bar}{\Lambda}_{t}^{trx, test} \right] - \sum_{r = 0}^{K} T_{b^{\ast}, r} \right\} D_{b, t} \right] \\
&=& \sum_{t = 0}^{K} \E \left[ \left\{ U - \Pi \left[ U | \underaccent{\bar}{\Lambda}_{t}^{trx, test} \right] - \sum_{r = 0}^{K} T_{b^{\ast}, r} \right\} \frac{A_{t} - \Pi_{t}}{W_{t + 1}} b_{t} (\bar{H}_{t}, \underaccent{\bar}{S}_{t}, Y^{d}) \right].
\end{eqnarray*}

Then, in particular, taking for each $t$ 
\begin{equation*}
b_{t} (\bar{H}_{t}, \underaccent{\bar}{S}_{t}, Y^{d}) = \E \left[ \left. \left\{ U - \Pi \left[ U | \underaccent{\bar}{\Lambda}_{t}^{trx, test} \right] - \sum_{r = 0}^{K} T_{b^{\ast}, r} \right\} \frac{A_{t} - \Pi_{t}}{W_{t + 1}} \right\vert \bar{H}_{t}, \underaccent{\bar}{S}_{t}, Y^{d} \right]
\end{equation*}
we conclude that $b^{\ast} \equiv \left(b_{0}^{\ast}, \ldots, b_{K}^{\ast} \right)$ must solve the system of equations for $t = 0, \ldots, K$ 
\begin{eqnarray}
&& \E \left[ \left. \left\{ U - \Pi \left[ U | \underaccent{\bar}{\Lambda}_{t}^{trx, test} \right] \right\} \frac{A_{t} - \Pi_{t}}{W_{t + 1}} \right\vert \bar{H}_{t}, \underaccent{\bar}{S}_{t}, Y^{d} \right] \label{eq_int} \\
&=& \sum_{r = 0}^{K} \E \left[ \left. \left\{ D_{b^{\ast}, r} - \Pi \left[ D_{b^{\ast}, r} | \underaccent{\bar}{\Lambda}_{r}^{trx, test} \right] \right\} \frac{A_{t} - \Pi_{t}}{W_{t + 1}} \right\vert \bar{H}_{t}, \underaccent{\bar}{S}_{t}, Y^{d} \right] \notag
\end{eqnarray}
or equivalently 
\begin{align*}
& \E \left[ \left. \left\{ U - \Pi \left[ U | \underaccent{\bar}{\Lambda}_{t}^{trx, test} \right] \right\} \frac{A_{t} - \Pi_{t}}{W_{t + 1}} \right\vert \bar{H}_{t}, \underaccent{\bar}{S}_{t}, Y^{d} \right] \\
= & \ \sum_{r = 0}^{K} \E \left[ \left. \left\{ \begin{array}{c}
D_{b^{\ast}, r} - \left\{ \E \left[ D_{b^{\ast}, r} | \bar{H}_{r}, S_{r}, A_{r} \right] - \E \left[ D_{b^{\ast}, r} | \bar{H}_{r}, S_{r} \right] \right\} \\ 
- \sum\limits_{m = r + 1}^{K} \left\{ \E \left[ D_{b^{\ast}, r} | \bar{H}_{m}, S_{m} \right] - \E \left[ D_{b^{\ast}, r} | \bar{H}_{m} \right] \right\}
\end{array} \right\} \frac{A_{t} - \Pi_{t}}{W_{t + 1}} \right\vert \bar{H}_{t}, \underaccent{\bar}{S}_{t}, Y^{d} \right]
\end{align*}
The system of equations \eqref{eq_int} for $t = 0, \ldots, K$ has a unique solution $b^{\ast} = \left( b_{0}^{\ast}, \ldots, b_{K}^{\ast} \right)$ such that each $b_{t}^{\ast} \in \mathcal{B}_{t}$, but the solution does not exist in closed form unless $U$ is discrete.

\subsection{Proof of Theorem \ref{thm:sub}}
\label{app:sub}

In this section, we prove Theorem \ref{thm:sub}.

\begin{proof}
We will show by induction in $K$ that with $T_{j}^{(t)}$, $0 \leq j \leq t \leq K$, defined as in the main text, $\Pi \left[ U | \sum_{j = 0}^{K} \Gamma_{j} \right] = \sum_{j = 0}^{K} d_{j}^{\ast} (\bar{H}_{j}) T_{j}$ where 
\begin{equation}
d_{0}^{\ast} (\bar{H}_{0}) \equiv \E \left[ \left. U T_{0}^{(0) \top} \right\vert \bar{H}_{0} \right] \E \left[ \left. T_{0}^{(0)} T_{0}^{(0) \top} \right\vert \bar{H}_{0} \right]^{-1} \label{main_0_app}
\end{equation}
and for $t = 1, \ldots, K$, 
\begin{equation}
d_{t}^{\ast} (\bar{H}_{t}) \equiv \E \left[ \left. \left\{ U - \sum_{j = 0}^{t - 1} d_{j}^{\ast} (\bar{H}_{j}) T_{j}^{(t)} \right\} T_{t}^{(t) \top} \right\vert \bar{H}_{t} \right] \E \left[ \left. T_{t}^{(t)} T_{t}^{(t) \top} \right\vert \bar{H}_{t} \right]^{-1}. \label{main_d1_app}
\end{equation}

Suppose $K = 0$. Let $d_{0}^{\ast}$ be such that $\Pi \left[ U | \Gamma_{0} \right] = d_{0}^{\ast} (\bar{H}_{0}) T_{0}$. Then for all $d_{0}$ such that $\E \left[ d_{0} (\bar{H}_{0})^{2} \right] < \infty$, it holds that 
\begin{equation*}
0 = \E \left[ \left\{ U - d_{0}^{\ast} (\bar{H}_{0}) T_{0} \right\} T_{0}^{\top} d_{0}^{\top} (\bar{H}_{0}) \right]
\end{equation*}
or equivalently, 
\begin{equation*}
\E \left[ \left\{ U - d_{0}^{\ast} (\bar{H}_{0}) T_{0} \right\} T_{0}^{\top} | \bar{H}_{0} \right] = 0.
\end{equation*}
Therefore 
\begin{equation*}
\E \left[ U T_{0}^{\top} | \bar{H}_{0} \right] - d_{0}^{\ast} (\bar{H}_{0}) \E \left[ T_{0} T_{0}^{\top} | \bar{H}_{0} \right] = 0.
\end{equation*}
Hence 
\begin{equation*}
d_{0}^{\ast} (\bar{H}_{0}) = \E \left[ U T_{0}^{\top} | \bar{H}_{0} \right] \E \left[ T_{0} T_{0}^{\top} | \bar{H}_{0} \right]^{-1}.
\end{equation*}

Suppose now that the theorem holds for $K = 0, 1, \ldots, t - 1$. We will show that it holds for $K = t$.

Let $d_{j}^{\ast}, j = 0, \ldots, t$ be such that $\Pi \left[ U | \sum_{j = 0}^{t} \Gamma_{j} \right] = \sum_{j = 0}^{t} d_{j}^{\ast} (\bar{H}_{j}) T_{j} $. The functions $d_{j}^{\ast}, j = 0, \ldots, t$ then satisfy 
\begin{equation*}
0 = \E \left[ \left\{ U - \sum_{k = 0}^{t} d_{k}^{\ast} (\bar{H}_{k}) T_{k} \right\} \left\{ \sum_{k = 0}^{t} T_{k}^{\top} d_{k}^{\top} (\bar{H}_{k}) \right\} \right] = 0
\end{equation*}
for all $d_{j}, j = 0, \ldots, t$ such that $\cov \left[ d_{j} (\bar{H}_{j}) \right] < \infty$. Choosing, in particular, $d_{j} (\bar{H}_{j}) = \E \left[ \left. \left\{ U - \sum_{k = 0}^{t} d_{k}^{\ast} (\bar{H}_{k}) T_{k} \right\} T_{j}^{\top} \right\vert \bar{H}_{j} \right] $ and $d_{k} (\bar{H}_{k}) = 0$ for $k \neq j$ we arrive at the set of equations
for $j = 0, \ldots, t$, 
\begin{equation}
0 = \E \left[ \left. \left\{ U - \sum_{k = 0}^{t} d_{k}^{\ast} (\bar{H}_{k}) T_{k} \right\} T_{j}^{\top} \right\vert \bar{H}_{j} \right]. \label{main_2}
\end{equation}

From equation \eqref{main_2} applied to $j = t$ we arrive at 
\begin{equation}
d_{t}^{\ast} (\bar{H}_{t}) = \E \left[ \left. \left\{ U - \sum_{k = 0}^{t - 1} d_{k}^{\ast} (\bar{H}_{k}) T_{k} \right\} T_{t}^{\top} \right\vert \bar{H}_{t} \right] \left\{ \E \left[ \left. T_{t} T_{t}^{\top} \right\vert \bar{H}_{t} \right] \right\}^{-1}  \label{dt_main}
\end{equation}
Replacing the right hand side into equation \eqref{main_2} we arrive at the set of equations for $j = 0, \ldots, t - 1$ 
\begin{equation}
0 = \E \left[ \left. \left\{ U - \sum_{k = 0}^{t - 1} d_{k}^{\ast} (\bar{H}_{k}) T_{k} - \E \left[ \left. \left\{ U - \sum_{k = 0}^{t - 1} d_{k}^{\ast} (\bar{H}_{k}) T_{k} \right\} T_{t}^{\top} \right\vert \bar{H}_{t} \right] \E \left[ \left. T_{t} T_{t}^{\top} \right\vert \bar{H}_{t} \right]^{-1} T_{t} \right\} T_{j}^{\top} \right\vert \bar{H}_{j} \right] \label{main_3}
\end{equation}
Defining $T_{k}^{(t - 1)} \equiv \left( T_{k} - \E \left[ \left. T_{k} T_{t}^{\top} \right\vert \bar{H}_{t} \right] \left\{ \E \left[ \left. T_{t} T_{t}^{\top} \right\vert \bar{H}_{t} \right] \right\}^{-1} T_{t} \right)$ for $k = 0, \ldots, t - 1$, the last equation is 
\begin{eqnarray}
0 &=& \E \left[ \left. \left\{ U - \E \left[ \left. U T_{t}^{\top} \right\vert \bar{H}_{t} \right] \left\{ \E \left[ \left. T_{t} T_{t}^{\top} \right\vert \bar{H}_{t} \right] \right\}^{-1} T_{t} - \sum_{k = 0}^{t - 1} d_{k}^{\ast} (\bar{H}_{k}) T_{k}^{(t - 1)} \right\} T_{j}^{\top} \right\vert \bar{H}_{j} \right] \notag \\
&=& \E \left[ \left. \left\{ U - \sum_{k = 0}^{t - 1} d_{k}^{\ast} (\bar{H}_{k}) T_{k}^{(t - 1)} \right\} \left\{ T_{j} - \E \left[ \left. T_{j} T_{t}^{\top} \right\vert \bar{H}_{t} \right] \left\{ \E \left[ \left. T_{t} T_{t}^{\top} \right\vert \bar{H}_{t} \right] \right\}^{-1} T_{t} \right\}^{\top} \right\vert \bar{H}_{j} \right] \notag \\
&\equiv& \E \left[ \left. \left\{ U - \sum_{k = 0}^{t - 1} d_{k}^{\ast} (\bar{H}_{k}) T_{k}^{(t - 1)} \right\} T_{j}^{(t - 1) \top} \right\vert \bar{H}_{j} \right]. \label{main_5}
\end{eqnarray}

The system of equations \eqref{main_5} for $j = 0, \ldots, t - 1$ holds if and only if 
\begin{equation*}
\E \left[ \left\{ U - \sum_{k = 0}^{t - 1} d_{k}^{\ast} (\bar{H}_{k}) T_{k}^{(t - 1)} \right\} \left\{ \sum_{k = 0}^{t - 1} d_{k} (\bar{H}_{k}) T_{k}^{(t - 1)} \right\} \right] = 0
\end{equation*}
for all $d_{0}, \ldots, d_{t - 1}$ such that $\cov \left[ d_{j} (\bar{H}_{j}) \right] < \infty$ for $j = 0, \ldots, t - 1$. Then, defining 
\begin{equation*}
\Gamma_{j, sub}^{\ast} \equiv \left\{ d_{j} (\bar{H}_{j}) T_{j}^{(t - 1)}: d_{j} (\bar{H}_{j}) \text{ any } 1 \times \delta _{j} \text{ vector with } \cov \left[ d_{j} (\bar{H}_{j}) \right] < \infty \right\}
\end{equation*}
we conclude that $\Pi \left[ U | \sum_{j = 0}^{t - 1} \Gamma_{j, sub}^{\ast} \right] = \sum_{k = 0}^{t - 1} d_{k}^{\ast} (\bar{H}_{k}) T_{k}^{(t - 1)}$. Then, for each $j = 0, \ldots, K - 1$, defining recursively for $l = t - 2, \ldots, j$, 
\begin{equation*}
T_{j}^{(l)} \equiv T_{j}^{(l + 1)} - \E \left[ \left. T_{j}^{(l + 1)} T_{l + 1}^{(l + 1) \top} \right\vert \bar{H}_{t + 1} \right] \left\{ \E \left[ \left. T_{l + 1}^{(l + 1)} T_{l + 1}^{(l + 1) \top} \right\vert \bar{H}_{l + 1} \right] \right\}^{-1} T_{l + 1}^{(l + 1)}
\end{equation*}
the inductive hypothesis implies that both equation \eqref{main_0_app} and 
\begin{equation}
d_{k}^{\ast} (\bar{H}_{k}) = \E \left[ \left. \left\{ U - \sum_{j = 0}^{k - 1} d_{j}^{\ast} (\bar{H}_{j}) T_{j}^{(k)} \right\} T_{k}^{(k) \top} \right\vert \bar{H}_{k} \right] \left\{ \E \left[ \left. T_{k}^{(k)} T_{k}^{(k) \top} \right\vert \bar{H}_{k} \right] \right\}^{-1} \notag
\end{equation}
for $k = 0, \ldots, t - 1$ hold. So, combining these identities with equation \eqref{dt_main} shows that the assertion of the theorem holds when $K = t$. This concludes the proof.
\end{proof}

\subsection{More details of Section \ref{sec:feasible}}
\label{app:drml}

We have the following remarks on the algorithm described in Section \ref{sec:feasible}.

\begin{remark}
\label{rem:cf} 
Estimation of unconditional expectations in $\beta_{opt} (\bm{q}, b, \Psi)$ from the estimation sample implements a (least squares) projection under the estimation sample's empirical measure and thus may improve the finite sample performance of $\tilde{\Psi}_{cf} (\bm{q}, b)$. Also we chose to split the data in half to facilitate exposition. One can also divide the data into $M > 2$ approximately equal-sized groups, construct $M$ separate estimators by using each group as the estimation sample and the other $M - 1$ groups as the nuisance sample, and finally obtain $\tilde{\Psi}_{cf} (\bm{q}, b)$ as the average of the $M$ different estimators.
\end{remark}

\begin{remark}
\label{rem:close} 
Because conditional expectations of functions that depend on $\Psi$ must be estimated in step (ii) of the algorithm, the estimator $\tilde{\Psi}^{(1)} (\bm{q}, b)$ must, in general, be solved iteratively and may be difficult to compute. There are several strategies to decrease the computational task. First if $\Psi_{t}$ and $\Psi_{t‘}$ are variation independent for $t \neq t’$, then one can recursively, for $t = K, \ldots, 0$, estimate $\Psi_{t}^{\ast}$, given the earlier $\hat{\Psi}_{t^{\prime}}$ for $t^{\prime} > t$. Then, if, for each $t$, $\gamma_{t} (\bar{H}_{t}, S_{t}, A_{t}; \Psi_{t})$ is a linear model, i.e., 
\begin{equation*}
\gamma_{t} (\bar{H}_{t}, S_{t}, A_{t}; \Psi_{t}) = \Psi_{t}^{\top} V_{t}, \text{ for a given vector } V_{t} = v_{t} (\bar{H}_{t}, S_{t}, A_{t}),
\end{equation*}
one can avoid iteration and construct a closed form estimator $\tilde{\Psi}^{(1)} (\bm{q}, b)$ by estimating $\E \left[ \left. \Delta_{t} (\Psi _{t}; \underaccent{\bar}{\Psi}_{t + 1}) \right\vert \bar{H}_{t} \right]$ as 
\begin{equation*}
\hat{\E} \left[ \left. \Delta_{t} (\Psi_{t}; \underaccent{\bar}{\hat{\Psi}}_{t + 1}) \right\vert \bar{H}_{t} \right] = \hat{\E} \left[ \left. Y + \sum_{m = t + 1}^{K} \gamma_{m} \left( \bar{H}_{m}, S_{m}^{opt} (\underaccent{\bar}{\hat{\Psi}}_{m + 1}), A_{m}^{opt} (\underaccent{\bar}{\hat{\Psi}}_{m + 1}); \hat{\Psi}_{m} \right) \right\vert \bar{H}_{t} \right] - \Psi_{t}^{\top} \hat{\E} \left[ \left. V_{t} \right\vert \bar{H}_{t} \right].
\end{equation*}
The resulting closed form estimator $\tilde{\Psi}^{(1)} (\bm{q}, b)$ never requires, for any $t$, that one estimate a conditional expectation of a function of a free (i.e. yet to be estimated) $\Psi_{t}$ and thus is easy to compute and analyze. The explicit formula for the closed form estimator in a simple example is given in Appendix \ref{app:close_obs}. One drawback of this approach is that $\tilde{\Psi}^{(1)} (\bm{q}, b)$ may fail to estimate $\Psi^{\ast}$ at an optimal rate if $\E \left[ \left. \Delta_{t} (\Psi_{t}^{\ast}; \underaccent{\bar}{\Psi}_{t + 1}^{\ast}) \right\vert \bar{H}_{t} \right]$ can be estimated at a faster rate than $\E \left[ \left. V_{t} \right\vert \bar{H}_{t} \right]$. In that case, one can improve the rate of estimation by repeating the algorithm for $\tilde{\Psi}_{cf}$ with the unknown functions in step (ii) evaluated at $\tilde{\Psi}^{(1)} (\bm{q}, b)$ rather at the free parameter $\Psi$. See \citet{kallus2019localized} for an extended discussion and theoretical details.
\end{remark}

The following theorem provides conditions under which asymptotic properties of the proposed estimators above hold are summarized in the theorem below.

\begin{theorem}
\label{thm:drml} 
The nuisance sample is denoted as $\mathsf{Nu}$. If a) $\E [\hat{\BU} (\bm{q}, \Psi^{\ast}) | \mathsf{Nu}]$ and $\sum\limits_{t = 0}^{K} \E [\hat{T}_{b, t} | \mathsf{Nu}]$ are $o_{p} (n_{1}^{-1/2})$, and therefore $\E [ \hat{\BU} (\bm{q}, b, \Psi^{\ast}) | \mathsf{Nu}]$ is $o_{p} (n_{1}^{-1/2})$ and b) all the estimated nuisance conditional expectations and density functions converge to their true values in $L_{2} (P)$, then 
\begin{eqnarray*}
\tilde{\Psi}^{(1)} (\bm{q}, b) - \Psi^{\ast} = n_{1}^{-1} \sum_{i=1}^{n_{1}} \mathsf{IF}_{i} + o_{p} (n_{1}^{-1/2}), \tilde{\Psi}_{cf} (\bm{q}, b) - \Psi^{\ast} = n^{-1} \sum_{i=1}^{n} \mathsf{IF}_{i} + o_{p} (n^{-1/2})
\end{eqnarray*}
where $\mathsf{IF} \equiv \mathsf{IF} (\bm{q}, b, \Psi^{\ast}) = J^{-1} (\BU (\bm{q}, \Psi^{\ast}) - \beta_{opt} (\bm{q}, b, \Psi^{\ast})^{\top} T_{b}) - \Psi^{\ast}$ is the influence function of $\tilde{\Psi}^{(1)} (\bm{q}, b)$, $\hat{\Psi}^{(1)} (\bm{q}, b)$, and $\tilde{\Psi}_{cf} (\bm{q}, b)$ and $J = \frac{\partial}{\partial \Psi^{\top}} \E [\BU (\bm{q})] \vert_{\Psi = \Psi^{\ast}}$. Further $n_{1}^{1/2} (\tilde{\Psi}^{(1)} (\bm{q}, b) - \Psi^{\ast})$ converges conditionally and unconditionally to a normal distribution with mean zero and $\tilde{\Psi}_{cf} (\bm{q}, b)$ is a RAL estimator of $\Psi^{\ast}$ with asymptotic variance equal to $\var [\mathsf{IF}]$ for $\Psi^{\ast}$ in model $\mathbb{M}$.
\end{theorem}

\begin{remark}[A nearly efficient estimator $\tilde{\Psi}^{(1)} (\bm{q}, \hat{b}_{sub})$]
\label{rem:feasible} 

It follows under the sufficient conditions described above the estimator $\tilde{\Psi}^{(1)} (\bm{q}, b_{sub})$ (or equivalently $\tilde{\Psi}^{(2)} (\bm{q}, b_{sub})$) with $b_{sub} (\bm{q})$ defined in Section \ref{sec:optimal_dr} is RAL with asymptotic variance $V_{oracle} (\bm{q}, b_{sub})$ which, as discussed earlier, is exactly or nearly equal to $V_{oracle} (\bm{q}, b_{opt})$ depending on whether $Y$ has finite or continuous support. However $\tilde{\Psi}^{(1)} (\bm{q}, b_{sub})$ is not a feasible estimator because $b_{sub} (\bm{q}) \equiv (b_{sub, 0} (\bm{q}), \ldots, b_{sub, K} (\bm{q}))^{\top}$ depends on the unknown nuisance functions of the distribution generating the data. Therefore let $\hat{b}_{sub} (\bm{q})$ be $b_{sub} (\bm{q})$ except with the unknown nuisance functions replaced by estimates computed from the nuisance sample. Then, under the weak condition that $\hat{b}_{sub} (\bm{q})$ converges to $b_{sub} (\bm{q})$ in $L_{2}(P)$, $\tilde{\Psi}^{(1)} (\bm{q}, \hat{b}_{sub})$ will be asymptotically equivalent to $\tilde{\Psi}^{(1)} (\bm{q}, b_{sub})$ with the same influence function and thus the same asymptotic variance. Therefore $\tilde{\Psi}^{(1)} (\bm{q}, \hat{b}_{sub})$ is the estimator that we recommend when the NDE assumption is true.
\end{remark}

Now we prove Theorem \ref{thm:drml} and provide sufficient conditions on the nuisance parameters to guarantee $\sqrt{n}$-consistency of the opt-SNMM estimators, before and after adjusting for the NDE assumption.

\begin{proof}
To do so we first need to derive the conditional biases $\E [\hat{\BU} (\bm{q}, \Psi) \vert \mathsf{Nu}]$, $\E [\hat{\BU} (\bm{q}, b, \Psi) \vert \mathsf{Nu}]$ and $\E [\hat{T}_{b, t} \vert \mathsf{Nu}]$ given the nuisance sample ($\mathsf{Nu}$) of the estimating equations $\hat{\BU} (\bm{q}, \Psi)$, $\hat{\BU} (\bm{q}, b, \Psi)$ and $\hat{T}_{b, t}$ as estimators of zero.

\begin{lemma}
\label{lem:unified_bias} 
\begin{equation}
\E \left[ \left. \hat{\BU} (\bm{q}, \Psi) \right\vert \mathsf{Nu} \right] = \sum_{t = 0}^{K} \E \left[ \left( \hat{\E} - \E \right) \left[ \left. \Delta_{t} \left( \Psi_{t}; \underaccent{\bar}{\Psi}_{t + 1} \right) \right\vert \bar{H}_{t} \right] \left( \hat{\E} - \E \right) \left[ Q_{t} (S_{t}, A_{t}) | \bar{H}_{t} \right] | \mathsf{Nu} \right] \label{bias:u}
\end{equation}
and 
\begin{equation*}
\E \left[ \hat{\BU} (\bm{q}, b, \Psi) | \mathsf{Nu} \right] = \E \left[ \hat{\BU} (\bm{q}, \Psi) | \mathsf{Nu} \right] - \hat{\beta}_{opt} \left( \bm{q}, b, \Psi \right)^{\top} \sum_{t = 0}^{K} \E \left[ \hat{T}_{b, t} | \mathsf{Nu} \right]
\end{equation*}
where 
\begin{align}
& \E \left[ \hat{T}_{b, t} | \mathsf{Nu} \right] \label{bias:tbt} \\
& = \sum_{j = t + 1}^{K + 1} \sum_{m = t + 1}^{j - 1} \E \left[
\left. \left\{ \begin{array}{c}
\dfrac{1}{\hat{W}_{t + 1}^{m - 1}} \left( \dfrac{1}{\hat{p} (S_{m} |  \bar{H}_{m})} - \dfrac{1}{p (S_{m} | \bar{H}_{m})} \right) \dfrac{1}{W_{m + 1}^{j - 1}} \\ 
\times \left\{ \E \left[ y_{j, \hat{\eta}_{j, t}} (\bar{H}_{j}) | \bar{H}_{j - 1}, S_{j - 1} \right] - \hat{\eta}_{j - 1, t} (\bar{H}_{j - 1}, S_{j - 1}) \right\}
\end{array} \right\} \left( A_{t} - \hat{\Pi}_{t} \right) \right\vert \mathsf{Nu} \right] \notag \\
& + \sum_{j = t + 1}^{K + 1} \E \left[ \left. \frac{1}{W_{t + 1}^{j - 1}} \left\{ \E \left[ \left. y_{j, \hat{\eta}_{j, t}} (\bar{H}_{j}) \right\vert \bar{H}_{j - 1}, S_{j - 1} \right] - \hat{\eta}_{j - 1, t} (\bar{H}_{j - 1}, S_{j - 1}) \right\} \left( \Pi_{t} - \hat{\Pi}_{t} \right) \right\vert \mathsf{Nu} \right].  \notag
\end{align}
\end{lemma}

The proof of the above lemma is given in Appendix \ref{app:bias}. Then we have the following theorem, the proof of which is standard; See for example \citet{chernozhukov2018double}.
\end{proof}

\subsubsection{Proof sketch of Lemma \ref{lem:unified_bias}}
\label{app:bias}

In this section, we sketch the proof of Lemma \ref{lem:unified_bias}.

\begin{proof}
The decomposition form of $\E \left[ \hat{\BU} \left( \bm{q}, \Psi \right) \vert \mathsf{Nu} \right]$ is trivial to obtain so we omit the derivation.

In terms of $\E \left[ \hat{T}_{b, t} \vert \mathsf{Nu} \right]$, it is easy to show that 
\begin{align*}
& \ \E \left[ \hat{T}_{b, t} \vert \mathsf{Nu} \right] \\
= & \ \underbrace{\E \left[ \left. \left\{ \sum_{j = t + 1}^{K} \left( \dfrac{1}{\hat{W}_{t + 1}^{j - 1}} - \frac{1}{W_{t + 1}^{j - 1}} \right) \left( \E \left[ y_{j, \hat{\eta}_{j, t}} | \bar{H}_{j - 1}, S_{j - 1} \right] - \hat{\eta}_{j - 1, t} \left( \bar{H}_{j - t}, S_{j - 1} \right) \right) \right\} \left( A_{t} - \hat{\Pi}_{t} \right) \right\vert \mathsf{Nu} \right]}_{(\star)} \\
& + \E \left[ \left. \left\{ \sum_{j = t + 1}^{K + 1} \dfrac{1}{W_{t + 1}^{j - 1}} \left( \E \left[ y_{j, \hat{\eta}_{j, t}} (\bar{H}_j) |\bar{H}_{j - 1}, S_{j - 1} \right] - \hat{\eta}_{j - 1, t} (\bar{H}_{j - t}, S_{j - 1}) \right) \right\} \left( \Pi_{t} - \hat{\Pi}_{t} \right) \right\vert \mathsf{Nu} \right].
\end{align*}
Then 
\begin{align*}
(\star) & = \sum_{j = t + 1}^{K + 1} \sum_{m = t + 1}^{j - 1} \E \left[ \left. \left\{ \begin{array}{c}
\dfrac{1}{\hat{W}_{t + 1}^{m - 1}} \left( \dfrac{1}{\hat{p} (S_{m} | \bar{H}_{m})} - \dfrac{1}{p (S_{m} | \bar{H}_{m})} \right) \dfrac{1}{W_{m + 1}^{j - 1}} \\ 
\times \left\{ \E \left[ y_{j, \hat{\eta}_{j, t}} (\bar{H}_{j}) | \bar{H}_{j - 1}, S_{j - 1} \right] - \hat{\eta}_{j - 1, t} (\bar{H}_{j - 1}, S_{j - 1}) \right\}%
\end{array} \right\} \left( A_{t} - \hat{\Pi}_{t} \right) \right\vert \mathsf{Nu} \right]
\end{align*}
which follows from the identity after equation (62) appeared in \citet[Section 5.2, page 60]{rotnitzky2017multiply}.
\end{proof}

\subsubsection{Sufficient conditions for the bias to be $o_{p} (n^{-1/2})$}

We now briefly discuss when the two bias terms $\E [\hat{\mathbb{U}} (\bm{q}, \Psi^{\ast}) | \mathsf{Nu}]$ and $\sum_{t = 0}^{K} \E [\hat{T}_{b, t} | \mathsf{Nu}]$ are $o_{p} (n^{-1/2})$. The first bias term can be controlled as follows: for any $\Psi$ 
\begin{align*}
\vert \E [\hat{\BU} (\bm{q}, \Psi) | \mathsf{Nu}] \vert \leq \sum_{t = 0}^{K + 1} \left\Vert (\hat{\E} - \E) \left[ \Delta_{t} \left( \Psi_{t}; \underaccent{\bar}{\Psi}_{t + 1} \right) | \bar{H}_{t} \right] \right\Vert \left\Vert (\hat{\E} - \mathsf{E}) \left[ Q_{t} \left( S_{t}, A_{t} \right) | \bar{H}_{t} \right]
\right\Vert
\end{align*}
where $\Vert \cdot \Vert$ denotes the $L_{2} (P)$-norm. Then a sufficient condition under which $\vert \E [\hat{\BU} (\bm{q}, \Psi^{\ast}) | \mathsf{Nu}] \vert$ is $o_{p} (n^{-1/2})$ is 
\begin{equation*}
\max_{t = 0, \ldots, K} \left\Vert \left( \hat{\E} - \E \right) \left[ \Delta_{t} \left( \Psi^{\ast}_{t}; \underaccent{\bar}{\Psi}^{\ast}_{t + 1} \right) | \bar{H}_{t} \right] \right\Vert \left\Vert \left( \hat{\E} - \E \right) \left[ Q_{t} \left( S_{t}, A_{t} \right) | \bar{H}_{t} \right] \right\Vert = o_{p} (n^{-1/2}).
\end{equation*}

That is, for every $t$, the rate of convergence in $L_{2} (P)$ of $\hat{\E} [\Delta_{t} (\Psi^{\ast}_{t}; \underaccent{\bar}{\Psi}^{\ast}_{t + 1}) | \bar{H}_{t}]$ to $\E [\Delta_{t} (\Psi^{\ast}_{t}; \underaccent{\bar}{\Psi}^{\ast}_{t + 1}) | \bar{H}_{t}]$ multiplied by rate of convergence of $\hat{\E} [Q_{t} (S_{t}, A_{t}) | \bar{H}_{t}]$ to $\E [Q_{t} (S_{t}, A_{t}) | \bar{H}_{t}]$ be $o_{p} (n^{-1/2})$. A bias that can be expressed as the sum of the product of the errors in the estimation of two different nuisance functions is referred to as rate double robustness by \citet{smucler2019unifying}.

A sufficient condition under which $\left\vert \sum_{t = 0}^{K} \E [\hat{T}_{b, t} | \mathsf{Nu}] \right\vert$ is $o_{p} (n^{-1/2})$ is 
\begin{align*}
& \ \left\vert \sum_{t = 0}^{K} \E [\hat{T}_{b, t} | \mathsf{Nu}] \right\vert \\
\leq & \ \sum_{t = 0}^{K} \sum_{j = t + 1}^{K + 1} \frac{1}{\varrho^{j - t - 1}} \left\Vert (\E - \hat{\E}) [y_{j, \hat{\eta}_{j, t}} | \bar{H}_{j - 1}, S_{j - 1}] \right\Vert \left\{ \begin{array}{c}
\Vert \Pi_{t} - \hat{\Pi}_{t} \Vert \\ 
+ \dfrac{\sup | A_{t} - \hat{\Pi}_{t} |}{\varrho^{-1}} \sum\limits_{m = t + 1}^{j - 1} \left\Vert \dfrac{1}{\hat{p} (S_{m} | \bar{H}_{m})} - \dfrac{1}{p (S_{m} | \bar{H}_{m})} \right\Vert 
\end{array} \right\} \\
= & \ o_{p} (n^{-1/2})
\end{align*}
where $\varrho$ is some constant that upper bounds $\max \{\max_{t = 0, \ldots, K} \{p (S_{t} | \bar{H}_{t})\}, \max_{t = 0, \ldots, K} \{\hat{p} (S_{t} | \bar{H}_{t})\}\}$, which is implied by the positivity condition \eqref{id1}. That is, the rates of convergence of $\hat{p} (S_{m} | \bar{H}_{m})$ to $p (S_{m} | \bar{H}_{m})$ and $\hat{\Pi}_{t}$ to $\Pi_{t}$ times the rate of convergence of $\hat{\eta}_{j - 1, t} (\bar{H}_{j - 1}, S_{j - 1}) = \hat{\E} \left[ y_{j, \hat{\eta}_{j, t}} (\bar{H}_{j}) | \bar{H}_{j - 1}, S_{j - 1} \right]$ to $\E \left[ y_{j, \hat{\eta}_{j, t}} (\bar{H}_{j}) | \bar{H}_{j - 1}, S_{j - 1} \right]$ is $o_{p} (n^{-1/2})$. Thus $\sum\limits_{t = 0}^{K} \E [\hat{T}_{b, t} | \mathsf{Nu}]$ is also rate doubly robust.

\subsection{The actual regimes of \citet{caniglia2019emulating}}
\label{app:actual}

The actual regimes used by \citet{caniglia2019emulating} were as follows. The two testing regimes $g_{x}$, required simultaneous tests for both HIV RNA and CD4 count every 9-12 months while RNA $<$ 200 copies/ml and CD4 $> x$ counts/ml for $x \in \{350, 500\}$; otherwise the tests are performed every 3-6 months. The two treatment regimes $g_{z}$ required one to switch to second line ART therapy 0-3 months after the first time the measured RNA level exceeds $z$ copies for $z \in \{200, 1000\}$. Let $g_{x, z} = (g_{x}, g_{z}), x \in \{350, 500\}, z \in \{200, 100\}$ denote the four regimes. The four-month intervals 9-12 months, 3-6 months, and 0-3 months are referred to as grace periods. Regimes with grace periods are often more clinically realistic and relevant than those without, as in practice, owing to logistical and medical problems, physicians require some flexibility in the timing of tests and treatment switches. However, a regime that includes a grace period is not well-defined without precisely specifying how the grace period is to be implemented. For example, in their implementation, \citet{caniglia2019emulating} assigned a marginal probability of 1/4 to each of the 4 months. Thus, regimes with grace periods are random rather than deterministic regimes.The simplified regimes considered in the main text eliminated the grace periods. Readers can consult \citet{caniglia2019emulating, robins2008estimation, neugebauer2017identification, kreif2020exploiting} for details concerning the analysis of random regimes. Due to such simplification, the counterfactual testing and treatment vector $(S_{m, g}, A_{m, g}) = g_{m} (\bar{H}_{m, g})$ is a deterministic function $g_{m} (\cdot)$ of the counterfactual history $\bar{H}_{m, g}$. We made other simplifications in the main text to avoid additional irrelevant complications. We simplified the dyn-MSM conditioned on baseline covariates actually used by \citet{caniglia2019emulating} by restricting consideration to saturated MSMs that do not condition on baseline covariates.

\subsection{Closed form expression for $\Lambda_{\mathsf{NDE}, obs}^{c, r^{\ast} \perp}$}
\label{app:close_obs}

\allowdisplaybreaks

In this section, we study the ortho-complement $\Lambda_{\mathsf{NDE}, obs}^{c, r^{\ast} \perp}$ to the tangent space in the model $\mathbb{M}_{\mathsf{NDE}, obs} (\bar{C}_{K}, \bar{R}_{K}^{\ast}, Y^{d})$ for the observed data induced by full data model $\mathbb{M}_{\mathsf{NDE}} (\bar{C}_{K}, \bar{R}_{K}^{\ast}, Y^{d})$ in which $\bar{R}_{K}^{\ast}$ is fully observed regardless of testing history. As defined in Section \ref{sec:nde_variety}, $\mathbb{M}_{\mathsf{NDE}} (\bar{C}_{K}, \bar{R}_{K}^{\ast}, Y^{d})$ asserts the truth of the following modified sequential randomization assumption and conditional independence for each $t = 0, \cdots, K$:
\begin{equation}
\begin{split} \label{modified}
(Y_{g}^{d}, Y_{g}, \underaccent{\bar}{R}_{t + 1,g}, \underaccent{\bar}{L}_{t + 1, g}, \underaccent{\bar}{R}_{t + 1, g}^{\ast}, \underaccent{\bar}{A}_{t, g}, \underaccent{\bar}{S}_{t, g}) & \amalg (A_{t}, S_{t}) | \bar{L}_{t}, \bar{R}_{t}, \bar{R}_{t}^{\ast}, (\bar{A}_{t - 1}, \bar{S}_{t - 1}) = \bar{g}_{t - 1} (\bar{H}_{t - 1}) \ \ \forall t, g, \\
\bar{R}_{t}^{\ast} & \amalg (S_{t}, A_{t}) | \bar{H}_{t}.
\end{split}
\end{equation}

Let $\bar{i}_{t - 1} \equiv \{i_{0}, \cdots, i_{t - 1}\}$ and $\underaccent{\bar}{i}_{t} \equiv \{i_{t}, \cdots, i_{K - 1}\}$ index arbitrary subsets of $\{0, 1\}^{t}$ and $\{0, 1\}^{K - t}$ respectively and $\bar{R}_{\bar{i}_{t - 1}}^{\ast} \equiv \{R_{k + 1}^{\ast}: 0 \leq k \leq t - 1, i_{k} = 1\}$ and $\underaccent{\bar}{R}_{\underaccent{\bar}{i}_{t}}^{\ast} \equiv \{R_{k + 1}^{\ast}: t \leq k \leq K - 1, i_{k} = 1\}$. Then it is easy to see that for every such subset $\bar{i}_{t - 1}$, if equations \eqref{modified} hold, they continue to hold when $\bar{R}_{t}^{\ast}$ is replaced by the subvector $\bar{R}_{\bar{i}_{t - 1}}^{\ast}$.

As a corollary of Theorem \ref{lem:nuisance_nde} in the main text, we can now characterize the ortho-complement $\Lambda_{\mathsf{NDE}}^{c, r^{\ast} \perp}$ to the tangent space of the model defined by equations \eqref{modified}. To do so, as in Section \ref{sec:nde}, we define for $t = 0, \cdots, K$,

\begin{equation*}
\mathcal{B}_{t}^{\text{full}} = \left\{ B_{t} \equiv b_{t} (\bar{R}_{t}^{\ast}, \bar{H}_{t}, \underaccent{\bar}{C}_{t + 1}, \underaccent{\bar}{R}_{t + 1}^{\ast}, \underaccent{\bar}{S}_{t}, Y^{d}): B_{t} \in L_{2} (P) \right\}
\end{equation*}
and for any $B_{t}$ in $\mathcal{B}_{t}^{\text{full}}$ we define
\begin{equation*}
D_{b, t}^{\text{full}} \equiv \frac{B_{t}}{W_{t + 1}} \left( \frac{A_{t}}{\Pi_{t}} - 1 \right).
\end{equation*}

\begin{corollary}
\begin{equation*}
\Lambda_{\mathsf{NDE}}^{c, r^{\ast} \perp} = \left\{ \sum_{t = 0}^{K - 1} D_{b, t}^{\text{full}} - \Pi [D_{b, t}^{\text{full}} \vert \Lambda_{\mathsf{ancillary}}]; B_{t} \in \mathcal{B}_{t}^{\text{full}}, t = 0, \cdots, K \right\}.
\end{equation*}
\end{corollary}

Note that if we adopt the convention $0 / 0 \equiv 1$, $D_{b, t}^{\text{full}}$ is well-defined and identically 0 at $t = K$ (when $A_{K} = \Pi_{K} = 0$). Therefore $\sum\limits_{t = 0}^{K - 1} D_{b, t}^{\text{full}} - \Pi [D_{b, t}^{\text{full}} \vert \Lambda_{\mathsf{ancillary}}] \equiv \sum\limits_{t = 0}^{K} D_{b, t}^{\text{full}} - \Pi [D_{b, t}^{\text{full}} \vert \Lambda_{\mathsf{ancillary}}]$ and we may write either $K$ or $K - 1$ for convenience.

\begin{lemma}
\label{lem:rewrite} 
Any $B_{t} \in \mathcal{B}_{t}^{\text{full}}$ can be rewritten as
\begin{equation*}
B_{t} = \sum_{\bar{i}_{t - 1} \in \{0, 1\}^{t}, \underaccent{\bar}{i}_{t} \in \{0, 1\}^{K - t}} b_{t, \bar{i}_{t - 1}, \underaccent{\bar}{i}_{t}} (\bar{R}_{\bar{i}_{t - 1}}^{\ast}, \bar{H}_{t}, \underaccent{\bar}{C}_{t + 1}, \underaccent{\bar}{R}_{\underaccent{\bar}{i}_{t}}^{\ast}, \underaccent{\bar}{S}_{t}, Y^{d})
\end{equation*}
for some functions $b_{t, \bar{i}_{t - 1}, \underaccent{\bar}{i}_{t}}$ that depend on $\bar{R}_{t}^{\ast}$ and $\underaccent{\bar}{R}_{t + 1}^{\ast}$ through $\bar{R}_{\bar{i}_{t - 1}}^{\ast}$ and $\underaccent{\bar}{R}_{\underaccent{\bar}{i}_{t}}^{\ast}$ only.
\end{lemma}

\begin{remark}
Since $\underaccent{\bar}{R}_{t + 1}^{\ast} \equiv \{R_{t + 1}^{\ast}, \cdots, R_{K}^{\ast}\}$ is not necessarily observed, $D_{b,t}^{\text{full}}$ and, consequently, the elements of $\Lambda_{\mathsf{NDE}}^{c, r^{\ast} \perp}$ may not be statistics.
\end{remark}

Next define
\begin{equation*}
\mathcal{B}_{t, \bar{i}_{t - 1}, \underaccent{\bar}{i}_{t + 1}} \coloneqq \left\{ B_{t, \bar{i}_{t - 1}, \underaccent{\bar}{i}_{t + 1}} \equiv b_{t, \bar{i}_{t - 1}, \underaccent{\bar}{i}_{t + 1}} (\bar{R}_{\bar{i}_{t - 1}}^{\ast}, \bar{H}_{t}, \underaccent{\bar}{C}_{t + 1}, \underaccent{\bar}{R}_{\underaccent{\bar}{i}_{t + 1}}^{\ast}, \underaccent{\bar}{S}_{t}, Y^{d}): B_{t, \bar{i}_{t - 1}, \underaccent{\bar}{i}_{t + 1}} \in L_{2} (P) \right\}
\end{equation*}
and for any $B_{t, \bar{i}_{t - 1}, \underaccent{\bar}{i}_{t + 1}}$ in $\mathcal{B}_{t, \bar{i}_{t - 1}, \underaccent{\bar}{i}_{t + 1}}$ define 
\begin{equation*}
\mathcal{D}_{t, \bar{i}_{t - 1}, \underaccent{\bar}{i}_{t + 1}}^{\text{obs}} \equiv \left\{ D_{b, t, \bar{i}_{t - 1}, \underaccent{\bar}{i}_{t + 1}}^{\text{obs}} \equiv \frac{B_{t, \bar{i}_{t - 1}, \underaccent{\bar}{i}_{t + 1}}}{W_{t + 1}} \left( \frac{A_{t}}{\Pi_{t}} - 1 \right) \prod_{t' = 0}^{t - 1} \left( \frac{A_{t'}}{\Pi_{t'}} \right)^{i_{t'}} \prod_{t' = t + 1}^{K - 1} \left( \frac{A_{t'}}{\Pi_{t'}} \right)^{i_{t'}}: B_{t, \bar{i}_{t - 1}, \underaccent{\bar}{i}_{t + 1}} \in \mathcal{B}_{t, \bar{i}_{t - 1}, \underaccent{\bar}{i}_{t + 1}} \right\}
\end{equation*}
and 
\begin{equation*}
\Upsilon^{\text{obs}} \equiv \left\{ \sum_{t = 0}^{K - 1} \sum_{\bar{i}_{t - 1} \in \{0, 1\}^{t}, \underaccent{\bar}{i}_{t + 1} \in \{0, 1\}^{K - t - 1}} \left\{ D_{b, t, \bar{i}_{t - 1}, \underaccent{\bar}{i}_{t + 1}}^{\text{obs}} - \Pi [D_{b, t, \bar{i}_{t - 1}, \underaccent{\bar}{i}_{t + 1}}^{\text{obs}} | \Lambda_{\mathsf{ancillary}}] \right\}: D_{b, t, \bar{i}_{t - 1}, \underaccent{\bar}{i}_{t + 1}}^{\text{obs}} \in \mathcal{D}_{t, \bar{i}_{t - 1}, \underaccent{\bar}{i}_{t + 1}}^{\text{obs}} \right\}.
\end{equation*}

\begin{theorem}
\label{thm:obs_nde}
For $K \geq 1$, $\Upsilon^{\text{obs}}$ is identical to the ortho-complement $\Lambda_{\mathsf{NDE}, obs}^{c, r^{\ast} \perp}$ to the tangent space of model $\mathbb{M}_{\mathsf{NDE}, obs} (\bar{C}_{K}, \bar{R}_{K}^{\ast}, Y^{d})$.
\end{theorem}

\begin{remark}
Notice that $D_{b, t, \bar{i}_{t - 1}, \underaccent{\bar}{i}_{t + 1}}^{\text{obs}}$ is a function of the observed data (and so is every member of $\Upsilon^{\text{obs}}$) because its value is unchanged if we replace $\bar{R}_{\bar{i}_{t - 1}}^{\ast}$ and $\underaccent{\bar}{R}_{\underaccent{\bar}{i}_{t + 1}}^{\ast}$ by $\bar{R}_{\bar{i}_{t - 1}}$ and $\underaccent{\bar}{R}_{\underaccent{\bar}{i}_{t + 1}}$. Also note that the functions $B_{t, \bar{i}_{t - 1}, \underaccent{\bar}{i}_{t + 1}}$ occurring in $\mathcal{D}_{t, \bar{i}_{t - 1}, \underaccent{\bar}{i}_{t + 1}}^{\text{obs}}$ do not depend on $R_{t + 1}^{\ast}$. Had $\mathcal{D}_{t, \bar{i}_{t - 1}, \underaccent{\bar}{i}_{t + 1}}^{\text{obs}}$ included a function $B_{t, \bar{i}_{t - 1}, \underaccent{\bar}{i}_{t + 1}}$ that depended on $R_{t + 1}^{\ast}$, it would have been necessary to include a factor $A_{t} / \Pi_{t}$ to make $D_{t, \bar{i}_{t - 1}, \underaccent{\bar}{i}_{t + 1}}^{\text{obs}}$ a function of the observed data, in which case $D_{t, \bar{i}_{t - 1}, \underaccent{\bar}{i}_{t + 1}}^{\text{obs}}$ would include the product $(A_{t} / \Pi_{t}) (A_{t} / \Pi_{t} - 1)$. Since this product does not have mean zero it follows that the mean of $D_{t, \bar{i}_{t - 1}, \underaccent{\bar}{i}_{t + 1}}^{\text{obs}}$ would also be nonzero.
\end{remark}

\begin{proof}
\allowdisplaybreaks
First, we observe that $\prod\limits_{t' = 0}^{t - 1} \left( \frac{A_{t'}}{\Pi_{t'}} \right)^{i_{t'}}$ is a function of $\bar{H}_{t}$. Thus, we can rewrite $D_{b, t, \bar{i}_{t - 1}, \underaccent{\bar}{i}_{t + 1}}^{\text{obs}}$ as 
\begin{equation*}
D_{b, t, \bar{i}_{t - 1}, \underaccent{\bar}{i}_{t + 1}}^{\text{obs}} = \frac{\widetilde{B}_{t, \bar{i}_{t - 1}, \underaccent{\bar}{i}_{t + 1}}}{W_{t + 1}} \left( \frac{A_{t}}{\Pi_{t}} - 1 \right) \prod_{t' = t + 1}^{K - 1} \left( \frac{A_{t'}}{\Pi_{t'}} \right)^{i_{t'}}
\end{equation*}
where $\widetilde{B}_{t, \bar{i}_{t - 1}, \underaccent{\bar}{i}_{t + 1}} \equiv B_{t, \bar{i}_{t - 1}, \underaccent{\bar}{i}_{t + 1}} \prod\limits_{t' = 0}^{t - 1} \left( \frac{A_{t'}}{\Pi_{t'}} \right)^{i_{t'}}$. Next write, 
\begin{align*}
\prod_{t' = t + 1}^{K - 1} \left( \frac{A_{t'}}{\Pi_{t'}} \right)^{i_{t'}} & = \prod_{t' = t + 1}^{K - 1} \left\{ \left( \frac{A_{t'}}{\Pi_{t'}} - 1 \right) + 1 \right\}^{i_{t'}} \\
& = 1 + \sum_{\mathfrak{s} \subset \{t + 1 \leq k \leq K - 1: i_{k} = 1\}, \mathfrak{s} \neq \emptyset} \prod_{m \in \mathfrak{s}} \left( \frac{A_{m}}{\Pi_{m}} - 1 \right).
\end{align*}
Then, by denoting $W_{t_{1}}^{t_{2}} \equiv \prod\limits_{t = t_{1}}^{t_{2}} p (S_{t} | \bar{H}_{t})$ for any $t_{1} < t_{2}$,
\begin{align*}
D_{b, t, \bar{i}_{t - 1}, \underaccent{\bar}{i}_{t + 1}}^{\text{obs}} & = \frac{\widetilde{B}_{t, \bar{i}_{t - 1}, \underaccent{\bar}{i}_{t + 1}}}{W_{t + 1}} \left( \frac{A_{t}}{\Pi_{t}} - 1 \right) + \sum_{\mathfrak{s} \subset \{t + 1 \leq k \leq K - 1: i_{k} = 1\}, \mathfrak{s} \neq \emptyset} \frac{\widetilde{B}_{t, \bar{i}_{t - 1}, \underaccent{\bar}{i}_{t + 1}}}{W_{t + 1}^{\bar{\mathfrak{s}}}} \frac{1}{W_{\bar{\mathfrak{s}} + 1}} \left( \frac{A_{t}}{\Pi_{t}} - 1 \right) \prod_{m \in \mathfrak{s}} \left( \frac{A_{m}}{\Pi_{m}} - 1 \right) \\
& = \frac{\widetilde{B}_{t, \bar{i}_{t - 1}, \underaccent{\bar}{i}_{t + 1}}}{W_{t + 1}} \left( \frac{A_{t}}{\Pi_{t}} - 1 \right) + \sum_{\mathfrak{s} \subset \{t + 1 \leq k \leq K - 1: i_{k} = 1\}, \mathfrak{s} \neq \emptyset} \frac{\widetilde{B}_{\bar{\mathfrak{s}}, \bar{i}_{\bar{\mathfrak{s}}}, \underaccent{\bar}{i}_{\bar{\mathfrak{s}}}}}{W_{\bar{\mathfrak{s}} + 1}} \left( \frac{A_{\bar{\mathfrak{s}}}}{\Pi_{\bar{\mathfrak{s}}}} - 1 \right)
\end{align*}
where $\bar{\mathfrak{s}} \equiv \max \{m: m \in \mathfrak{s}\}$ and 
\begin{equation*}
\widetilde{B}_{\bar{\mathfrak{s}}, \bar{i}_{\bar{\mathfrak{s}}}, \underaccent{\bar}{i}_{\bar{\mathfrak{s}}}} \equiv \widetilde{B}_{t, \bar{i}_{t - 1}, \underaccent{\bar}{i}_{t + 1}} \frac{1}{W_{t + 1}^{\bar{\mathfrak{s}}}} \left( \frac{A_{t}}{\Pi_{t}} - 1 \right) \prod_{m \in \mathfrak{s}, m \neq \bar{\mathfrak{s}}} \left( \frac{A_{m}}{\Pi_{m}} - 1 \right).
\end{equation*}
Next, noticing that $\widetilde{B}_{\bar{\mathfrak{s}}, \bar{i}_{\bar{\mathfrak{s}} - 1}, \underaccent{\bar}{i}_{\bar{\mathfrak{s}}}}$ is a function, say $\widetilde{b}_{\bar{\mathfrak{s}}, \bar{i}_{\bar{\mathfrak{s}} - 1}, \underaccent{\bar}{i}_{\bar{\mathfrak{s}}}}$, of 
\begin{equation*}
(\bar{R}_{\bar{i}_{\bar{\mathfrak{s}} - 1}}^{\ast}, \bar{H}_{\bar{\mathfrak{s}}}, \underaccent{\bar}{C}_{\bar{\mathfrak{s}} + 1}, \underaccent{\bar}{R}_{\underaccent{\bar}{i}_{\bar{\mathfrak{s}}}}^{\ast}, \underaccent{\bar}{S}_{\bar{\mathfrak{s}}}, Y^{d})
\end{equation*}
we conclude that $D_{b, t, \bar{i}_{t - 1}, \underaccent{\bar}{i}_{t + 1}}^{\text{obs}}$ is of the form $\sum\limits_{t=0}^{K}D_{b,t}^{\text{full}}$ and consequently $$\sum_{t = 0}^{K} \sum_{\bar{i}_{t - 1}, \underaccent{\bar}{i}_{t + 1}} \left\{ D_{b, t, \bar{i}_{t - 1}, \underaccent{\bar}{i}_{t + 1}}^{\text{obs}} - \Pi [D_{b, t, \bar{i}_{t - 1}, \underaccent{\bar}{i}_{t + 1}}^{\text{obs}} | \Lambda_{\mathsf{ancillary}}] \right\} \in \Lambda_{\mathsf{NDE}}^{c, r^{\ast} \perp}.$$
Furthermore, $D_{b, t, \bar{i}_{t - 1}, \underaccent{\bar}{i}_{t + 1}}^{\text{obs}}$ is a function of the observed data $O = (\bar{H}_{K}, S_{K}, Y^{d})$. Hence $\Upsilon^{\text{obs}} \subset \Lambda_{\mathsf{NDE}, \text{obs}}^{c, r^{\ast} \perp}$.

Turn now to the proof that $\Lambda_{\mathsf{NDE}, \text{obs}}^{c, r^{\ast} \perp} \subset \Upsilon^{\text{obs}}$. We show this inclusion for the case $K = 2$ in detail and sketch the proof for general $K$ as it does not involve new ideas but requires cumbersome notation. Let 
\begin{align*}
\mathcal{D}^{\text{full}} \equiv & \ \left\{ D_{b, 0}^{\text{full}} + D_{b, 1}^{\text{full}}: B_{t} \in \mathcal{B}_{t}^{\text{full}}, t = 0, 1 \right\} 
\\
\equiv & \ \left\{ \frac{B_{0}}{p (S_{1} | \bar{H}_{1}) p (S_{2} | \bar{H}_{2})} \left( \frac{A_{0}}{\Pi_{0}} - 1 \right) + \frac{B_{1}}{p (S_{2} | \bar{H}_{2})} \left( \frac{A_{1}}{\Pi_{1}} - 1 \right): \right. \\
& \left. B_{0} = b_{0} (L_{0}, (C_{1}, C_{2}), (R_{1}^{\ast}, R_{2}^{\ast}), (S_{0}, S_{1}, S_{2}), Y^{d}) \in \mathcal{B}_{0}^{\text{full}}\text{ and } \right. \\
& \left. B_{1} = b_{1} (R_{1}^{\ast}, \bar{H}_{1}, C_{2}, R_{2}^{\ast}, (S_{1}, S_{2}), Y^{d}) \in \mathcal{B}_{1}^{\text{full}} \right\}.
\end{align*}
When $K = 2$, we have 
\begin{equation*}
\Lambda_{\mathsf{NDE}}^{c, r^{\ast} \perp} = \left\{ D_{b, 0}^{\text{full}} + D_{b, 1}^{\text{full}} - \Pi [D_{b, 0}^{\text{full}} + D_{b, 1}^{\text{full}} | \Lambda_{\mathsf{ancillary}}]: D_{b, 0}^{\text{full}} + D_{b, 1}^{\text{full}} \in \mathcal{D}^{\text{full}} \right\} 
\end{equation*}
To show that $\Lambda_{\mathsf{NDE}, \text{obs}}^{c, r^{\ast} \perp} \subset \Upsilon^{\text{obs}}$ it suffices to show that any element of $\mathcal{D}^{\text{full}}$ that depends only on the observed data $O$, is an element of the set 
\begin{eqnarray*}
\mathcal{D}^{\text{obs}} &\equiv& \left\{ \sum_{t = 0}^{1} \sum_{\bar{i}_{t - 1}, \underaccent{\bar}{i}_{t + 1}} D_{b, t, \bar{i}_{t - 1}, \underaccent{\bar}{i}_{t + 1}}^{\text{obs}}: D_{b, t, \bar{i}_{t - 1}, \underaccent{\bar}{i}_{t + 1}}^{\text{obs}} \in \mathcal{D}_{t, \bar{i}_{t - 1}, \underaccent{\bar}{i}_{t + 1}}^{\text{obs}} \right\} \\
&=& \left\{ D^{\text{obs}}_{\bar{b}} \equiv \frac{\bar{b}_{0, \emptyset, 1} (L_{0}, (C_{1}, C_{2}), R_{2}^{\ast}, (S_{0}, S_{1}, S_{2}), Y^{d})}{p (S_{1} | \bar{H}_{1}) p (S_{2} | \bar{H}_{2})} \left( \frac{A_{0}}{\Pi_{0}} - 1 \right) \frac{A_{1}}{\Pi_{1}} \right.  \\
&& \left. + \frac{\bar{b}_{0, \emptyset, 0} (L_{0}, (C_{1}, C_{2}), (S_{0}, S_{1}, S_{2}), Y^{d})}{p (S_{1} | \bar{H}_{1}) p (S_{2} | \bar{H}_{2})} \left( \frac{A_{0}}{\Pi_{0}} - 1 \right) \right. \\
&& \left. + \frac{\bar{b}_{1, 1, \emptyset} (R_{1}^{\ast}, \bar{H}_{1}, C_{2}, (S_{1}, S_{2}), Y^{d})}{p (S_{2} | \bar{H}_{2})} \frac{A_{0}}{\Pi_{0}}\left( \frac{A_{1}}{\Pi_{1}} - 1 \right) \right. \\
&& \left. + \frac{\bar{b}_{1, 0, \emptyset} (\bar{H}_{1}, C_{2}, (S_{1}, S_{2}), Y^{d})}{p (S_{2} | \bar{H}_{2})} \left( \frac{A_{1}}{\Pi_{1}} - 1 \right): \bar{B}_{0, \emptyset, 1}, \bar{B}_{0, \emptyset, 0}, \bar{B}_{1, 0, \emptyset}, \bar{B}_{1, 1, \emptyset} \text{ arbitrary} \right\}.
\end{eqnarray*}
To show this we first note that $\mathcal{D}^{\text{full}}$ is equal to the set
\begin{align*}
\widetilde{\mathcal{D}}^{\text{full}} \equiv \left\{ \begin{array}{c}
\widetilde{D}_{\tilde{b}}^{\text{full}} \equiv \dfrac{\widetilde{B}_{0, \emptyset}}{p (S_{1} | \bar{H}_{1}) p (S_{2} | \bar{H}_{2})} \left( \dfrac{A_{0}}{\Pi_{0}} - 1 \right) \dfrac{A_{1}}{\Pi_{1}} \\
+ \dfrac{\widetilde{B}_{1, 1}}{p (S_{2} | \bar{H}_{2})} \dfrac{A_{0}}{\Pi_{0}} \left( \dfrac{A_{1}}{\Pi_{1}} - 1 \right) + \dfrac{\widetilde{B}_{1, 0}}{p (S_{2} | \bar{H}_{2})} \left( \dfrac{A_{0}}{\Pi_{0}} - 1 \right) \left( \dfrac{A_{1}}{\Pi_{1}} - 1 \right): \\
\widetilde{B}_{0, \emptyset} \equiv \widetilde{b}_{0, \emptyset} (L_{0}, (C_{1}, C_{2}), (R_{1}^{\ast}, R_{2}^{\ast}), (S_{0}, S_{1}, S_{2}), Y^{d}) \in \mathcal{B}_{0}^{\text{full}}, \\
\widetilde{B}_{1, 1} \equiv \widetilde{b}_{1, 1} (R_{1}^{\ast}, \bar{H}_{1}, C_{2}, R_{2}^{\ast}, (S_{1}, S_{2}), Y^{d}) \text{ and }\widetilde{B}_{1, 0} \equiv \widetilde{b}_{1, 0} (R_{1}^{\ast}, \bar{H}_{1}, C_{2}, R_{2}^{\ast}, (S_{1}, S_{2}), Y^{d}) \in \mathcal{B}_{1}^{\text{full}}
\end{array} \right\} 
\end{align*}
The inclusion $\mathcal{D}^{\text{full}} \subset \widetilde{\mathcal{D}}^{\text{full}}$ follows because given $B_{0} \equiv b_{0} (L_{0}, (C_{1}, C_{2}), (R_{1}^{\ast}, R_{2}^{\ast}), (S_{0}, S_{1}, S_{2}), Y^{d})$ and $B_{1} \equiv b_{1} (R_{1}^{\ast}, \bar{H}_{1}, C_{2}, R_{2}^{\ast}, (S_{1}, S_{2}), Y^{d})$ we can write 
\begin{align*}
& \ \frac{B_{0}}{p (S_{1} | \bar{H}_{1}) p (S_{2} | \bar{H}_{2})} \left( \frac{A_{0}}{\Pi_{0}} - 1 \right) + \frac{B_{1}}{p (S_{2} | \bar{H}_{2})} \left( \frac{A_{1}}{\Pi_{1}} - 1 \right) \\
= & \ \frac{\tilde{B}_{0, \emptyset}}{p (S_{1} | \bar{H}_{1}) p (S_{2} | \bar{H}_{2})} \left( \frac{A_{0}}{\Pi_{0}} - 1 \right) \frac{A_{1}}{\Pi_{1}} + \frac{\widetilde{B}_{1, 1}}{p (S_{2} | \bar{H}_{2})} \frac{A_{0}}{\Pi_{0}} \left( \frac{A_{1}}{\Pi_{1}} - 1 \right) + \frac{\widetilde{B}_{1, 0}}{p (S_{2} | \bar{H}_{2})} \left( \frac{A_{0}}{\Pi_{0}} - 1 \right) \left( \frac{A_{1}}{\Pi_{1}} - 1 \right) 
\end{align*}
with
\begin{align*}
\tilde{B}_{0, \emptyset} & = B_{0}, \\
\widetilde{B}_{1, 1} & = - \widetilde{B}_{1, 0} = B_{1} - \frac{B_{0}}{p (S_{1} | \bar{H}_{1})} \left( \frac{A_{0}}{\Pi_{0}} - 1 \right) = B_{1} - \frac{B_{0}}{W_{1}^{1}} \left( \frac{A_{0}}{\Pi_{0}} - 1 \right).
\end{align*}
The inclusion $\widetilde{\mathcal{D}}^{\text{full}} \subset \mathcal{D}^{\text{full}}$ follows by writing 
\begin{align*}
& \ \frac{\widetilde{B}_{0, \emptyset}}{p (S_{1} | \bar{H}_{1}) p (S_{2} | \bar{H}_{2})} \left( \frac{A_{0}}{\Pi_{0}} - 1 \right) \frac{A_{1}}{\Pi_{1}} + \frac{\widetilde{B}_{1, 1}}{p (S_{2} | \bar{H}_{2})} \frac{A_{0}}{\Pi_{0}} \left( \frac{A_{1}}{\Pi_{1}} - 1 \right) + \frac{\widetilde{B}_{1, 0}}{p (S_{2} | \bar{H}_{2})} \left( \frac{A_{0}}{\Pi_{0}} - 1 \right) \left( \frac{A_{1}}{\Pi_{1}} - 1 \right) \\
= & \ \frac{\widetilde{B}_{0, \emptyset}}{p (S_{1} | \bar{H}_{1}) p (S_{2} | \bar{H}_{2})} \left( \frac{A_{0}}{\Pi_{0}} - 1 \right) \left( \frac{A_{1}}{\Pi_{1}} - 1 + 1 \right) + \frac{\widetilde{B}_{1, 1}}{p (S_{2} | \bar{H}_{2})} \frac{A_{0}}{\Pi_{0}} \left( \frac{A_{1}}{\Pi_{1}} - 1 \right) \\
& + \frac{\widetilde{B}_{1, 0}}{p (S_{2} | \bar{H}_{2})} \left( \frac{A_{0}}{\Pi_{0}} - 1 \right) \left( \frac{A_{1}}{\Pi_{1}} - 1 \right) \\
= & \ \frac{\widetilde{B}_{0, \emptyset}}{p (S_{1} | \bar{H}_{1}) p (S_{2} | \bar{H}_{2})} \left( \frac{A_{0}}{\Pi_{0}} - 1 \right) \\
& + \left\{ \frac{\widetilde{B}_{0, \emptyset}}{p (S_{1} | \bar{H}_{1}) p (S_{2} | \bar{H}_{2})} \left( \frac{A_{0}}{\Pi_{0}} - 1 \right) + \frac{\widetilde{B}_{1, 1}}{p (S_{2} | \bar{H}_{2})} \frac{A_{0}}{\Pi_{0}} + \frac{\widetilde{B}_{1, 0}}{p (S_{2} | \bar{H}_{2})} \left( \frac{A_{0}}{\Pi_{0}} - 1 \right) \right\} \left( \frac{A_{1}}{\Pi_{1}} - 1 \right) \\
= & \ \frac{\widetilde{B}_{0, \emptyset}}{p (S_{1} | \bar{H}_{1}) p (S_{2} | \bar{H}_{2})} \left( \frac{A_{0}}{\Pi_{0}} - 1 \right) + \frac{\frac{\widetilde{B}_{0, \emptyset}}{p (S_{1} | \bar{H}_{1})} \left( \frac{A_{0}}{\Pi_{0}} - 1 \right) + \tilde{B}_{1, 1} \frac{A_{0}}{\Pi_{0}} + \tilde{B}_{1, 0} \left( \frac{A_{0}}{\Pi_{0}} - 1 \right)}{p (S_{2} | \bar{H}_{2})} \left( \frac{A_{1}}{\Pi_{1}} - 1 \right)
\end{align*}
and thus we can set $B_{0} \equiv \widetilde{B}_{0, \emptyset}$ and $B_{1} \equiv \frac{\widetilde{B}_{0, \emptyset}}{p (S_{1} | \bar{H}_{1})} \left( \frac{A_{0}}{\Pi_{0}} - 1 \right) + \tilde{B}_{1, 1} \frac{A_{0}}{\Pi_{0}} + \tilde{B}_{1, 0} \left( \frac{A_{0}}{\Pi_{0}} - 1 \right)$.

Now, suppose that $\widetilde{D}_{\tilde{b}} \in \widetilde{\mathcal{D}}^{\text{full}}$ depends only on the observed data $O$. Then, evaluating $\widetilde{D}_{\tilde{b}}$ at $A_{0} = A_{1} = 0$, we conclude that $\widetilde{b}_{1, 0} (R_{1}^{\ast}, \bar{H}_{1}, C_{2}, R_{2}^{\ast}, (S_{1}, S_{2}), Y^{d})$ is neither a function of $R_{1}^{\ast}$ nor of $R_{2}^{\ast}$. So, we can write,
\begin{equation*}
\widetilde{b}_{1, 0} (R_{1}^{\ast}, \bar{H}_{1}, C_{2}, R_{2}^{\ast}, (S_{1}, S_{2}), Y^{d}) =\widetilde{b}_{1, 0} (\bar{H}_{1}, C_{2}, (S_{1}, S_{2}), Y^{d}).
\end{equation*}
Next, evaluating at $A_{0} = 1, A_{1} = 0$ we conclude that 
\begin{equation*}
- \frac{\widetilde{b}_{1, 1} (R_{1}^{\ast}, \bar{H}_{1}, C_{2}, R_{2}^{\ast}, (S_{1}, S_{2}), Y^{d})}{p (S_{2} | \bar{H}_{2})} \frac{1}{\Pi_{0}} - \frac{\widetilde{b}_{1, 0} (\bar{H}_{1}, C_{2}, (S_{1}, S_{2}), Y^{d})}{p (S_{2} | \bar{H}_{2})} \left( \frac{1}{\Pi_{0}} - 1 \right)
\end{equation*}
cannot depend on $R_{2}^{\ast}$, from where we deduce that $\widetilde{b}_{1, 1} (R_{1}^{\ast}, \bar{H}_{1}, C_{2}, R_{2}^{\ast}, (S_{1}, S_{2}), Y^{d})$ does not depend on $R_{2}^{\ast}$, so we can write 
\begin{equation*}
\widetilde{b}_{1, 1} (R_{1}^{\ast}, \bar{H}_{1}, C_{2}, R_{2}^{\ast}, (S_{1}, S_{2}), Y^{d}) = \widetilde{b}_{1, 1} (R_{1}^{\ast}, \bar{H}_{1}, C_{2}, (S_{1}, S_{2}), Y^{d}).
\end{equation*}
Finally, evaluating at $A_{0} = 0, A_{1} = 1$, we conclude that 
\begin{equation*}
- \frac{\widetilde{b}_{0, \emptyset} (L_{0}, (C_{1}, C_{2}), (R_{1}^{\ast}, R_{2}^{\ast}), (S_{0}, S_{1}, S_{2}), Y^{d})}{p (S_{1} | \bar{H}_{1}) p (S_{2} | \bar{H}_{2})} \frac{1}{\Pi_{1}} + \frac{\widetilde{b}_{1, 0} (\bar{H}_{1}, C_{2}, (S_{1}, S_{2}), Y^{d})}{p (S_{2} | \bar{H}_{2})}\left( \frac{1}{\Pi_{1}} - 1 \right)
\end{equation*}
cannot depend on $R_{1}^{\ast}$, from where we deduce that $\widetilde{b}_{0, \emptyset} (L_{0}, (C_{1}, C_{2}), (R_{1}^{\ast}, R_{2}^{\ast}), (S_{0}, S_{1}, S_{2}), Y^{d})$ does not depend on $R_{1}^{\ast}$, so we can write
\begin{equation*}
\widetilde{b}_{0, \emptyset} (L_{0}, (C_{1}, C_{2}), (R_{1}^{\ast}, R_{2}^{\ast}), (S_{0}, S_{1}, S_{2}), Y^{d}) = \widetilde{b}_{0, \emptyset} (L_{0}, (C_{1}, C_{2}), R_{2}^{\ast}, (S_{0}, S_{1}, S_{2}), Y^{d}).
\end{equation*}

We have thus arrived at the conclusion that 
\begin{eqnarray*}
\widetilde{\mathcal{D}}^{\text{obs}} &\equiv &\left\{ \widetilde{D}_{\widetilde{b}} \in \widetilde{\mathcal{D}}^{\text{full}}: \widetilde{D}_{\widetilde{b}} \text{ depends only on the observed data } O \right\} \\
&=& \left\{ \widetilde{D}_{\widetilde{b}} \in \mathcal{\widetilde{\mathcal{D}}}^{\text{full}}: \widetilde{B}_{0, \emptyset} \text{ does not depend on } R_{1}^{\ast}, \widetilde{B}_{1, 1} \text{ does not depend on } R_{2}^{\ast}, \right. \\
&& \left. \widetilde{B}_{1, 0} \text{ does not depend on } (R_{1}^{\ast}, R_{2}^{\ast}) \right\}.
\end{eqnarray*}

Clearly, $\widetilde{\mathcal{D}}^{\text{obs}} \subset \mathcal{D}^{\text{obs}}$ since given $(\widetilde{B}_{0, \emptyset}, \widetilde{B}_{1, 1}, \widetilde{B}_{1, 0})$, $\widetilde{D}_{\widetilde{b}} = D_{\bar{b}}^{\text{obs}}$ for $\bar{B}_{0, \emptyset, 1} = \widetilde{B}_{0, \emptyset}$, $\bar{B}_{0, \emptyset, 0} = 0, \bar{B}_{1, 1, \emptyset} = \widetilde{B}_{1, 1}$ and $\bar{B}_{1, 0, \emptyset} = \widetilde{B}_{1, 0} \left( \frac{A_{0}}{\Pi_{0}} - 1 \right)$. This concludes the proof when $K = 2$. 

We now turn to the proof for general $K$. Again, to show that $\Lambda_{\mathsf{NDE}, \text{obs}}^{c, r^{\ast} \perp} \subset \Upsilon^{\text{obs}}$ it suffices to show that any element of $\mathcal{D}^{\text{full}}$ that depends only on the observed data $O$, is an element of the set
\begin{align*}
\mathcal{D}^{\text{obs}} \equiv \left\{ \sum_{t = 0}^{K - 1} \sum_{\bar{i}_{t - 1}, \underaccent{\bar}{i}_{t + 1}} D_{b, t, \bar{i}_{t - 1}, \underaccent{\bar}{i}_{t + 1}}^{\text{obs}}: D_{b, t, \bar{i}_{t - 1}, \underaccent{\bar}{i}_{t + 1}}^{\text{obs}} \in \mathcal{D}_{b, t, \bar{i}_{t - 1}, \underaccent{\bar}{i}_{t + 1}}^{\text{obs}} \right\}.
\end{align*}

First, we define $\mathcal{D}^{\text{full}}$ and $\tilde{\mathcal{D}}^{\text{full}}$ as follows:
\begin{align*}
\mathcal{D}^{\text{full}} & \equiv \left\{ \sum_{t = 0}^{K - 1} \dfrac{B_{t}}{W_{t + 1}} \left( \dfrac{A_{t}}{\Pi_{t}} - 1 \right): B_{t} \in \mathcal{B}_{t}^{\text{full}}, t = 0, 1, \cdots, K - 1 \right\}. \\
& \equiv \left\{ \sum_{t = 0}^{K - 1} \sum\limits_{\bar{i}_{t - 1} \in \{0, 1\}^{t}, \underaccent{\bar}{i}_{t} \in \{0, 1\}^{K - t}} \dfrac{B_{t, \bar{i}_{t - 1}, \underaccent{\bar}{i}_{t}}}{W_{t + 1}} \left( \dfrac{A_{t}}{\Pi_{t}} - 1 \right): B_{t} \in \mathcal{B}_{t}^{\text{full}}, t = 0, 1, \cdots, K - 1 \right\}
\end{align*}
and
\begin{align*}
\tilde{\mathcal{D}}^{\text{full}} \equiv \left\{ \begin{array}{c}
\sum\limits_{t = 0}^{K - 1} \sum\limits_{\bar{j}_{t - 1} \in \{0, 1\}^{t}} \dfrac{\tilde{B}_{t, \bar{j}_{t - 1}}}{W_{t + 1}} \left[ \prod\limits_{t' = 0}^{t - 1} \left( \dfrac{A_{t'}}{\Pi_{t'}} \right)^{j_{t'}} \left( \dfrac{A_{t'}}{\Pi_{t'}} - 1 \right)^{1 - j_{t'}} \right] \left( \dfrac{A_{t}}{\Pi_{t}} - 1 \right) \prod\limits_{t' = t + 1}^{K - 1} \dfrac{A_{t'}}{\Pi_{t'}}: \\
\tilde{B}_{t, \bar{j}_{t - 1}} \in \mathcal{B}_{t}^{\text{full}}, \forall \ \bar{j}_{t - 1} \in \{0, 1\}^{t} \text{ and } \forall \ t = 0, \cdots, K - 1
\end{array} \right\}.
\end{align*}
Notice that the subscript $\bar{j}_{t - 1}$ in $\tilde{B}_{t, \bar{j}_{t - 1}}$ does not indicate the dependence on subvector of $\bar{R}_{K}^{\ast}$ and it is simply an indexing subscript to distinguish different $\tilde{B}_{t}$'s. 

We now show that $\mathcal{D}^{\text{full}} = \tilde{\mathcal{D}}^{\text{full}}$. First, we show $\mathcal{D}^{\text{full}} \supset \tilde{\mathcal{D}}^{\text{full}}$. Take any $\tilde{D}^{\text{full}} \in \tilde{\mathcal{D}}^{\text{full}}$, it can be rewritten as: denoting $x_{t} \equiv A_{t} / \Pi_{t} - 1$
\begin{align*}
& \ \tilde{D}^{\text{full}} \\
= & \ \sum\limits_{t = 0}^{K - 1} \sum\limits_{\bar{j}_{t - 1} \in \{0, 1\}^{t}} \dfrac{\tilde{B}_{t, \bar{j}_{t - 1}}}{W_{t + 1}} \left[ \prod\limits_{t' = 0}^{t - 1} \left( \dfrac{A_{t'}}{\Pi_{t'}} \right)^{j_{t'}} \left( \dfrac{A_{t'}}{\Pi_{t'}} - 1 \right)^{1 - j_{t'}} \right] \left( \dfrac{A_{t}}{\Pi_{t}} - 1 \right) \prod\limits_{t' = t + 1}^{K - 1} \dfrac{A_{t'}}{\Pi_{t'}} \\
= & \ \sum\limits_{t = 0}^{K - 1} \sum\limits_{\bar{j}_{t - 1} \in \{0, 1\}^{t}} \dfrac{\tilde{B}_{t, \bar{j}_{t - 1}}}{W_{t + 1}} \left[ \prod\limits_{t' = 0}^{t - 1} (x_{t'} + j_{t'}) \right] x_{t} \prod\limits_{t' = t + 1}^{K - 1} (x_{t'} + 1) \\
= & \ \sum\limits_{t = 0}^{K - 1} \sum\limits_{\bar{j}_{t - 1} \in \{0, 1\}^{t}} \dfrac{\tilde{B}_{t, \bar{j}_{t - 1}}}{W_{t + 1}} \left[ \prod\limits_{t' = 0}^{t - 1} (x_{t'} + j_{t'}) \right] x_{t} \left\{ 1 + \sum\limits_{\mathfrak{s}_{t} \subset \{t + 1, \cdots, K - 1\}, \mathfrak{s} \neq \emptyset} \prod_{m \in \mathfrak{s}} x_{m} \right\} \\
= & \ \sum\limits_{t = 0}^{K - 1} \sum\limits_{\bar{j}_{t - 1} \in \{0, 1\}^{t}} \dfrac{\tilde{B}_{t, \bar{j}_{t - 1}}}{W_{t + 1}} \left[ \prod\limits_{t' = 0}^{t - 1} (x_{t'} + j_{t'}) \right] \left\{ x_{t} + \sum\limits_{\mathfrak{s}_{t} \subset \{t + 1, \cdots, K - 1\}, \mathfrak{s} \neq \emptyset} x_{t} \prod_{m \in \mathfrak{s}} x_{m} \right\} \\
= & \ \sum\limits_{t = 0}^{K - 1} \left\{ \begin{array}{c}
\dfrac{\underbrace{\sum\limits_{\bar{j}_{t - 1} \in \{0, 1\}^{t}} \tilde{B}_{t, \bar{j}_{t - 1}} \prod\limits_{t' = 0}^{t - 1} (x_{t'} + j_{t'})}_{=: \bar{B}_{t}}}{W_{t + 1}} x_{t} \\
+ \sum\limits_{\mathfrak{s}_{t} \subset \{t + 1, \cdots, K - 1\}, \mathfrak{s}_{t} \neq \emptyset} \dfrac{\underbrace{\sum\limits_{\bar{j}_{t - 1} \in \{0, 1\}^{t}} \frac{\tilde{B}_{t, \bar{j}_{t - 1}}}{W_{t + 1}^{\bar{\mathfrak{s}}_{t}}} \prod\limits_{t' = 0}^{t - 1} (x_{t'} + j_{t'}) x_{t} \prod\limits_{m \in \mathfrak{s}_{t}, m \neq \bar{\mathfrak{s}}_{t}} x_{m}}_{=: \bar{B}_{\bar{\mathfrak{s}}_{t} + 1}^{\dag}}}{W_{\bar{\mathfrak{s}}_{t} + 1}} x_{\bar{\mathfrak{s}}_{t}}
\end{array} \right\} \\
\equiv & \ \sum\limits_{t = 0}^{K - 1} \left\{ \frac{\bar{B}_{t}}{W_{t + 1}} x_{t} + \sum\limits_{\mathfrak{s}_{t} \subset \{t + 1, \cdots, K - 1\}, \mathfrak{s}_{t} \neq \emptyset} \frac{\bar{B}_{\mathfrak{s}_{t}}^{\dag}}{W_{\bar{\mathfrak{s}}_{t} + 1}} x_{\bar{\mathfrak{s}}_{t}} \right\} \\
\equiv & \ \sum\limits_{t = 0}^{K - 1} \frac{B_{t}}{W_{t + 1}} x_{t} \equiv \sum\limits_{t = 0}^{K - 1} \frac{B_{t}}{W_{t + 1}} \left( \frac{A_{t}}{\Pi_{t}} - 1 \right)
\end{align*}
is of the form $D^{\text{full}}$, where in the last line $B_{t} \equiv \bar{B}_{t} + \sum\limits_{t' \leq t} \sum\limits_{\mathfrak{s}_{t' - 1} \subset \{t', \cdots, K - 1\}, \mathfrak{s}_{t' - 1} \neq \emptyset, \bar{\mathfrak{s}}_{t' - 1} = t} \bar{B}_{\mathfrak{s}_{t' - 1}}^{\dag}$.

Second, we show $\mathcal{D}^{\text{full}} \subset \tilde{\mathcal{D}}^{\text{full}}$. Take any $D^{\text{full}} \in \mathcal{D}^{\text{full}}$. In this case, given fixed $B_{0}, \cdots, B_{K - 1}$, we want to find the corresponding $\tilde{B}$'s in $\tilde{\mathcal{D}}^{\text{full}}$. This can be done by identifying the coefficients in front of the terms $x_{t}, t = 0, \cdots, K - 1$ in both $D^{\text{full}}$ and $\tilde{D}^{\text{full}}$, which leads to the following linear system of equations (arranged in the order from $t = 0$ to $t = K - 1$):
\begin{align*}
\left\{ \begin{array}{c}
\bar{B}_{0} \equiv \tilde{B}_{0} = B_{0} \\
\bar{B}_{1} + \bar{B}_{\mathfrak{s}_{0} = \{1\}}^{\dag} \equiv \{\tilde{B}_{1, 0} x_{0} + \tilde{B}_{1, 1} (x_{0} + 1)\} + \dfrac{\tilde{B}_{0} x_{0}}{W_{1}^{1}} = B_{1} \\
\boldsymbol{\vdots} \\
\bar{B}_{t} + \sum\limits_{t' \leq t} \sum\limits_{\mathfrak{s}_{t' - 1} \subset \{t', \cdots, K - 1\}, \mathfrak{s}_{t' - 1} \neq \emptyset, \bar{\mathfrak{s}}_{t' - 1} = t} \bar{B}_{\mathfrak{s}_{t' - 1}}^{\dag} = B_{t} \\
\boldsymbol{\vdots} \\
\bar{B}_{K - 1} + \sum\limits_{t \leq K - 1} \sum\limits_{\mathsf{s}_{t - 1} \subset \{t, \cdots, K - 1\}: \bar{\mathfrak{s}}_{t - 1} = K - 1} \bar{B}_{\mathfrak{s}_{t - 1}}^{\dag} = B_{K - 1}.
\end{array} \right.
\end{align*}
The above system of linear equations is underdetermined so may have multiple roots. Any root gives us appropriate $\tilde{B}$'s such that $D^{\text{full}} \in \tilde{\mathcal{D}}^{\text{full}}$. Note that when $K = 2$, the above system of equations reduces to
\begin{align*}
\left\{ \begin{array}{c}
\bar{B}_{0} \equiv \tilde{B}_{0} = B_{0}, \\
\bar{B}_{1} + \bar{B}_{\mathfrak{s}_{0} = \{1\}}^{\dag} \equiv \{\tilde{B}_{1, 0} x_{0} + \tilde{B}_{1, 1} (x_{0} + 1)\} + \dfrac{\tilde{B}_{0} x_{0}}{W_{1}^{1}} = B_{1},
\end{array} \right.
\end{align*}
one solution of which is
\begin{align*}
\left\{ \begin{array}{l}
\tilde{B}_{0} = B_{0}, \\
\tilde{B}_{1, 1} = B_{1} - \frac{\tilde{B}_{0} x_{0}}{W_{1}^{1}}, \\
\tilde{B}_{1, 0} = - \tilde{B}_{1, 1},
\end{array} \right.
\end{align*}
and this is exactly the solution that was obtained above when $K = 2$.

Next as in the case of $K = 2$, $\tilde{\mathcal{D}}^{\text{obs}}$, the subset of $\tilde{\mathcal{D}}^{\text{full}}$ that depends only on the observed data, can be written as
\begin{align*}
\tilde{\mathcal{D}}^{\text{obs}} & \equiv \left\{ \tilde{D} \in \tilde{\mathcal{D}}^{\text{full}}: \text{ $\tilde{D}$ depends only on the observed data} \right\} \\
& = \left\{ \begin{array}{c}
\sum\limits_{t = 0}^{K - 1} \sum\limits_{\bar{i}_{t - 1} \in \{0, 1\}^{t}} \dfrac{\tilde{B}_{t, \bar{i}_{t - 1}}}{W_{t + 1}} \left[ \prod\limits_{t' = 0}^{t - 1} \left( \dfrac{A_{t'}}{\Pi_{t'}} \right)^{i_{t'}} \left( \dfrac{A_{t'}}{\Pi_{t'}} - 1 \right)^{1 - i_{t'}} \right] \left( \dfrac{A_{t}}{\Pi_{t}} - 1 \right) \prod\limits_{t' = t + 1}^{K - 1} \dfrac{A_{t'}}{\Pi_{t'}}: \\
\tilde{B}_{t, \bar{i}_{t - 1}} \in \mathcal{B}_{t}^{\text{full}} \text{ does not depend on $R_{j + 1}^{\ast}$, $\forall \ j$ such that $i_{j} = 0$}, \\
\forall \ \bar{i}_{t - 1} = \{i_{0}, \cdots, i_{t - 1}\} \in \{0, 1\}^{t} \text{ and } \forall \ t = 0, \cdots, K - 1
\end{array} \right\}.
\end{align*}

We are left to show $\tilde{\mathcal{D}}^{\text{obs}} \subset \mathcal{D}^{\text{obs}}$, which is true because we can identify $$\dfrac{\tilde{B}_{t, \bar{i}_{t - 1}}}{W_{t + 1}} \left[ \prod\limits_{t' = 0}^{t - 1} \left( \dfrac{A_{t'}}{\Pi_{t'}} \right)^{i_{t'}} \left( \dfrac{A_{t'}}{\Pi_{t'}} - 1 \right)^{1 - i_{t'}} \right] \left( \dfrac{A_{t}}{\Pi_{t}} - 1 \right) \prod\limits_{t' = t + 1}^{K - 1} \dfrac{A_{t'}}{\Pi_{t'}}$$ with one $D_{b, t, \bar{i}_{t - 1}, \underaccent{\bar}{i}_{t + 1}}^{\text{obs}}$, by setting $B_{t, \bar{i}_{t - 1}, \underaccent{\bar}{i}_{t + 1}} \equiv \tilde{B}_{t, \bar{i}_{t - 1}} \prod\limits_{t' = 0}^{t - 1} \left( \dfrac{A_{t'}}{\Pi_{t'}} - 1 \right)^{1 - i_{t'}}$ if $\underaccent{\bar}{i}_{t + 1} = \{1\}^{K - t - 1}$ and setting $B_{t, \bar{i}_{t - 1}, \underaccent{\bar}{i}_{t + 1}} \equiv 0$ if otherwise. This concludes the proof that $\Lambda_{\mathsf{NDE}, \text{obs}}^{c, r^{\ast} \perp} \subset \Upsilon^{\text{obs}}$.

\end{proof}

With the characterization of $\Lambda_{\mathsf{NDE}, obs}^{c, r^{\ast} \perp}$, we can also define its subspace $\Omega_{obs}^{c, r^{\ast}}$ that facilitates computation as in Section \ref{sec:optimal_dr} in the main text, such that for any random variable $U$, $\Pi (U | \Omega_{obs}^{c, r^{\ast}})$ has closed form expression. To this end, denote $\bm{\varphi}_{t, \underaccent{\bar}{i}_{t + 1}} = (\varphi_{1, t, \underaccent{\bar}{i}_{t + 1}}, \cdots, \varphi_{\xi, t, \underaccent{\bar}{i}_{t + 1}})^{\top}$ as a $\xi$-dimensional transformation (basis) of the vector $(\underaccent{\bar}{C}_{t + 1}, \underaccent{\bar}{R}_{\underaccent{\bar}{i}_{t + 1}}^{\ast}, Y^{d})$ and define
\begin{equation*}
b_{t, \underaccent{\bar}{i}_{t + 1}}^{\dag} (\underaccent{\bar}{C}_{t + 1}, \underaccent{\bar}{R}_{\underaccent{\bar}{i}_{t + 1}}^{\ast}, \underaccent{\bar}{S}_{t}, Y^{d}) = (\varphi_{1, t, \underaccent{\bar}{i}_{t + 1}} (\underaccent{\bar}{C}_{t + 1}, \underaccent{\bar}{R}_{\underaccent{\bar}{i}_{t + 1}}^{\ast}, Y^{d}) I_{t}^{\top}, \cdots, \varphi_{\xi, t, \underaccent{\bar}{i}_{t + 1}} (\underaccent{\bar}{C}_{t + 1}, \underaccent{\bar}{R}_{\underaccent{\bar}{i}_{t + 1}}^{\ast}, Y^{d}) I_{t}^{\top})^{\top}
\end{equation*}
where $I_{t}$ is defined in Section \ref{sec:optimal_dr} of the main text. Then define 
\begin{align*}
\mathcal{B}_{t, \bar{i}_{t - 1}, \underaccent{\bar}{i}_{t + 1}}^{\dag} \coloneqq \left\{ B_{t, \bar{i}_{t - 1}, \underaccent{\bar}{i}_{t + 1}}^{\dag} \equiv d_{t, \bar{i}_{t - 1}} (\bar{R}_{\bar{i}_{t - 1}}^{\ast}, \bar{H}_{t}) b_{t, \underaccent{\bar}{i}_{t + 1}}^{\dag} (\underaccent{\bar}{C}_{t + 1}, \underaccent{\bar}{R}_{\underaccent{\bar}{i}_{t + 1}}^{\ast}, \underaccent{\bar}{S}_{t}, Y^{d}): B_{t, \bar{i}_{t - 1}, \underaccent{\bar}{i}_{t + 1}}^{\dag} \in L_{2} (P) \right\}
\end{align*}
where $d_{t, \bar{i}_{t - 1}} (\bar{R}_{\bar{i}_{t - 1}}^{\ast}, \bar{H}_{t})$ is a row vector with dimension $\xi$. Further define
$$
\tilde{d}_{t, \bar{i}_{t - 1}} (\bar{R}_{\bar{i}_{t - 1}}^{\ast}, \bar{H}_{t}) = d_{t, \bar{i}_{t - 1}} (\bar{R}_{\bar{i}_{t - 1}}^{\ast}, \bar{H}_{t}) \prod\limits_{t' = 0}^{t - 1} \left( \dfrac{A_{t'}}{\Pi_{t'}} \right)^{i_{t'}}
$$
and notice that $\tilde{d}_{t, \bar{i}_{t - 1}} (\bar{R}_{\bar{i}_{t - 1}}^{\ast}, \bar{H}_{t}) = \tilde{d}_{t, \bar{i}_{t - 1}} (\bar{R}_{\bar{i}_{t - 1}}, \bar{H}_{t}) = \widetilde{d}_{t, \bar{i}_{t - 1}}^{\dag} (\bar{H}_{t})$ for some function $\widetilde{d}_{t, \bar{i}_{t - 1}}^{\dag}$.

Then, define
\begin{align*}
\mathcal{D}_{t, \bar{i}_{t - 1}, \underaccent{\bar}{i}_{t + 1}}^{\text{obs} \dag} & \equiv \left\{ D_{t, \bar{i}_{t - 1}, \underaccent{\bar}{i}_{t + 1}}^{\text{obs} \dag} \equiv \frac{B_{t, \underaccent{\bar}{i}_{t + 1}}^{\dag}}{W_{t + 1}} \left( \frac{A_{t}}{\Pi_{t}} - 1 \right) \prod\limits_{t' = 0}^{t - 1} \left( \dfrac{A_{t'}}{\Pi_{t'}} \right)^{i_{t'}} \prod\limits_{t' = t + 1}^{K - 1} \left( \dfrac{A_{t'}}{\Pi_{t'}} \right)^{i_{t'}}: B_{t, \bar{i}_{t - 1}, \underaccent{\bar}{i}_{t + 1}}^{\dag} \in \mathcal{B}_{t, \bar{i}_{t - 1}, \underaccent{\bar}{i}_{t + 1}}^{\dag} \right\} \\
& \equiv \left\{ D_{t, \bar{i}_{t - 1}, \underaccent{\bar}{i}_{t + 1}}^{\text{obs} \dag} \equiv \tilde{d}_{t, \bar{i}_{t - 1}}^{\dag} (\bar{H}_{t}) \frac{B_{t, \underaccent{\bar}{i}_{t + 1}}^{\dag}}{W_{t + 1}} \left( \frac{A_{t}}{\Pi_{t}} - 1 \right) \prod\limits_{t' = t + 1}^{K - 1} \left( \dfrac{A_{t'}}{\Pi_{t'}} \right)^{i_{t'}}: B_{t, \bar{i}_{t - 1}, \underaccent{\bar}{i}_{t + 1}}^{\dag} \in \mathcal{B}_{t, \bar{i}_{t - 1}, \underaccent{\bar}{i}_{t + 1}}^{\dag} \right\}
\end{align*}
where the second line follows from the definition of $\tilde{d}_{t, \bar{i}_{t - 1}} (\bar{R}_{\bar{i}_{t - 1}}^{\ast}, \bar{H}_{t})$ given above.

Finally, define $\Omega_{obs}^{c, r^{\ast}}$ as 
\begin{equation*}
\Omega_{obs}^{c, r^{\ast}} = \left\{ \sum_{t = 0}^{K} \sum_{\bar{i}_{t - 1}, \underaccent{\bar}{i}_{t + 1}} \Omega_{obs, t, \bar{i}_{t - 1}, \underaccent{\bar}{i}_{t + 1}}^{c, r^{\ast}}: D_{t, \bar{i}_{t - 1}, \underaccent{\bar}{i}_{t + 1}}^{\text{obs} \dag} \in \mathcal{D}_{t, \bar{i}_{t - 1}, \underaccent{\bar}{i}_{t + 1}}^{\text{obs} \dag} \right\}
\end{equation*}
where $\Omega_{obs, t, \bar{i}_{t - 1}, \underaccent{\bar}{i}_{t + 1}}^{c, r^{\ast}} \equiv D_{t, \bar{i}_{t - 1}, \underaccent{\bar}{i}_{t + 1}}^{\text{obs} \dag} - \Pi [D_{t, \bar{i}_{t - 1}, \underaccent{\bar}{i}_{t + 1}}^{\text{obs} \dag} | \Lambda_{\mathsf{ancillary}}]$. Then we have
\begin{corollary}
\label{lem:obs_proj} 
\begin{equation*}
\Pi [\BU (\bm{q}, \Psi^{\ast}) | \Omega_{obs}^{c, r^{\ast}}] = \sum_{t = 0}^{K} \sum_{\bar{i}_{t - 1}, \underaccent{\bar}{i}_{t + 1}} \tilde{d}_{t, \bar{i}_{t - 1}}^{\ast \ \dag} (\bar{H}_{t}) T_{b^{\dag}, t, \underaccent{\bar}{i}_{t + 1}}^{obs}
\end{equation*}
where $\tilde{d}_{t, \bar{i}_{t - 1}}^{\ast \ \dag} (\bar{H}_{t})$ is defined as in equations \eqref{main_0} and \eqref{main_d1} but with $T_{t}$ replaced by $T_{b^{\dag}, t, \underaccent{\bar}{i}_{t + 1}}^{obs}$ and
\begin{align*}
T_{b^{\dag}, t, \underaccent{\bar}{i}_{t + 1}}^{obs} \coloneqq & \ \frac{b_{t, \underaccent{\bar}{i}_{t + 1}}^{\dag} (\underaccent{\bar}{C}_{t + 1}, \underaccent{\bar}{R}_{\underaccent{\bar}{i}_{t + 1}}^{\ast}, \underaccent{\bar}{S}_{t}, Y^{d})}{W_{t + 1}} \left( \frac{A_{t}}{\Pi_{t}} - 1 \right) \prod\limits_{t' = t + 1}^{K - 1} \left( \dfrac{A_{t'}}{\Pi_{t'}} \right)^{i_{t'}} \\
& - \Pi \left[ \left. \frac{b_{t, \underaccent{\bar}{i}_{t + 1}}^{\dag} (\underaccent{\bar}{C}_{t + 1}, \underaccent{\bar}{R}_{\underaccent{\bar}{i}_{t + 1}}^{\ast}, \underaccent{\bar}{S}_{t}, Y^{d})}{W_{t + 1}} \left( \frac{A_{t}}{\Pi_{t}} - 1 \right) \prod\limits_{t' = t + 1}^{K - 1} \left( \dfrac{A_{t'}}{\Pi_{t'}} \right)^{i_{t'}} \right\vert \Lambda_{\mathsf{ancillary}} \right].
\end{align*}
\end{corollary}

\begin{proof}
The proof immediately follows from the proof of Theorem \ref{thm:sub} and Corollary \ref{cor:sub}.
\end{proof}

Finally, note that $T_{b^{\dag}, t, \underaccent{\bar}{i}_{t + 1}}^{obs}$ is a statistic because $b_{t, \underaccent{\bar}{i}_{t + 1}}^{\dag} (\underaccent{\bar}{C}_{t + 1}, \underaccent{\bar}{R}_{\underaccent{\bar}{i}_{t + 1}}^{\ast}, \underaccent{\bar}{S}_{t}, Y^{d}) \equiv b_{t, \underaccent{\bar}{i}_{t + 1}}^{\dag} (\underaccent{\bar}{C}_{t + 1}, \underaccent{\bar}{R}_{\underaccent{\bar}{i}_{t + 1}}, \underaccent{\bar}{S}_{t}, Y^{d})$ when multiplied by $\prod\limits_{t' = t + 1}^{K - 1} \left( \dfrac{A_{t'}}{\Pi_{t'}} \right)^{i_{t'}}$.

\subsection{Proof of Lemma \ref{lem:nde-ipw-1} and \ref{lem:nde-ipw-3} in the main text}
\label{app:lem_err}

In this section, we only prove Lemma \ref{lem:nde-ipw-1} and \ref{lem:nde-ipw-3} for the oracle estimators assuming all the nuisance functions to be known. The results for the feasible estimators then follow immediately. 

\begin{proof}[Proof of Lemma \ref{lem:nde-ipw-1}]
We prove Lemma \ref{lem:nde-ipw-1} by introducing the notation $\hat{\theta}_{nde\mbox{-}ipw, g}^{z} \coloneqq \mathbb{P}_{n} \left[ V_{nde\mbox{-}ipw, g}^{z} \right]$
where 
\begin{equation*}
V_{nde\mbox{-}ipw, g}^{z} \coloneqq \prod\limits_{t = 0}^{K} \left\{ \frac{1 - A_{t}}{1 - \Pi_{t}} (1 - Z_{t}) + \frac{A_{t}}{\Pi_{t}} Z_{t} \right\}^{1 - A_{t, g}} \prod\limits_{t = 0}^{K - 1} \left\{ \frac{A_{t}}{\Pi_{t}} \right\}^{A_{t, g}} \frac{\prod\limits_{t = 0}^{K} \mathbbm{1} \{S_{t} = S_{t, g}\} Y}{W_{0}}.
\end{equation*}
So $\hat{\theta}_{nde\mbox{-}ipw, g}^{z} = \hat{\theta}_{nde\mbox{-}ipw, g}$ when we choose $Z_{t} = \Pi_{t}$ and $\hat{\theta}_{nde\mbox{-}ipw, g}^{z} = \hat{\theta}_{ipw, g}$ when we choose $Z_{t} = 0$. Thus we only need to consider the more general $\hat{\theta}_{nde\mbox{-}ipw, g}^{z}$.

Denote $\bar{i}_{t - 1} \equiv \{ i_{j} \in \{0, 1\}, j = 0, 1, \cdots, t - 1\}$ for any $t = 0, \cdots, K - 1$ and $\bar{R}_{t, \bar{i}_{t - 1}} \equiv \{ R_{t'} i_{t' - 1} + ? i_{t' - 1}: 1 \leq t' \leq t \}$ for any $t = 1, \cdots, K$. Then we have
\begin{align*}
V_{nde\mbox{-}ipw, g}^{z} = & \ \prod\limits_{t = 0}^{K} \left\{ \frac{1 - A_{t}}{1 - \Pi_{t}} (1 - Z_{t}) + \frac{A_{t}}{\Pi_{t}} Z_{t} \right\}^{1 - A_{t, g}} \prod\limits_{t = 0}^{K - 1} \left\{ \frac{A_{t}}{\Pi_{t}} \right\}^{A_{t, g}} \frac{\prod\limits_{t = 0}^{K} \mathbbm{1} \{S_{t} = S_{t, g}\} Y}{W_{0}} \\
= & \ \left[ \prod\limits_{t = 0}^{K - 1} \left\{ \left( \frac{1 - A_{t}}{1 - \Pi_{t}} (1 - Z_{t}) + \frac{A_{t}}{\Pi_{t}} Z_{t} \right)^{1 - A_{t, g}} \left( \frac{A_{t}}{\Pi_{t}} \right)^{A_{t, g}} \right\} \frac{\mathbbm{1} \{S_{t} = S_{t, g}\}}{P_{t}} \right] \left[ \frac{\mathbbm{1} \{S_{K} = S_{K, g}\} Y}{P_{K}} \right] \\
= & \ \sum_{\bar{i}_{K} \in \{0, 1\}^{K}} \left[ \begin{array}{c}
\prod\limits_{t = 0}^{K - 1} \left\{ \left( \frac{1 - A_{t}}{1 - \Pi_{t}} (1 - Z_{t}) + \frac{A_{t}}{\Pi_{t}} Z_{t} \right)^{1 - i_{t}} \left( \frac{A_{t}}{\Pi_{t}} \right)^{i_{t}} \mathbbm{1} \{ g_{a, t} (\bar{R}_{t, g}, \bar{C}_{t, g}, \bar{S}_{t - 1, g}) = i_{t} \} \right\} \\
\times \frac{\mathbbm{1} \{S_{t} = g_{s, t} (\bar{R}_{t, g}, \bar{C}_{t, g}, \bar{S}_{t - 1, g})\}}{P_{t}}
\end{array} \right] \\
& \times \left[ \frac{\mathbbm{1} \{S_{K} = g_{s, K} (\bar{R}_{K, g}, \bar{C}_{K, g}, \bar{S}_{K - 1, g})\} Y}{P_{K}} \right] \\
= & \ \sum_{\bar{i}_{K} \in \{0, 1\}^{K}} \left[ \begin{array}{c}
\left\{ \left( \frac{1 - A_{0}}{1 - \Pi_{0}} (1 - Z_{0}) + \frac{A_{0}}{\Pi_{0}} Z_{0} \right)^{1 - i_{0}} \left( \frac{A_{0}}{\Pi_{0}} \right)^{i_{0}} \mathbbm{1} \{ g_{a, 0} (\bar{R}_{0}, \bar{C}_{0}) = i_{0} \} \frac{\mathbbm{1} \{S_{0} = g_{s, 0} (\bar{R}_{0}, \bar{C}_{0})\}}{P_{0}} \right\} \\
\prod\limits_{t = 1}^{K - 1} \left\{ \left( \frac{1 - A_{t}}{1 - \Pi_{t}} (1 - Z_{t}) + \frac{A_{t}}{\Pi_{t}} Z_{t} \right)^{1 - i_{t}} \left( \frac{A_{t}}{\Pi_{t}} \right)^{i_{t}} \mathbbm{1} \{ g_{a, t} (\bar{R}_{t, g}, \bar{C}_{t, g}, \bar{S}_{t - 1, g}) = i_{t} \} \right\} \\
\times \frac{\mathbbm{1} \{S_{t} = g_{s, t} (\bar{R}_{t, g}, \bar{C}_{t, g}, \bar{S}_{t - 1, g})\}}{P_{t}}
\end{array} \right] \\
& \times \left[ \frac{\mathbbm{1} \{S_{K} = g_{s, K} (\bar{R}_{K, g}, \bar{C}_{K, g}, \bar{S}_{K - 1, g})\} Y}{P_{K}} \right] \\
= & \ \sum_{\bar{i}_{K} \in \{0, 1\}^{K}} \left[ \begin{array}{c}
\left\{ \left( \frac{1 - A_{0}}{1 - \Pi_{0}} (1 - Z_{0}) + \frac{A_{0}}{\Pi_{0}} Z_{0} \right)^{1 - i_{0}} \left( \frac{A_{0}}{\Pi_{0}} \right)^{i_{0}} \mathbbm{1} \{ g_{a, 0} (\bar{R}_{0}, \bar{C}_{0}) = i_{0} \} \frac{\mathbbm{1} \{S_{0} = g_{s, 0} (\bar{R}_{0}, \bar{C}_{0})\}}{P_{0}} \right\} \\
\prod\limits_{t = 1}^{K - 1} \left\{ \left( \frac{1 - A_{t}}{1 - \Pi_{t}} (1 - Z_{t}) + \frac{A_{t}}{\Pi_{t}} Z_{t} \right)^{1 - i_{t}} \left( \frac{A_{t}}{\Pi_{t}} \right)^{i_{t}} \mathbbm{1} \{ g_{a, t} (\bar{R}_{t, \bar{i}_{t - 1}}, \bar{C}_{t}, \bar{S}_{t - 1}) = i_{t} \} \right\} \\
\times \frac{\mathbbm{1} \{S_{t} = g_{s, t} (\bar{R}_{t, \bar{i}_{t - 1}}, \bar{C}_{t}, \bar{S}_{t - 1})\}}{P_{t}}
\end{array} \right] \\
& \times \left[ \frac{\mathbbm{1} \{S_{K} = g_{s, K} (\bar{R}_{K, \bar{i}_{K - 1}}, \bar{C}_{K}, \bar{S}_{K - 1})\} Y}{P_{K}} \right]
\end{align*}
is a function of the observed data and thus a statistic. Note that in the last line of the above display, we use the NDE$(\bar{R}_{K}^{\ast}, \bar{C}_{K}, Y)$ assumption and in the second but last line, we use the fact that at $t = 0$, $L_{0, g} = L_{0}, C_{0, g} = C_{0}$ and $R_{0, g} = R_{0} \equiv 0$ (by definition as there is no test before $t = 0$).

\end{proof}

\begin{proof}[Proof of Lemma \ref{lem:nde-ipw-3}]
Now we will show $V_{nde\mbox{-}ipw, g} - V_{ipw, g} - \Pi [V_{nde\mbox{-}ipw, g} - V_{ipw, g} | \Lambda_{\mathsf{ancillary}}] \in \Lambda_{\mathsf{NDE}, obs}^{c, r^{\ast} \perp}$. It suffices to show that $V_{nde\mbox{-}ipw, g} - V_{ipw, g}$ is of the form
$
\sum\limits_{t = 0}^{K - 1} \sum\limits_{\bar{i}_{t - 1}, \underaccent{\bar}{i}_{t + 1}} D_{b, t, \bar{i}_{t - 1}, \underaccent{\bar}{i}_{t + 1}}^{\text{obs}}
$
and then apply Theorem \ref{thm:obs_nde}. We now show the former as follows.

Define set $I_{\bar{i}_{K}} \coloneqq \{t: i_{t} = 1\}$. Then 
\begin{align*}
& \ V_{nde\mbox{-}ipw, g} - V_{ipw, g} \\
= & \ \sum_{\bar{i}_{K} \in \{0, 1\}^{K}} \left\{ \prod_{t = 0}^{K} \frac{\mathbbm{1} \{S_{t} = g_{s, t} (\bar{R}_{t, \bar{i}_{t - 1}}, \bar{C}_{t}, \bar{S}_{t - 1})\}}{p (S_{t} | \bar{H}_{t})} \mathbbm{1} \{g_{a, t} (\bar{R}_{t, \bar{i}_{t - 1}}, \bar{C}_{t}, \bar{S}_{t - 1}) = i_{t} \} \right\} \prod_{t \in I_{\bar{i}_{K}}} \left( \frac{A_{t}}{\Pi_{t}} \right) \\
& \times \left\{ 1 - \prod_{t \not\in I_{\bar{i}_{K}}} \left( \frac{1 - A_{t}}{1 - \Pi_{t}} \right) \right\} Y \\
\coloneqq & \ \sum_{\bar{i}_{K} \in \{0, 1\}^{K}} \left\{ \frac{\prod_{t = 0}^{K} \mathbbm{1}\{S_{t} = g_{s, t} (\bar{R}_{t, \bar{i}_{t - 1}}, \bar{C}_{t}, \bar{S}_{t - 1})\} \mathbbm{1} \{g_{a, t} (\bar{R}_{t, \bar{i}_{t - 1}}, \bar{C}_{t}, \bar{S}_{t - 1}) = i_{t}\}}{W_{0}} \right\} e_{\bar{i}_{K}} (\bar{O}_{K}).
\end{align*}
For each $\bar{i}_{K}$, 
\begin{align*}
e_{\bar{i}_{K}} (\bar{O}_{K}) = & \ \prod_{t \in I_{\bar{i}_{K}}} \left( \frac{A_{t}}{\Pi_{t}} \right) \left\{1 - \prod_{t \not\in I_{\bar{i}_{K}}} \left( \frac{1 - A_{t}}{1 - \Pi_{t}} \right) \right\} Y \\
= & \ \left\{ \prod_{t = 0}^{K - 1} \left( \frac{A_{t}}{\Pi_{t}} \right)^{\mathbbm{1} \{t \in I_{\bar{i}_{K}}\}} - \prod_{t = 0}^{K - 1} \left( \frac{A_{t}}{\Pi_{t}} \right)^{\mathbbm{1} \{t \in I_{\bar{i}_{K}}\}} \left( \frac{1 - A_{t}}{1 - \Pi_{t}} \right)^{\mathbbm{1} \{t \not\in I_{\bar{i}_{K}}\}} \right\} Y \\
= & \ \sum_{t = 0}^{K - 1} \prod_{t_{1} < t, t_{1} \in I_{\bar{i}_{K}}} \frac{A_{t_{1}}}{\Pi_{t_{1}}} \left[ \mathbbm{1} \{i_{t} = 0\} \frac{A_{t} - \Pi_{t}}{1 - \Pi_{t}} \right] \prod_{t_{2} > t, t_{2} \in I_{\bar{i}_{K}}} \frac{A_{t_{2}}}{\Pi_{t_{2}}} \prod_{t_{2} > t, t_{2} \not\in I_{\bar{i}_{K}}}\frac{1 - A_{t_{2}}}{1 - \Pi_{t_{2}}} Y.
\end{align*}
Now it is easy to see that $V_{nde\mbox{-}ipw, g} - V_{ipw, g}$ is of the form $\sum\limits_{t = 0}^{K} \sum\limits_{\bar{i}_{t - 1}, \underaccent{\bar}{i}_{t + 1}} D_{b, t, \bar{i}_{t - 1}, \underaccent{\bar}{i}_{t + 1}}^{\text{obs}}$ by observing that $\dfrac{1 - A_{t}}{1 - \Pi_{t}} = 1 - \dfrac{A_{t} - \Pi_{t}}{1 - \Pi_{t}}$. Now by changing the actual form of $B_{t, \bar{i}_{t - 1}, \underaccent{\bar}{i}_{t + 1}}$ from line to line, we have
\begin{align*}
& \ V_{nde\mbox{-}ipw, g} - V_{ipw, g} \\
= & \ \sum_{t = 0}^{K - 1} \sum_{\bar{i}_{t - 1}, \underaccent{\bar}{i}_{t + 1}} \frac{B_{t, \bar{i}_{t - 1}, \underaccent{\bar}{i}_{t + 1}}}{W_{t + 1}} \prod_{t' = 0}^{t - 1} \left( \frac{A_{t'}}{\Pi_{t'}} \right)^{i_{t'}} \left[ \mathbbm{1} \{i_{t} = 0\} \left( \frac{A_{t}}{\Pi_{t}} - 1 \right) \frac{\Pi_{t}}{1 - \Pi_{t}} \right] \\
& \times \prod_{t' = t + 1}^{K - 1} \left( \frac{A_{t'}}{\Pi_{t'}} \right)^{i_{t'}} \left( 1 - \frac{A_{t'} - \Pi_{t'}}{1 - \Pi_{t'}} \right)^{1 - i_{t'}} Y \\
= & \ \sum_{t = 0}^{K - 1} \sum_{\bar{i}_{t - 1}, \underaccent{\bar}{i}_{t + 1}} \frac{B_{t, \bar{i}_{t - 1}, \underaccent{\bar}{i}_{t + 1}}}{W_{t + 1}} \prod_{t' = 0}^{t - 1} \left( \frac{A_{t'}}{\Pi_{t'}} \right)^{i_{t'}} \left( \frac{A_{t}}{\Pi_{t}} - 1 \right) \prod_{t' = t + 1}^{K - 1} \left( \frac{A_{t'}}{\Pi_{t'}} \right)^{i_{t'}} \left( 1 + \frac{A_{t'} - \Pi_{t'}}{\Pi_{t'} - 1} \right)^{1 - i_{t'}} Y \\
= & \ \sum_{t = 0}^{K - 1} \sum_{\bar{i}_{t - 1}, \underaccent{\bar}{i}_{t + 1}} \frac{B_{t, \bar{i}_{t - 1}, \underaccent{\bar}{i}_{t + 1}}}{W_{t + 1}} \prod_{t' = 0}^{t - 1} \left( \frac{A_{t'}}{\Pi_{t'}} \right)^{i_{t'}} \left( \frac{A_{t}}{\Pi_{t}} - 1 \right) \prod_{t' = t + 1}^{K - 1} \left( \frac{A_{t'}}{\Pi_{t'}} \right)^{i_{t'}} \\
& \times \left\{ 1 + \sum_{\mathfrak{s} \subset \{t + 1 \leq k \leq K - 1: i_{k} = 0\}, \mathfrak{s} \neq \emptyset} \prod_{m \in \mathfrak{s}} \left( \frac{A_{m} - \Pi_{m}}{\Pi_{m} - 1} \right) \right\} Y \\
= & \ \sum_{t = 0}^{K - 1} \sum_{\bar{i}_{t - 1}, \underaccent{\bar}{i}_{t + 1}} \frac{B_{t, \bar{i}_{t - 1}, \underaccent{\bar}{i}_{t + 1}}}{W_{t + 1}} \prod_{t' = 0}^{t - 1} \left( \frac{A_{t'}}{\Pi_{t'}} \right)^{i_{t'}} \left( \frac{A_{t}}{\Pi_{t}} - 1 \right) \prod_{t' = t + 1}^{K - 1} \left( \frac{A_{t'}}{\Pi_{t'}} \right)^{i_{t'}} Y,
\end{align*}
where the last line follows from exactly the same calculations as in the proof of Theorem \ref{thm:obs_nde} and this concludes the proof.
\end{proof}

\begin{remark}\label{rem:app_cost}
Suppose the model defined by the NDE$(\bar{C}_{K}, \bar{R}_{K}^{\ast}, Y^{d})$ and the modified identifiability assumptions hold. Consider a testing cost function $c^{\ast} (t, \bar{c}_{t}, \bar{r}_{t}, \bar{s}_{t})$. In Remarks \ref{rem:cost} and \ref{rem:6}.1 of Section \ref{sec:formal}, we pointed out that, for non-NDE-censored subjects, the results stated for $V_{nde\mbox{-}ipw, g}$ in the main text remain true for $V_{nde\mbox{-}ipw, g}^{c^{\ast}}$ where $V_{nde\mbox{-}ipw, g}^{c^{\ast}}$ is $V_{nde\mbox{-}ipw, g}$ except with
$Y_{g} = Y + \sum\limits_{t = 0}^{K - 1} \mathbbm{1} \{A_{t} = 1, A_{t, g} = 0\} c^{\ast} (t, \bar{C}_{t}, \bar{R}_{t}, \bar{S}_{t}) = Y^{d} - \sum\limits_{t = 0}^{K - 1} A_{t, g} c^{\ast} (t, \bar{C}_{t}, \bar{R}_{t}, \bar{S}_{t})$ substituted for $Y$. Note $V_{nde\mbox{-}ipw, g}^{c^{\ast}}$ like $V_{nde\mbox{-}ipw, g}$ is zero for all NDE-censored subjects. Consider the statistic
\begin{align*}
V_{nde\mbox{-}ipw, g}^{c^{\ast} \dag} \equiv & \ - \sum\limits_{j = 0}^{K - 1} A_{j, g} c^{\ast} (j,\bar{C}_{j},\bar{R}_{j},\bar{S}_{j}) \prod\limits_{t = 0}^{j - 1} \left( \frac{A_{t}}{\Pi_{t}} \right)^{A_{t, g}} \frac{\prod\limits_{t = 0}^{j} \mathbbm{1} \{S_{t} = S_{t, g}\}}{W_{0}^{j}} \\
& + \prod\limits_{t = 0}^{K - 1} \left( \frac{A_{t}}{\Pi_{t}} \right)^{A_{t, g}} \frac{\prod\limits_{t = 0}^{K} \mathbbm{1} \{S_{t} = S_{t, g}\}}{W_{0}} Y^{d}
\end{align*}
where we recall that $W_{0}^{j} = \prod\limits_{t = 0}^{j} p (S_{t} | \bar{H}_{t})$ and $W_{0} = W_{0}^{K}$. $V_{nde\mbox{-}ipw, g}^{c^{\ast} \dag}$, like $V_{nde\mbox{-}ipw, g}^{c^{\ast}}$, has mean $\theta_{g} = \E [Y_{g}]$ but, unlike $V_{nde\mbox{-}ipw, g}^{c^{\ast}}$, need not be zero for all NDE-censored subjects, raising the question whether using $V_{nde\mbox{-}ipw, g}^{c^{\ast} \dag}$ offers efficiency advantages. Although it is true that under the model, the variance of $V_{nde\mbox{-}ipw, g}^{c^{\ast} \dag}$ is often less and never greater than that of $V_{nde\mbox{-}ipw, g}^{c^{\ast}}$, nonetheless, the residuals from the projection on $\Lambda_{\mathsf{ancillary}}$ are algebraically identical: i.e. $V_{nde\mbox{-}ipw, g}^{c^{\ast}} - \Pi [V_{nde\mbox{-}ipw, g}^{c^{\ast}} | \Lambda_{\mathsf{ancillary}}] = V_{nde\mbox{-}ipw, g}^{c^{\ast} \dag} - \Pi [V_{nde\mbox{-}ipw, g}^{c^{\ast} \dag} | \Lambda_{\mathsf{ancillary}}]$. It follows that any RAL estimator constructed from $V_{nde\mbox{-}ipw, g}^{c^{\ast} \dag}$ will have the same influence function as the analogous estimator constructed from $V_{nde\mbox{-}ipw, g}^{c^{\ast}}$, although $V_{nde\mbox{-}ipw, g}^{c^{\ast} \dag}$ would be expected to have better finite sample performance.
\end{remark}

\subsection{Further results for the simulation study in Section \ref{sec:simulation}}
\label{sim:further}

Table \ref{tab:dgp1-2voi} reports the Monte Carlo means and standard deviations of $\tilde{\theta}_{ipw, g^{opt}}$, $\tilde{\theta}_{nde\mbox{-}ipw, g^{opt}}$, $\tilde{\theta}_{ipw, g^{opt}} (\hat{b}_{sub, ipw, g^{opt}})$ and $\tilde{\theta}_{nde\mbox{-}ipw, g^{opt}} (\hat{b}_{sub, nde\mbox{-}ipw, g^{opt}})$ when $\rho = 0.5$. We observe that the RE of adjusting for the NDE assumption by projection is $(0.494 / 0.169)^2 = 8.544$ (or $(0.441 / 0.115)^2 = 14.706$) at $n = 2.5 \times 10^{4}$ (or at $n = 5 \times 10^{4}$), higher than the RE of adjusting for the NDE assumption by simply not censoring $\{A_{0} = 1\}$, which is $(0.494 / 0.295)^2 = 2.804$ (or $(0.441 / 0.223)^2 = 3.911$) at $n = 2.5 \times 10^{4}$ (or at $n = 5 \times 10^{4}$).

\begin{table}[h]
\centering
\begin{tabular}{c|cc}
\hline
& $\tilde{\theta}_{ipw, g^{opt}}$ & $\tilde{\theta}_{nde\mbox{-}ipw, g^{opt}}$ \\ \hline
$n = 2.5 \times 10^{4}$ & 26.036 (0.494) & 25.982 (0.295) \\ 
$n = 5 \times 10^{4}$ & 25.926 (0.441) & 25.946 (0.223) \\ \hline
& $\tilde{\theta}_{ipw, g^{opt}} (\hat{b}_{sub, ipw, g^{opt}})$ & $\tilde{\theta}_{nde\mbox{-}ipw, g^{opt}} (\hat{b}_{sub, nde\mbox{-}ipw, g^{opt}})$ \\ \hline
$n = 2.5 \times 10^{4}$ & 25.927 (0.169) & 25.939 (0.169) \\ 
$n = 5 \times 10^{4}$ & 25.961 (0.115) & 25.966 (0.115) \\ \hline
Truth & 25.971 & 25.971 \\ \hline
\end{tabular}
\caption{Simulation results for DGP: A comparison between the estimated value functions with and without adjusting for the NDE assumption. The numerical values outside (inside) the parentheses are the Monte Carlo means (standard deviations).}
\label{tab:dgp1-2voi}
\end{table}

In some other settings, the efficiency gain can be much greater than those reported in Table \ref{tab:dgp1-2voi}. For example, we could make the following modification to the treatment probability in DGP while fixing $\rho = 0.5$: 
\begin{align*}
S_{1} \sim & \ \text{Bernoulli} (0.90) \mathbbm{1} \{A_{0} = 1, R_{1} = 0\} + \text{Bernoulli} (0.95) \mathbbm{1} \{A_{0} = 1, R_{1} = 1\} \\
& + \text{Bernoulli} (0.01) \mathbbm{1} \{A_{0} = 0, R_{1} = ?\}.
\end{align*}
We call this DGP'. What distinguishes DGP' from DGP is the high correlation between $A_{0}$ and $S_{1}$. In DGP', the optimal regime is still ``always treat''. In the simulation, we again observe that the optimal regime is selected in 100 out of 100 replications. Reading from Table \ref{tab:dgp1voi}, the RE of adjusting for the NDE assumption by not censoring $\{A_{0} = 1\}$ is $(0.881 / 0.550)^2 = 2.566$ (or $(0.743 / 0.396)^2 = 3.520$) at $n = 2.5 \times 10^{4}$ (or $n = 5 \times 10^{4}$). Ignoring screening only improves efficiency by a small margin, owing to the strong correlation between $S_{1}$ and $A_{0}$ in DGP' (informally, $S_{1} = A_{0}$ with large probability in this simulation). The RE of adjusting for the NDE assumption by projection is much greater: $(0.881 / 0.167)^2 = 27.830$ (or $(0.743 / 0.114)^2 = 42.478$) at $n = 2.5 \times 10^{4}$ (or $n = 5 \times 10^{4}$).

\begin{table}[h]
\centering
\begin{tabular}{c|cc}
\hline
& $\tilde{\theta}_{ipw, g^{opt}}$ & $\tilde{\theta}_{nde\mbox{-}ipw, g^{opt}}$ \\ \hline
$n = 2.5 \times 10^{4}$ & 26.161 (0.881) & 26.000 (0.550) \\ 
$n = 5 \times 10^{4}$ & 25.921 (0.743) & 25.915 (0.396) \\ \hline
& $\tilde{\theta}_{ipw, g^{opt}} (\hat{b}_{sub, ipw, g^{opt}})$ & $\tilde{\theta}_{nde\mbox{-}ipw, g^{opt}} (\hat{b}_{sub, nde\mbox{-}ipw, g^{opt}})$ \\ \hline
$n = 2.5 \times 10^{4}$ & 25.923 (0.167) & 25.937 (0.167) \\ 
$n = 5 \times 10^{4}$ & 25.963 (0.114) & 25.968 (0.114) \\ \hline
Truth & 25.971 & 25.971 \\ \hline
\end{tabular}
\caption{Simulation results for DGP': A comparison between the estimated value functions with and without adjusting for the NDE assumption. The numerical values outside (inside) the parentheses are the Monte Carlo means (standard deviations).}
\label{tab:dgp1voi}
\end{table}

\subsection{A simulation study on cost benefit analysis based on VoI estimated by opt-SNMM}
\label{sim:cost} 

In this section, we slightly modify the DGP in the main text by introducing a baseline covariate $L_{0}$ which is a noisy measurement of the disease status, such as some symptom of the disease. We call this DGP as DGP2. In particular, DGP2 includes the following time-ordered random variables: $(Z_{0}, L_{0}, A_{0}, R_{1}, S_{1}, Y^{d})$ with $Z_{0}$ being latent. DGP2 is generated as follows:

\begin{itemize}
\item $Z_{0} \sim \text{Bernoulli} (0.6)$ describes the unobserved underlying disease status of a subject at the start of the study ($t = 0$).

\item $L_{0}$ is a noisy measurement of the latent disease status $Z_{0}$, where $L_{0} \sim \mathbbm{1} \{ Z_{0} = 1 \} \text{Bernoulli} (0.55) + \mathbbm{1} \{ Z_{0} = 0 \} \text{Bernoulli} (0.45)$.

\item $A_{0} \sim \mathbbm{1} \{ L_{0} = 1 \} \text{Bernoulli}(0.1) + \mathbbm{1} \{ L_{0} = 0 \} \text{Bernoulli} (0.9)$.

\item When $A_{0} = 0$, $R_{1} = ?$; when $A_{0} = 1$, the test result $R_{1}$, unlike $L_{0}$, is a very accurate measurement of the disease status $Z_{0}$: $R_{1} \sim \mathbbm{1} \{ Z_{0} = 1 \} \text{Bernoulli} (0.95) + \mathbbm{1} \{ Z_{0} = 0 \} \text{Bernoulli} (0.05)$.

\item For $S_{1}$ 
\begin{align*}
& \ S_{1} \sim \text{Bernoulli} (0.45) \mathbbm{1} \{L_{0} = 0, A_{0} = 1, R_{1} = 0\} + \text{Bernoulli} (0.55) \mathbbm{1} \{L_{0} = 1, A_{0} = 1, R_{1} = 0\} \\
& + \text{Bernoulli} (0.45) \mathbbm{1} \{L_{0} = 0, A_{0} = 0, R_{1} = ?\} + \text{Bernoulli} (0.55) \mathbbm{1} \{L_{0} = 1, A_{0} = 0, R_{1} = ?\} \\
& + \text{Bernoulli} (0.85) \mathbbm{1} \{L_{0} = 0, A_{0} = 1, R_{1} = 1\} + \text{Bernoulli} (0.9) \mathbbm{1} \{L_{0} = 1, A_{0} = 1, R_{1} = 1\}.
\end{align*}
Given $L_{0}$, the probability of taking the treatment at $t = 1$ is increased when $A_{0} = 1$ and $R_{1} = 1$. The probability of taking the treatment is also larger when $L_{0} = 1$ compared to when $L_{0} = 0$, given $(A_{0}, R_{1})$.

\item $Y^{d} \sim N(- 5 Z_{0} + 9 S_{1} Z_{0} - 3 S_{1} (1 - Z_{0}), 1)$. The utility is defined as $Y \coloneqq Y^{d} - c^{\ast} A_{0}$ with the cost of screening $c^{\ast} = 0.9$.
\end{itemize}

In DGP2, we are interested in the counterfactual utility under the optimal regimes, the counterfactual utility under the optimal regime withholding screening, and the value of information (VoI), i.e. their difference, both unconditional and conditional on $L_{0} = 0, 1$. The conditional value functions and VoIs are often of interest in personalized medicine. In this section we only consider the opt-SNMM estimators with and without adjusting for the NDE assumption by projection. The methods to compute the (adjusted) opt-SNMM estimators can be found in Appendix \ref{app:sim}. They are special cases of the methods described in Sections \ref{sec:estimator} and \ref{sec:optimal_dr}.

A large $\mathsf{VoI}$ indicates that the screening test, at the cost $c$, is valuable for treatment decision making because it improves the utility under the optimal regimes with the extra information provided by the test result. In practice, one may choose some predetermined threshold $\nu$ and test $\mathsf{H}_{0}: \mathsf{VoI} \leq \nu$. If rejected, there is strong evidence that the screening should be allowed in decision making at the current cost $c^{\ast}$. The choice of $\nu$ may be determined not only by scientific and statistical concerns, but also by issues related to sociological, economical, and even ethical concerns. It is beyond the scope of this paper to propose any appropriate procedure for choosing such a $\nu$ in practice\footnote{In practice, one can also consider rescaling VoI by the value function. Then the rescaled VoI will be bounded between 0 and 1, easier for choosing $\nu$.}. Later we will demonstrate that the efficiency gain by adjusting for the NDE assumption can make a difference in a cost-benefit analysis based on VoI.

First, using a dataset with $10^{7}$ observations as the truth, the optimal regime for DGP2 is: 
\begin{align}
a_{0}^{opt} & = 1 - l_{0},  \label{opt0_1} \\
s_{1}^{opt} (a_{0}^{opt}) & = a_{0}^{opt} r_{1} + (1 - a_{0}^{opt}).
\label{opt1_1}
\end{align}
In words, the optimal screening strategy at $t = 0$ is to screen if and only if $L_{0} = 0$; the optimal treatment strategy at $t = 1$ is: if not tested at $t = 0$, always treat; if tested at $t = 0$, treat if and only if the test result is positive ($R_{1} = 1$). When withholding screening, the optimal treatment regime at $t = 1$ (always treat) for DGP2: 
\begin{align} \label{opt1_2}
s_{1}^{opt} (a_{0} = 0) = 1.
\end{align}
Further, the value function (based on this dataset with size $10^{7}$) is 1.265 and the utility under the optimal regime withholding screening is 1.204. Hence the VoI is 0.061 and the screening test is cost-beneficial in DGP2 (see the last row of Table \ref{tab:dgp2voi}). Conditioning on $L_{0} = 1$, the value function based on this huge dataset is 1.529 and the utility under the optimal regime withholding screening is the same. Hence the conditional VoI given $L_{0} = 1$ is 0 and the screening test is not cost-beneficial given $L_{0} = 1$ (see the last row of Table \ref{tab:dgp2voi}). Conditioning on $L_{0} = 0$, the value function based on this huge dataset is 0.990 and the utility under the optimal regime withholding screening is 0.871. Hence the VoI is 0.119 and the screening test is cost-beneficial given $L_{0} = 0$ (see the last row of Table \ref{tab:dgp2voi}).

Then we compare the percentages of the 100 replications selecting the actual optimal regimes with and without adjusting for the NDE assumption in Table \ref{tab:dgp2perc}:

\begin{itemize}
\item When withholding screening (the third column of Table \ref{tab:dgp2perc}), 100 out of 100 replications select the optimal treatment strategy given the status $L_{0}$ with and without adjusting for the NDE assumption.

\item At $t = 1$ (the first column of Table \ref{tab:dgp2perc}), 100 out of 100 replications select the optimal treatment strategy given the status $L_{0}$, the screening strategy at $t = 0$, and the test result $R_{1}$, at both sample sizes $n = 2.5 \times 10^{4}, 5 \times 10^{4}$, with and without adjusting for the NDE assumption.

\item At $t = 0$ (the second column of Table \ref{tab:dgp2perc}), the percentage of selecting the actual optimal screening strategy is improved from 82\% at $n = 2.5 \times 10^{4}$ (100\% at $n = 5 \times 10^{4}$) to 85\% at $n = 2.5 \times 10^{4}$ (100\% at $n = 5 \times 10^{4}$) after adjusting for the NDE assumption. This result demonstrates that adjusting for the NDE assumption can help improve the probability of selecting the optimal regimes.
\end{itemize}

\begin{table}[h]
\centering
\begin{tabular}{c|c|c|c}
\hline
& $t = 0$: optimal & $t = 1$: optimal & $t = 1$: optimal no screening \\ 
\hline
w/o NDE ($n = 2.5 \times 10^{4}$) & 82\% & 100\% & 100\% \\ 
w/ NDE ($n = 2.5 \times 10^{4}$) & 100\% & 100\% & 100\% \\ 
w/o NDE ($n = 5 \times 10^{4}$) & 85\% & 100\% & 100\% \\ 
w/ NDE ($n = 5 \times 10^{4}$) & 100\% & 100\% & 100\% \\ \hline
\end{tabular}
\caption{Simulation results for DGP2: A comparison between the percentages of all the 100 replications selecting the actual optimal regimes (or the optimal regimes withholding screening) at $t = 0, 1$ with and without adjusting for the NDE assumption.}
\label{tab:dgp2perc}
\end{table}

Next, Table \ref{tab:dgp2voi} shows the Monte Carlo means and Monte Carlo standard deviations of the estimated value function, the estimated utility under the optimal regime withholding screening, and the estimated VoI with and without adjusting for the NDE assumption:

\begin{itemize}
\item We first look at the second column of Table \ref{tab:dgp2voi}, where we show the Monte Carlo mean and standard deviation of the estimated utility under the optimal regime withholding screening (i.e. always treat, see equation \eqref{opt1_2}), marginally (upper panel) and conditional on $L_{0} = 1$ (middle panel) and $L_{0} = 0$ (lower panel). After adjusting for the NDE assumption, the Monte Carlo mean of the estimated marginal and conditional utilities get closer to the true values. At $n = 2.5 \times 10^{4}$ (or $n = 5 \times 10^{4}$), the RE of adjusting for the NDE assumption for marginal utility is $(0.0764 / 0.0263)^{2} \approx 8$ (or $(0.0591 / 0.0182)^{2} \approx 10$). Conditional on $L_{0} = 0$, the RE of adjusting for the NDE assumption for the conditional utility is $(0.129 / 0.0215)^{2} \approx 36$ (or $(0.106 / 0.0173)^{2} \approx 38$) at $n = 2.5 \times 10^{4}$ (or $n = 5 \times 10^{4}$). Conditional on $L_{0} = 1$, however, we do not observe large efficiency improvement by adjusting for the NDE assumption. The RE of adjusting for the NDE assumption for the conditional utility is $(0.0418 / 0.0378)^{2} \approx 1.22$ (or $(0.0308 / 0.0274)^{2} \approx 1.26$) at $n = 2.5 \times 10^{4}$ (or $n = 5 \times 10^{4}$). The different magnitude in efficiency gains while conditioning on $L_{0} = 1$ or $L_{0} = 0$ can be explained by the difference in the percentages of subjects following the optimal regime (always treat) withholding for testing in the simulation, which equal: 
\begin{align*}
& \ \Pr (L_{0} = 1, A_{0} = 0, R_{1} = ?, S_{1} = 1) \\
= & \ \sum_{z = 0, 1} \Pr (Z_{0} = z, L_{0} = 1, A_{0} = 0, R_{1} = ?, S_{1} = 1) \\
= & \ 0.6 \times 0.55 \times 0.9 \times 0.55 + 0.4 \times 0.45 \times 0.9 \times 0.55 \approx 0.25
\end{align*}
and 
\begin{align*}
& \ \Pr (L_{0} = 0, A_{0} = 0, R_{1} = ?, S_{1} = 1) \\
= & \ \sum_{z = 0, 1} \Pr (Z_{0} = z, L_{0} = 0, A_{0} = 0, R_{1} = ?, S_{1} = 1) \\
= & \ 0.6 \times 0.45 \times 0.1 \times 0.45 + 0.4 \times 0.55 \times 0.1 \times 0.45 \approx 0.02.
\end{align*}
Given that the percentages of $L_{0} = 1$ (51\%) and that of $L_{0} = 0$ (49\%) are very close in the simulation, it is not surprising to see a greater efficiency gain given $L_{0} = 0$ because only 2\% of subjects are following the regime of ``no testing and always treat'' without incorporating the NDE assumption.

Recall that in Table \ref{tab:dgp2perc}, for both sample sizes ($n = 2.5 \times 10^{4}, 5 \times 10^{4}$), all the 100 replications select the actual optimal regime withholding screening, whether or not adjusting for the NDE assumption. Thus the above efficiency gain either marginally or conditional on $L_{0} = 0$ is not mainly due to the higher chance of selecting the actual optimal regime.

\item The first column of Table \ref{tab:dgp2voi} shows the Monte Carlo mean and standard deviation of the estimated value functions, marginally (upper panel) and conditional on $L_{0} = 1$ (middle panel) and $L_{0} = 0$ (lower panel). Again, we observe that after adjusting for the NDE assumption, the Monte Carlo mean of the estimated marginal and conditional value function gets closer to the truth. At $n = 2.5 \times 10^{4}$, the RE of adjusting for the NDE assumption is $(0.0315 / 0.0243)^{2} \approx 1.68$. This is partially due to the higher chance of selecting the optimal regime after adjusting for the NDE assumption (see the first column of Table \ref{tab:dgp2perc}). At $n = 5 \times 10^{4}$, the RE of adjusting for the NDE assumption decreases to $(0.0198 / 0.0167)^{2} \approx 1.47$, because the percentage of selecting the optimal regime is 85\% without adjusting for the NDE assumption, closer to 100\% than the percentage (63\%) at $n = 2.5 \times 10^{4}$. The results conditional on $L_{0}$ are quite similar and hence we omit the details.

\item The third column of Table \ref{tab:dgp2voi} shows the Monte Carlo mean and standard deviation of the estimated VoIs, marginally (upper panel) and conditional on $L_{0} = 1$ (middle panel) and $L_{0} = 0$ (lower panel). We again observe that the estimated marginal and conditional VoIs are closer to the truth after adjusting for the NDE assumption. At $n = 2.5 \times 10^{4}$ (or $n = 5 \times 10^{4}$), the RE of adjusting for the NDE assumption is $(0.0666 / 0.0106)^{2} \approx 39$ (or $(0.0517 / 0.00848)^{2} \approx 37$). Such efficiency gain can be valuable in a cost-benefit analysis for the screening test, which we will discuss next. Similarly, conditional on $L_{0} = 0$, the RE of adjusting for the NDE assumption is $(0.129 / 0.0215)^{2} \approx 36$ (or $(0.106 / 0.0173)^{2} \approx 38$) at sample size $n = 2.5 \times 10^{4}$ (or $n = 5 \times 10^{4}$). Given $L_{0} = 1$, the VoI is very close to 0 without adjusting for the NDE assumption, but becomes 0 (the true value) after adjusting for the NDE assumption. In the strata $L_{0} = 1$, it is optimal to treat regardless of the screening status, or equivalently screening is not beneficial. Hence the VoI conditional on $L_{0} = 1$ should be 0. As we mentioned above, since only 2\% of the sample with $L_{0} = 0$ follow the optimal regime withholding screening, it is not surprising to see a very large efficiency gain in the VoI conditional on $L_{0} = 0$.
\end{itemize}

\begin{table}[h]
\centering
\begin{tabular}{c|c|c|c}
\hline
& $\E [Y_{g^{opt}}]$ & $\E [Y_{a_{0} = 0, s_{1}^{opt} (a_{0}
= 0)}]$ & VoI \\ \hline
w/o NDE ($n = 2.5 \times 10^{4}$) & 1.270 (0.0315) & 1.194 (0.0764) & 0.0758 (0.0666) \\ 
w/ NDE ($n = 2.5 \times 10^{4}$) & 1.257 (0.0243) & 1.194 (0.0263) & 0.0631 (0.0106) \\ 
w/o NDE ($n = 5 \times 10^{4}$) & 1.269 (0.0198) & 1.198 (0.0591) & 0.0701 (0.0517) \\ 
w/ NDE ($n = 5 \times 10^{4}$) & 1.264 (0.0163) & 1.200 (0.0182) & 0.0636 (0.00848) \\ \hline
Truth & 1.265 & 1.204 & 0.0610 \\ \hline
w/o NDE ($n = 2.5 \times 10^{4}$) $\vert L_{0} = 1$ & 1.525 (0.0407) & 1.524 (0.0418) & $1.179 \times 10^{-3}$ (0.0117) \\ 
w/ NDE ($n = 2.5 \times 10^{4}$) $\vert L_{0} = 1$ & 1.517 (0.0378) & 1.517 (0.0378) & 0 (0) \\ 
w/o NDE ($n = 5 \times 10^{4}$) $\vert L_{0} = 1$ & 1.531 (0.0308) & 1.531 (0.0308) & $1.057 \times 10^{-4}$ ($1.057 \times 10^{-3}$) \\ 
w/ NDE ($n = 5 \times 10^{4}$) $\vert L_{0} = 1$ & 1.528 (0.0274) & 1.528 (0.0274) & 0 (0) \\ \hline
Truth $\vert L_{0} = 1$ & 1.529 & 1.529 & 0 \\ \hline
w/o NDE ($n = 2.5 \times 10^{4}$) $\vert L_{0} = 0$ & 1.004 (0.0466) & 0.851 (0.155) & 0.154 (0.129) \\ 
w/ NDE ($n = 2.5 \times 10^{4}$) $\vert L_{0} = 0$ & 0.987 (0.0247) & 0.858 (0.0329) & 0.135 (0.0215) \\ 
w/o NDE ($n = 5 \times 10^{4}$) $\vert L_{0} = 0$ & 0.995 (0.0282) & 0.852 (0.116) & 0.143 (0.106) \\ 
w/ NDE ($n = 5 \times 10^{4}$) $\vert L_{0} = 0$ & 0.988 (0.0179) & 0.858 (0.0251) & 0.130 (0.0173) \\ \hline
Truth $\vert L_{0} = 0$ & 0.990 & 0.871 & 0.119 \\ \hline
\end{tabular}
\caption{Simulation results for DGP2: A comparison between the estimated value functions, the estimated utilities under the optimal regime withholding screening, and the estimated VoIs with and without adjusting for the NDE assumption. Upper panel: marginal versions; middle panel: conditional version given $L_{0} = 1$; lower panel: conditional version given $L_{0} = 0$. The numerical values outside (inside) the parentheses are the Monte Carlo means (standard deviations).}
\label{tab:dgp2voi}
\end{table}

Finally, we discuss how a cost-benefit analysis can be done based on the estimated VoI and its estimated standard deviation. We denote the VoI conditional on $L_{0} = l_{0}$ as $\text{cVoI}(l_{0})$. To determine if the screening test at $t = 0$ is cost effective, we are interested in testing the following null hypothesis $\mathsf{H}_{0}: \text{VoI} \leq \nu$ with $\nu = 0$ or $\mathsf{cH}_{0} (L_{0} = l_{0}): \text{cVoI} (l_{0}) \leq \nu$ with $\nu = 0$, $l_{0} = 0, 1$, depending on whether we want a cost effectiveness analysis on the population level or on the subgroup level given $L_{0}$. For each simulated dataset, we approximate the distribution of the estimated VoI or $\text{cVoI}(l_{0})$ by bootstrap sampling 100 times. In the first two columns of Table \ref{tab:dgp2cost}, we report the Monte Carlo means and Monte Carlo standard deviations of (1) the bootstrapped averages of the estimated VoIs and (2) the bootstrapped standard deviations of the estimated VoIs, marginally (upper panel) and conditional on $L_{0} = 1$ (middle panel) and $L_{0} = 0$ (lower panel). The two summary statistics based on nonparametric bootstrap are very close to the ``truths'' reported in the third column of Table \ref{tab:dgp2voi}. For example, at $n = 5 \times 10^{4}$, before (or after) adjusting for the NDE assumption, the bootstrapped average of the estimated marginal VoI has Monte Carlo mean 0.0759 (or 0.0632) and Monte Carlo standard deviation 0.0461 (or 0.00848) over 100 replications. Hence its Monte Carlo standard error is $0.0461 / \sqrt{100} = 0.00461$ or ($0.00848 / \sqrt{100} = 0.000848$). The corresponding Monte Carlo mean of the estimated VoI reported in the third column of Table \ref{tab:dgp2voi} is 0.0701 (or 0.0637), within the interval $0.0759 \pm 2 \times 0.00461 = [0.0667, 0.0851]$ (or $0.0632 \pm 2 \times 0.000848 = [0.0615, 0.0649]$).

In each of the 100 replications, if 0 is below the 5\% quantile of the empirical distribution of the 100 bootstrapped VoI estimators, we reject $\mathsf{H}_{0}$ and conclude that the screening test is marginally cost effective. The third column of Table \ref{tab:dgp2cost} reports empirical rejection rate of the null hypothesis $\mathsf{H}_{0}$ over 100 replications: adjusting for NDE assumption improves the empirical rejection rate from 25\% (or 40\%) to 100\% (or 100\%) at $n = 2.5 \times 10^{4}$ (or at $n = 5 \times 10^{4}$). We can perform similar analyses for conditional VoIs. The empirical rejection rate of $\mathsf{cH}_{0}(L_{0} = 1)$ is 0\% regardless of the sample size or whether adjusting for the NDE assumption. As we have seen previously, in the strata $L_{0} = 1$, it is optimal to always treat without screening even allowing screening. Hence screening is not beneficial for subjects whose $L_{0}$ is 1, in terms of VoI. On the other hand, the empirical rejection rate of $\mathsf{cH}_{0}(L_{0} = 0)$ is improved from 22\% (or 40\%) to 100\% (or 100\%) at $n = 2.5 \times 10^{4}$ (or at $n = 5 \times 10^{4}$) after adjusting for the NDE assumption.

\begin{landscape}
\begin{table}[h]
\centering
\begin{tabular}{c|c|c|c}
\hline
& bootstrap mean & bootstrap standard deviation & rejection rate \\
\hline
w/o NDE ($n = 2.5 \times 10^{4}$) & 0.0865 (0.0608) & 0.0600 (0.0176) & 25\%  \\
w/ NDE ($n = 2.5 \times 10^{4}$) & 0.0623 (0.0107) & 0.0122 (0.000889) & 100\% \\
w/o NDE ($n = 5 \times 10^{4}$) & 0.0759 (0.0461) & 0.0454 (0.0109) & 40\% \\
w/ NDE ($n = 5 \times 10^{4}$) & 0.0632 (0.00848) & 0.00865 (0.000716) & 100\% \\
\hline
w/o NDE ($n = 2.5 \times 10^{4}$) $\vert L_{0} = 1$ & 0.00426 (0.0141) & 0.00895 (0.0143) & 0\%  \\
w/ NDE ($n = 2.5 \times 10^{4}$) $\vert L_{0} = 1$ & $8.625 \times 10^{-5}$ ($5.189 \times 10^{-4}$) & $3.880 \times 10^{-4}$ ($1.535 \times 10^{-3}$) & 0\% \\
w/o NDE ($n = 5 \times 10^{4}$) $\vert L_{0} = 1$ & 0.000798 (0.00328) & 0.00257 (0.00576) & 0\% \\
w/ NDE ($n = 5 \times 10^{4}$) $\vert L_{0} = 1$ & $1.424 \times 10^{-5}$ ($1.424 \times 10^{-6}$) & $1.424 \times 10^{-6}$ ($1.424 \times 10^{-5}$) & 0\% \\
\hline
w/o NDE ($n = 2.5 \times 10^{4}$) $\vert L_{0} = 0$ & 0.172 (0.123) & 0.121 (0.0358) & 22\%  \\
w/ NDE ($n = 2.5 \times 10^{4}$) $\vert L_{0} = 0$ & 0.127 (0.0217) & 0.0249 (0.00175) & 100\% \\
w/o NDE ($n = 5 \times 10^{4}$) $\vert L_{0} = 0$ & 0.154 (0.0944) & 0.0923 (0.0224) & 40\% \\
w/ NDE ($n = 5 \times 10^{4}$) $\vert L_{0} = 0$ & 0.129 (0.0173) & 0.0176 (0.00146) & 100\% \\
\hline
\end{tabular}
\caption{Cost-effective analysis for DGP2. The first and second columns display the Monte Carlo means and standard deviations of the bootstrapped means and standard deviations of the estimated VoIs. The third column reports the empirical rejection rate of the null hypothesis $\mathsf{H}_{0}$ that the screening test is not cost effectiveness. Upper panel: marginal versions; middle panel: conditional version given $L_{0} = 1$; lower panel: conditional version given $L_{0} = 0$.} \label{tab:dgp2cost}
\end{table}
\end{landscape}

In summary, the above simulation demonstrates that adjusting for the NDE assumption can (1) improve the probability of selecting the optimal regimes in the opt-SNMM estimation procedure, (2) improve efficiency of the estimated value functions and VoIs, and thus (3) be a potentially valuable tool for cost-benefit evaluations of some expensive clinical tests. It should be noted that (1) may be a reason for the efficiency gain but efficiency can still be improved without increasing the probability of selecting the optimal regimes.

\subsection{opt-SNMM estimators in Appendix \ref{sim:cost}}
\label{app:sim} 

In our simulations, since all random variables but $Y$ are $\{0, 1\}$-valued ($R_1$ has an extra missingness category; see Remark \ref{rem:missing}), we can estimate the population (conditional) expectations by computing the empirical proportions of each strata of the binary random variables being conditioned on so that we can focus on the issue of adjusting for the NDE assumption without worrying about model misspecification.

In terms of adjusting for the NDE assumption, since there is a single screening occasion at $t = 0$, we only need to find one optimal function $b \left( \bar{H}_{0}, \underaccent{\bar}{S}_{0}, Y^{d} \right) \equiv b_{0} \left( \bar{H}_{0}, \underaccent{\bar}{S}_{0}, Y^{d} \right)$ using the techniques developed in Section \ref{sec:optimal_dr}. For simplicity, we consider the following class of $b \left( \bar{H}_{0}, \underaccent{\bar}{S}_{0}, Y^{d} \right)$'s: 
\begin{align*}
b \left( \bar{H}_{0}, \underaccent{\bar}{S}_{0}, Y^{d} \right) \equiv b \left( L_{0}, S_{1}, Y^{d} \right) = \bm{\beta}_{0, 4\xi}^{\top} \underbrace{\left( \begin{array}{c}
L_{0} S_{1} \bm{\varphi} (Y^d) \\ 
L_{0} (1 - S_{1}) \bm{\varphi} (Y^d) \\ 
(1 - L_{0}) S_{1} \bm{\varphi} (Y^d) \\ 
(1 - L_{0}) (1 - S_{1}) \bm{\varphi} (Y^d)
\end{array} \right)_{4\xi \times 1}}_{\coloneqq \ \breve{\bm{\varphi}} (L_0, S_1, Y^d)}
\end{align*}
where $\bm{\varphi} (\cdot): \mathbb{R} \rightarrow \mathbb{R}^{\xi}$ is an $\xi$-dimensional vector of (basis) functions, e.g. natural splines, wavelets, etc. In particular, $\bm{\varphi}$ denotes the natural spline transformations with $\xi$ degrees of freedom. We choose $\xi = 6$ in the simulation. Then finding the optimal $b_{sub} \left( \bar{H}_{0}, \underaccent{\bar}{S}_{0}, Y^{d} \right)$ is equivalent to finding the optimal coefficients $\bm{\beta}_{0, 8\xi}$.

\begin{remark}
\label{rem:connection_to_sub} 
In particular, the form of $b_{0} \left( \bar{H}_{0}, \underaccent{\bar}{S}_{0}, Y^{d} \right)$ given above can be represented as a special case of equation \eqref{eq:sublinear2} in Corollary \ref{cor:sub}, with 
\begin{align*}
d_{0}^{\ast} (L_{0}) = \bm{\beta}_{0, 4\xi}^{\top} \circ \left( \underbrace{L_0, \ldots, L_0}_{2\xi \text{ copies}}, \underbrace{1 - L_0, \ldots, 1 - L_0}_{2\xi \text{ copies}} \right)^{\top}
\end{align*}
where $\circ$ is the matrix Hadamard product and 
\begin{align*}
b_{0}^{\ast} \left( \bar{H}_{0}, \underaccent{\bar}{S}_{0}, Y^{d} \right) \equiv b_{0}^{\ast} \left( L_{0}, S_{1}, Y^{d} \right) = \left( \begin{array}{c}
S_{1} \bm{\varphi} (Y^d) \\ 
(1 - S_{1}) \bm{\varphi} (Y^d)
\end{array} \right)_{2\xi \times 1}
\end{align*}
is a special case of equation \eqref{eq:sublinear1}.
\end{remark}

The corresponding $T_{b, 0}$ is then of the following form: 
\begin{align*}
T_{b, 0} = \bm{\beta}_{0, 4\xi}^{\top} \underbrace{\left\{ \begin{array}{c}
\frac{\breve{\bm{\varphi}} ( L_0, S_1, Y^d )}{p_1 ( S_1 \vert \bar{H}_1 )} - \left( \E \left[ \frac{\breve{\bm{\varphi}} ( L_0, S_1, Y^d )}{p_1 (S_1 \vert \bar{H}_1 )} \vert L_{0}, A_{0} \right] - \E \left[ \frac{\breve{\bm{\varphi}} ( L_0, S_1, Y^d )}{p_1 (S_1 \vert \bar{H}_1)} \vert L_{0} \right] \right) \\ 
- \ \left( \E \left[ \frac{\breve{\bm{\varphi}} ( L_0, S_1, Y^d )}{p_1 ( S_1 \vert \bar{H}_1 )} \vert \bar{H}_{1}, S_{1} \right] - \E \left[ \frac{\breve{\bm{\varphi}} ( L_0, S_1, Y^d ))}{p_1 (S_1 \vert \bar{H}_1)} \vert \bar{H}_{1} \right] \right)
\end{array} \right\} (A_0 - \Pi_0)}_{\coloneqq \ \breve{T}_{b, 0}}.
\end{align*}

For any random variable $U$, the optimal $\bm{\beta}_{0, 8\xi}^{\ast, U}$ corresponding to the projection of $U$ onto $T_{b, 0}$ is simply 
\begin{align}  \label{eq:beta}
\bm{\beta}_{0, 8\xi}^{\ast, U} = \left\{ \E \left[ \breve{T}_{b, 0} \breve{T}_{b, 0}^{\top} \right] \right\}^{-1} \E \left[ U \breve{T}_{b, 0} \right],
\end{align}
where $\breve{T}_{b, 0}$ is defined in the above display.

We now describe how to estimate the opt-SNMM parameters for the simulation setup described in Section \ref{sec:simulation} and Appendix \ref{sim:further} and \ref{sim:cost}. The opt-SNMM parameters are $\Psi_{0} = \{ \Psi_{0, L_{0} = l_0}^{\top}, l_0 = 0, 1 \}^{\top}$ and $\Psi_{1} = \{ \Psi_{1, \bar{H}_{1} = \bar{h}_{1}}^{\top}, \bar{h}_{1} \in \bar{\mathcal{H}}_{1} \}^{\top}$ where $\bar{\mathcal{H}}_{1}$ is the sample space of $\bar{H}_{1}$. The estimating equations $U_0$ and $U_1$ (see equation \eqref{eq:U}) without adjusting for the NDE assumption are of the following forms: 
\begin{align*}
U_{1} = & \ \left( U_{1, \bar{H}_{1} = \bar{h}_{1}}, \bar{h}_{1} \in \bar{\mathcal{H}}_{1} \right) = \left( \mathcal{E}_{\bar{H}_{1} = \bar{h}_{1}} \left[ Y - \Psi_{1, \bar{H}_{1} = \bar{h}_{1}} S_{1} \right] \mathcal{E}_{\bar{H}_{1} = \bar{h}_{1}} \left[ S_{1} \right], \bar{h}_{1} \in \bar{\mathcal{H}}_{1} \right) \\
U_{0} = & \ \left( U_{0, L_{0} = l_0}, l_0 = 0, 1 \right) = \left( \mathcal{E}_{L_0 = l_0} \left[ Y_{\mathrm{blip}, 1} \left( \Psi_{1} \right) - \Psi_{0, L_{0} = 0}^{\top} A_{0} \right] \mathcal{E}_{L_0 = l_0} \left[ A_{0} \right], l_0 = 0, 1 \right)
\end{align*}
where $\mathcal{E}_{Z = z}(\cdot) \coloneqq id (\cdot) - \E [ \cdot \vert Z = z ]$, $\bar{H}_{1} = \left( L_0, A_0, R_1 \right)$, $Y_{\mathrm{blip}, 1} \left( \tilde{\Psi}_{1, \bar{H}_{1}} \right) \coloneqq Y - \tilde{\Psi}_{1, \bar{H}_{1}}^{\top} S_{1} + \tilde{\Psi}_{1, \bar{H}_{1}}^{\top} S_{1}^{opt}$. So the corresponding $\bm{q} \equiv (q_{t} = Q_{t}, t = 0, 1)$ in the definition of $U_t$ (see equation \eqref{eq:U}) are 
\begin{align*}
Q_1 (S_{1}) & = \left( \mathbbm{1} \{ \bar{H}_{1} = \bar{h}_{1} \} S_{1}, \bar{h}_{1} \in \bar{\mathcal{H}}_{1} \right), \\
Q_0 (A_{0}) & = \left( \mathbbm{1} \{ L_{0} = l_{0} \} A_{0}, l_{0} = 0, 1 \right).
\end{align*}
Now at the first occasion ($t = 0$), the total number of parameters in $\Psi_{0}$ is $2$ (counting $L_0 = 0$ or $L_0 = 1$), whereas at the second occasion ($t = 1$), the total number of parameters in $\Psi_{1}$ is $6$ (counting all possible combinations of $L_0$, $A_0$ and $R_1$). Then the unadjusted estimator $\tilde{\underaccent{\bar}{\Psi}}_0 (\bm{q})$ can be computed by solving the above system of linear equations: 
\begin{align*}
\tilde{\Psi}_{1, \bar{H}_{1} = \bar{h}_{1}} (\bm{q}) = & \; \mathbb{P}_n \left[ \hat{\mathcal{E}}_{\bar{H}_{1} = \bar{h}_{1}} \left[ S_{1} \right]^{2} \mathbbm{1} \{ \bar{H}_{1} = \bar{h}_{1} \} \right]^{-1} \mathbb{P}_n \left[ Y \hat{\mathcal{E}}_{\bar{H}_{1} = \bar{h}_{1}} \left[ S_{1} \right] \mathbbm{1} \{ \bar{H}_{1} = \bar{h}_{1} \} \right] \\
\tilde{\Psi}_{0, L_{0} = l_0} (\bm{q}) = & \; \mathbb{P}_n \left[ \hat{\mathcal{E}}_{L_0 = l_0} \left[ A_{0} \right]^{2} \mathbbm{1} \{ L_0 = l_0 \} \right]^{-1} \mathbb{P}_n \left[ Y_{\mathrm{blip}, 1} \left( \tilde{\Psi}_{1, \bar{H}_{1}} (q) \right) \hat{\mathcal{E}}_{L_0 = l_0} \left[ A_0 \right] \mathbbm{1} \{ L_0 = l_0 \} \right],
\end{align*}
where $\hat{\mathcal{E}}_{Z = z} = id (\cdot) - \hat{\E} [\cdot \vert Z = z]$ with conditional expectation replaced by its corresponding estimator (in our case the empirical mean in each strata of $Z$).

In the next step, plugging the unadjusted estimator $\tilde{\underaccent{\bar}{\Psi}}_0 (q)$ into $U_0$ and $U_1$, we obtain the corresponding $\hat{U}_{0}$ and $\hat{U}_{1}$. In turn, we can replace $U$ in equation \eqref{eq:beta} by $\hat{U}_{t, \cdot}$ (in each dimension) to obtain $\hat{\bm{\beta}}_{0, 4\xi}^{\ast, \hat{U}_{t, \cdot}}$, with all the population expectations replaced by empirical means. Define $\hat{T}_{b^{\ast}, 0}^{\hat{U}_{t, \cdot}} = \left( \hat{\bm{\beta}}_{0, 4\xi}^{\ast, \hat{U}_{t, \cdot}} \right)^{\top} \breve{T}_{b, 0}$, which is the estimated version of $T_{b^{\ast}, 0}^{U_{t, \cdot}}$.

Finally, we need to obtain the adjusted estimator $\tilde{\underaccent{\bar}{\Psi}}_{0} (\bm{q}, b_{sub})$ by plugging in the estimated $\hat{T}_{b^{\ast}, 0}^{\hat{U}_{t, \cdot}}$ computed from the procedure described above, which encodes the NDE assumption described in Section \ref{sec:nde}. The existence of a non-iterative closed-form solution has been discussed in Remark \ref{rem:close} and the solution is of the following form: 
\begin{align*}
\tilde{\Psi}_{1, \bar{H}_{1} = \bar{h}_{1}} (\bm{q}, b_{sub}) = & \ \BP_{n} \left[ \hat{\mathcal{E}}_{\bar{H}_{1} = \bar{h}_{1}} \left[ S_{1} \right]^{2} \mathbbm{1} \{ \bar{H}_{1} = \bar{h}_{1} \} \right]^{-1} \mathbb{P}_n \left[ \left\{ Y \hat{\mathcal{E}}_{\bar{H}_{1} = \bar{h}_{1}} \left[ S_{1} \right] - \hat{T}_{b^{\ast}, 0}^{\hat{U}_{1, \bar{H}_{1} = \bar{h}_{1}}} \right\} \mathbbm{1} \{ \bar{H}_{1} = \bar{h}_{1} \} \right] \\
\tilde{\Psi}_{0, L_{0} = l_{0}} (\bm{q}, b_{sub}) = & \ \mathbb{P}_n \left[ \hat{\mathcal{E}}_{L_0 = l_0} \left[ A_{0} \right]^{2} \mathbbm{1} \{ L_0 = l_0 \} \right]^{-1} \\
& \times \mathbb{P}_n \left[ \left\{ Y_{\mathrm{blip}, 1} \left( \tilde{\Psi}_{1, \bar{H}_{1}} (q) \right) \hat{\mathcal{E}}_{L_0 = l_0} \left[ A_{0} \right] - \hat{T}_{b^{\ast}, 0}^{\hat{U}_{0, L_{0} = l_{0}}} \right\} \mathbbm{1} \{ L_0 = l_0 \} \right].
\end{align*}

\subsection{The last example in Section \ref{sec:simulation}}
\label{app:worse}

We consider the following data generating process (abbreviated as DGP3):

\begin{itemize}
\item $Z^{\prime}_{0} \sim \text{Bernoulli}(0.5)$ and $Z_{0} = 2 Z_{0}^{\prime} - 1$.

\item $A_{0} \sim \text{Bernoulli}(\rho)$, with $\rho = 0.5$.

\item If $A_{0} = 1$, the test result $R_{1} = Z_{0}$; otherwise, $R_{1} = ?$.

\item The treatment $S_{1}$ is assigned by the following law: 
\begin{align*}
S_{1} & \sim \text{Bernoulli} (\epsilon_{?}) \mathbbm{1} \{A_{0} = 0\} \\
& + \text{Bernoulli}(\epsilon^{\prime }) \mathbbm{1} \{A_{0} = 1, R_{1} = -1\} + \text{Bernoulli}(\epsilon) \mathbbm{1} \{A_{0} = 1, R_{1} = 1\}
\end{align*}

\item $Y^{d} = 10 Z_{0} + N(0, 1)$ and we assume the cost $c^{\ast} = 0$ (i.e. $Y \equiv Y^{d}$).
\end{itemize}

With $g = (a_{0} = 0, s_{1} = 1)$, we are interested in $\theta_{g} \coloneqq \E [Y_{g}] (\equiv 0)$, which is estimated by the following estimators, depending on whether we impose the NDE assumption: 
\begin{align*}
\hat{\theta}_{ipw, g} & = \sum_{i = 1}^{n} V_{ipw, i} \equiv \frac{1}{n} \sum_{i = 1}^{n} \frac{(1 - A_{0, i}) S_{i} Y_{i}}{\Pr (A_{0} = 0) \Pr (S_{1} = 1 | A_{0, i}, R_{1, i})} \\
\hat{\theta}_{nde\mbox{-}ipw, g} & = \frac{1}{n} \sum_{i = 1}^{n} V_{nde\mbox{-}ipw, i} \equiv \frac{1}{n} \sum_{i = 1}^{n} \frac{S_{i} Y_{i}}{\Pr (S_{1} = 1 | A_{0, i}, R_{1, i})}.
\end{align*}
For simplicity, we only consider the oracle estimators and assume that the oracle probabilities of testing and treatment are known. In particular, we deliberately design the simulation setup such that the projections of $V_{ipw, i}$ or $V_{nde\mbox{-}ipw, i}$ onto the scores of the treatment $S_{1}$ and the test $A_{0}$ are zero.

The variances of the above two estimators are: 
\begin{align*}
n \var (\hat{\theta}_{ipw, g}) & = \var [V_{ipw, 1}] = \var \left[ \frac{(1 - A_{0}) S_{1} Y}{\Pr (A_{0} = 0) \Pr (S_{1} = 1 | A_{0} = 0, R_{1})} \right] \\
& = \E \left[ \frac{(1 - A_{0}) S_{1} Y^{2}}{\Pr (A_{0} = 0)^{2} \Pr (S_{1} = 1 | A_{0} = 0, R_{1})^{2}} \right] - \theta_{g}^{2} \\
& = \E \left[ \frac{101}{\Pr (A_{0} = 0) \Pr (S_{1} = 1 | A_{0} = 0, R_{1} = ?)} \right] - 0 \\
& = \frac{101}{\rho \epsilon_{?}} \\
n \var (\hat{\theta}_{nde\mbox{-}ipw, g}) & = \var (V_{nde\mbox{-}ipw, 1}) = \var \left[ \frac{S_{1} Y}{\Pr (S_{1} = 1 | A_{0}, R_{1})} \right] = \E \left[ \frac{S_{1} Y^{2}}{\Pr (S_{1} = 1 | A_{0}, R_{1})^{2}} \right] - \theta_{g}^{2} \\
& = \frac{101 \Pr (A_{0} = 0)}{\Pr (S_{1} = 1 | A_{0} = 0)} + \frac{101 \Pr (Z_{0} = 1, A_{0} = 1)}{\Pr (S_{1} = 1 | A_{0} = 1, R_{1} = 1)} + \frac{101 \Pr (Z_{0} = -1, A_{0} = 1)}{\Pr (S_{1} = 1 | A_{0} = 1, R_{1} = -1)} \\
& = 101 \rho \left( \frac{1}{\epsilon_{?}} + \frac{1}{2 \epsilon} + \frac{1}{2 \epsilon^{\prime}} \right).
\end{align*}
Thus the ratio of these two variances is 
\begin{align*}
\frac{\var (\hat{\theta}_{nde\mbox{-}ipw, g})}{\var (\hat{\theta}_{ipw, g})} = \rho^{2} \left( 1 + \frac{\epsilon_{?}}{2 \epsilon} + \frac{\epsilon_{?}}{2 \epsilon^{\prime}} \right)
\end{align*}
which can be made arbitrarily large by letting either of the two ratios $\frac{\epsilon_{?}}{\epsilon}$ or $\frac{\epsilon_{?}}{\epsilon^{\prime}} \rightarrow \infty$ (i.e. $\epsilon$ or $\epsilon^{\prime} \rightarrow 0$).

\begin{remark}
The reason for this counterintuitive phenomenon is the following. When $\epsilon$ (or $\epsilon^{\prime}$) is very small, if $A_{0} = 1$, the treatment is assigned with very low probability when the test result is +1 (or -1), which leads to inflated variance compared to $\hat{\theta}_{ipw, g}$ if we do not censor the patients incompatible with the testing regime of interest (that is, $a_{0} = 0$). Here we can consider an even more simplified example where there is no test result involved: $S_{1} = \text{Bernoulli} (\epsilon_{?}) \mathbbm{1} \{A_{0} = 0\} + \text{Bernoulli} (\epsilon) \mathbbm{1} \{A_{0} = 1\}$. In this scenario, when $\epsilon$ is small, we will again have inflated variance compared to $\hat{\theta}_{ipw, g}$ if the patients with $A_{0} = 1$ is not censored because the inverse probability of $A_{0} = 1$ is extremely large.
\end{remark}


\end{document}